\documentclass[a4paper,UKenglish,cleveref, autoref]{lipics-v2019}
\newcommand{\proofsubparagraph}[1]{\vspace{-1mm}\subparagraph*{#1}}

\usepackage{etex}
\usepackage[utf8]{inputenc}
\usepackage{mathtools}

\usepackage{amsmath,amssymb,tabularx,hhline,rotating,enumerate,xspace}
\usepackage{amsthm} 
\usepackage{tikz}

\usepackage{multirow}

\usepackage{pgfplots}
\usepackage[disable]{todonotes}

\usepackage{longtable}

\usepackage[capitalise]{cleveref}

\usepackage{adjustbox}
\usepackage{xparse}

\makeatletter
\NewDocumentCommand{\raisedminus}{m}{%
	\raisebox{0.1em}{$\m@th#1{-}$}%
}
\NewDocumentCommand{\unaryminus}{}{%
	\mathbin{%
		\mathchoice{%
			\raisedminus\scriptstyle
		}{%
			\raisedminus\scriptstyle
		}{%
			\raisedminus\scriptscriptstyle
		}{%
			\raisedminus\scriptscriptstyle
		}%
	}%
}
\makeatother

\makeatletter
\newcommand\footnoteref[1]{\protected@xdef\@thefnmark{\ref{#1}}\@footnotemark}
\makeatother

\newcommand{\OLDcompproblem}[3][]{%
	\par\vspace{0.125cm plus 0.1cm minus 0.05cm}\begin{tabularx}{\linewidth-2\parindent}{lX}%
		\if\relax\detokenize{#1}\relax%
		\else%
		\textnormal{\textsf{Constant:}}&#1\\%
		\fi%
		\textnormal{\textsf{Input:}}&#2\\%
		\textnormal{$\mathrlap{\textsf{Output:}}\hphantom{\textsf{Question:}}$}&#3\\%
	\end{tabularx}\vspace{0.125cm plus 0.1cm minus 0.05cm}\par%
}

\newcommand{\OLDproblem}[3][]{%
	\par\vspace{0.125cm plus 0.1cm minus 0.05cm}\begin{tabularx}{\linewidth-2\parindent}{lX}%
		\if\relax\detokenize{#1}\relax%
		\else%
		\textnormal{\textsf{Constant:}}&#1\\%
		\fi%
		\textnormal{\textsf{Input:}}&#2\\%
		\textnormal{\textsf{Question:}}&#3\\%
	\end{tabularx}\vspace{0.125cm plus 0.1cm minus 0.05cm}\par%
}

\newcommand{\compproblem}[3][]{%
	\par\vspace{0.125cm plus 0.1cm minus 0.05cm}\adjustbox{valign=t}{\begin{tabularx}{\linewidth-2\parindent}{@{}lX}%
			\if\relax\detokenize{#1}\relax%
			\else%
			\textnormal{\textsf{Constant:}}&#1\\%
			\fi%
			\textnormal{\textsf{Input:}}&#2\\%
			\textnormal{$\mathrlap{\textsf{Output:}}\hphantom{\textsf{Question:}}$\!\!}&#3\\%
	\end{tabularx}}\vspace{0.125cm plus 0.1cm minus 0.05cm}\par%
}

\newcommand{\problem}[3][]{%
	\par\vspace{0.125cm plus 0.1cm minus 0.05cm}\adjustbox{valign=t}{\begin{tabularx}{\linewidth-2\parindent}{@{}lX}%
			\if\relax\detokenize{#1}\relax%
			\else%
			\textnormal{\textsf{Constant:}}&#1\\%
			\fi%
			\textnormal{\textsf{Input:}}&#2\\%
			\textnormal{\textsf{Question:}\!\!}&#3\\%
	\end{tabularx}}\vspace{0.125cm plus 0.1cm minus 0.05cm}\par%
}

\newcommand\Wlog{W.\,l.\,o.\,g.\ }

\newenvironment{vd}{\noindent\color{blue} VD :  }{}

\newenvironment{AM}{\noindent\color{red} AM : }{}

\newenvironment{aw}{\noindent\color{magenta} AW :  }{}

\usepackage{amsthm}

\newtheorem{fact}[theorem]{Fact}

\definecolor{lred}{rgb}{1.0,0.651,0.651}

\newcommand{\Min}{\operatorname{Min}}

\newcommand{\binM}[2]{\operatorname{digit}_{#2}(#1)}
\newcommand{\csdr}{compact signed-digit representation\xspace}
\newcommand{\sdr}{signed-digit representation\xspace}
 \newcommand{\stepone}{\proc{UpdateNodes}}
\newcommand{\steptwo}{\proc{ExtendChains}}
\newcommand{\stepthree}{\proc{UpdateMarkings}}

\newcommand{\startc}[1]{\operatorname{start}(#1)}

\newcommand{\equival}{\sim_{\epsilon}}

\newcommand{\set}[2]{\left\{#1\mathrel{\left|\vphantom{#1}\vphantom{#2}\right.}#2\right\}}

\newcommand{\oneset}[1]{\left\{\mathinner{#1}\right\}}

\newcommand{\abs}[1]{\left\lvert\mathinner{#1}\right\rvert}

\newcommand{\abssmall}[1]{\lvert\mathinner{#1}\rvert}
\newcommand{\absbig}[1]{\bigl\lvert\mathinner{#1}\bigr\rvert}

\newcommand{\floor}[1]{\left\lfloor\mathinner{#1} \right\rfloor}
\newcommand{\ceil}[1]{\left\lceil\mathinner{#1} \right\rceil}

\newcommand{\Gen}[2]{\left< \mathinner{#1} \mid \mathinner{#2}\right>}

\newcommand{\N}{\ensuremath{\mathbb{N}}}
\newcommand{\Z}{\ensuremath{\mathbb{Z}}}
\newcommand{\Q}{\ensuremath{\mathbb{Q}}}
\newcommand{\R}{\ensuremath{\mathbb{R}}}

\newcommand{\TC}{\ensuremath{\mathsf{TC}^0}\xspace}
\newcommand{\NC}{\ensuremath{\mathsf{NC}}\xspace}

\newcommand{\sdZ}{\ensuremath{\Z[1/2] \rtimes \Z}}

\renewcommand{\phi}{\varphi}
\newcommand{\eps}{\varepsilon}
\renewcommand{\epsilon}{\varepsilon}
\newcommand{\depth}{\operatorname{depth}}
\newcommand{\clone}{\proc{Clone}}

\newcommand{\val}{\operatorname{val}}

\newcommand{\cor}{\operatorname{CR}}

\newcommand{\e}{\eps} 

\newcommand{\mOne}{\xspace{\unaryminus\!1}}
\newcommand{\supp}{\sigma}

\newcommand{\del}{\delta}
\newcommand{\lam}{\lambda}

\newcommand{\Sig}{\Sigma}

\newcommand\GG{\Gamma}
\newcommand\LL{\Lambda}

\newcommand{\Oh}{\mathcal{O}}

\newcommand{\cB}{\mathcal{B}}

\newcommand{\cC}{\mathcal{C}}

\newcommand{\cX}{\mathcal{X}}

\newcommand{\Breduced}{Britton-re\-du\-ced\xspace}

\newcommand{\BS}[2]{\ensuremath{\mathrm{\mathbf{BS}}_{#1,#2}}}
\newcommand{\BG}{\ensuremath{\mathrm{\mathbf{G}}_{1,2}}\xspace} 

\newcommand{\oi}[1]{{#1}^{-1}}

\newcommand{\DSPACE}{\ensuremath{\mathsf{DSPACE}}\xspace} 

\newcommand{\PSPACE}{\ensuremath{\mathsf{PSPACE}}\xspace} 

\newcommand{\LOGCFL}{\ensuremath{\mathsf{LOGCFL}}\xspace} %
\newcommand{\LOGSPACE}{\ensuremath{\mathsf{LOGSPACE}}\xspace} %
\newcommand{\NL}{\ensuremath{\mathsf{NL}}\xspace} %
\newcommand{\Tc}[1]{\ensuremath{\mathsf{TC}^{#1}}\xspace}
\newcommand{\Ac}[1]{\ensuremath{\mathsf{AC}^{#1}}\xspace}
\newcommand{\Nc}[1]{\ensuremath{\mathsf{NC}^{#1}}\xspace}

\renewcommand{\P}{\ensuremath{\mathsf{P}}\xspace}

\newcommand{\GapL}{\ensuremath{\mathsf{GapL}}\xspace} %
\newcommand{\DLOGTIME}{\ensuremath{\mathsf{DLOGTIME}}\xspace} %

\newcommand{\PTc}[2]{\ensuremath{\mathsf{DepParaTC}^{#1}}\xspace}

\newcommand{\smallwidetilde}[1]
{{\mspace{1mu}\widetilde{\mspace{-1mu}#1\mspace{-1mu}}\mspace{1mu}}}

\newcommand{\wt}[1]{\smallwidetilde{#1}}

\newcommand\tow{\mathop \tau} 

\newcommand{\PC}{power circuit\xspace}

\newcommand\lds{,\ldots ,} 
\newcommand{\sse}{\subseteq}
\newcommand{\es}{\emptyset}

\newcommand{\OS}{\oneset{\mOne,0,+1}}

\newcommand{\interval}[2]{[ \mathinner{#1}..\mathinner{#2}] }

\newcommand{\oneinterval}[1]{[ \mathinner{#1}] }

\newcommand\ei[1]{{\emph{#1}\xspace}\index{#1}} 

\newcommand{\proc}[1]{\ensuremath{\text{\textsc{#1}}}\xspace}

\newcommand{\gnot}{\proc{Not}}
\newcommand{\gand}{\proc{And}}
\newcommand{\gor}{\proc{Or}}
\newcommand{\gmaj}{\proc{Majority}}

\newcommand\ie{i.e.,\xspace}
\newcommand\eg{e.g.\xspace}

\newcommand{\comparator}{\mathrel{\triangle}}
\newcommand{\compOpSet}{\oneset{ = , \neq, <,\leq ,>, \geq }}

\newcommand{\eqBG}{=_{\BG}}
\newcommand{\neqBG}{\neq_{\BG}}

\usepackage{graphicx}

\usepackage{latexsym}

\title{Parallel algorithms for power circuits and the word problem of the Baumslag group}

\titlerunning{Parallel algorithms for the Baumslag group} 

\author{Caroline Mattes}{Universität Stuttgart, Institut für Formale Methoden der Informatik (FMI), Universitätsstraße 38,
  70569 Stuttgart, Germany}{caroline.mattes@fmi.uni-stuttgart.de}{}{}

\author{Armin Weiß}{Universität Stuttgart, Institut für Formale Methoden der Informatik (FMI), Universitätsstraße 38,
  70569 Stuttgart, Germany}{armin.weiss@fmi.uni-stuttgart.de}{https://orcid.org/0000-0002-7645-5867}{Funded by DFG project DI 435/7-1.}

\authorrunning{C. Mattes and A. Weiß} 

\Copyright{Caroline Mattes and Armin Weiß} 



\category{} 

\relatedversion{conference version \cite{MattesW21} at MFCS 2021}

\nolinenumbers 

\hideLIPIcs  

\EventEditors{Filippo Bonchi and Simon J. Puglisi}
\EventNoEds{2}
\EventLongTitle{46th International Symposium on Mathematical Foundations of Computer Science (MFCS 2021)}
\EventShortTitle{MFCS 2021}
\EventAcronym{MFCS}
\EventYear{2021}
\EventDate{August 23--27, 2021}
\EventLocation{Tallinn, Estonia}
\EventLogo{}
\SeriesVolume{202}
\ArticleNo{13}

\begin{document}

\maketitle

\begin{abstract}
Power circuits have been introduced in 2012 by Myasnikov, Ushakov and Won as a data structure for non-elementarily compressed integers supporting the arithmetic operations addition and $(x,y) \mapsto x\cdot 2^y$. The same authors applied power circuits to give a polynomial-time solution to the word problem of the Baumslag group, which has a non-elementary Dehn function.

In this work, we examine power circuits and the word problem of the Baumslag group under parallel complexity aspects. 
In particular, we establish that the word problem of the Baumslag group can be solved in \NC~-- even though one of the essential steps is to compare two integers given by power circuits and this, in general, is shown to be \P-complete. The key observation is that the depth of the occurring power circuits is logarithmic and such power circuits can be compared in \NC. 
\end{abstract}

{\small\tableofcontents}

\section{Introduction}\label{sec:intro} 

The {\em word problem} of a finitely generated group $G$ is as follows: does a given word over the generators of $G$ represent the identity of $G$? It was first studied by Dehn as one of the basic algorithmic problems in group theory \cite{dehn11}.
Already in the 1950s, Novikov and Boone succeeded to construct finitely presented groups with an undecidable word problem \cite{boone59,nov55}. Nevertheless, many natural classes of groups have an (efficiently) decidable word problem~-- most prominently the class of linear groups (groups embeddable into a matrix group over some field): their word problem is in \LOGSPACE \cite{lz77,Sim79}~-- hence, in particular, in \NC, \ie decidable by Boolean circuits of polynomial size and polylogarithmic depth (or, equivalently decidable in polylogarithmic time using polynomially many processors).

There are various other results on word problems of groups in small parallel complexity classes defined by circuits. For example
 the word problems of solvable linear groups are even in \TC (constant depth with threshold gates) \cite{KonigL18poly} and the word problems of Baumslag-Solitar groups and of right-angled Artin groups are \Ac{0}-Turing-reducible to the word problem of a non-abelian free group \cite{Weiss16,Kausch17diss}.
Moreover, Thompson's groups are co-context-free \cite{LehnertS07} and hyperbolic groups have word problem in \LOGCFL \cite{Lo05ijfcs}. All these complexity classes are contained within \NC.
 On the other hand, there are also finitely presented groups with a decidable word problem but with arbitrarily high complexity \cite{brs02}. 

A mysterious class of groups under this point of view are one-relator groups, i.e.\ groups that can be written as a free group modulo a normal subgroup generated by a single element (\emph{relator}). 
Magnus \cite{mag32} showed that one-relator groups have a decidable word problem; his algorithm is called the Magnus breakdown procedure (see also \cite{LS01,mks}). Nevertheless, the complexity remains an open problem~-- although it is not even clear whether the word problems of one-relator groups are solvable in elementary time, in \cite{BMS} the question is raised whether they are actually decidable in polynomial time.

 In 1969 Gilbert Baumslag defined the group
	$\BG = \Gen{a,b}{bab^{-1} a = a^2bab^{-1}}$
as an example of a one-relator group which enjoys certain remarkable properties. It is infinite and non-abelian, but all its finite quotients are cyclic and, thus, it is not residually finite \cite{baumslag69}. 
Moreover, Gersten showed that the Dehn function of $\BG$ is non-elementary \cite{gersten91} and Platonov \cite{plat04} made this more precise by proving that it is (roughly) $\tau(\log n)$ where $\tow(0) = 1$ and $\tow(i+1) = 2^{\tow(i)}$ for $i \geq 0$ is the tower function (note that he calls the group Baumslag-Gersten group). 
Since the Dehn function gives an upper bound on the complexity of the word problem, the Baumslag group was a candidate for a group with a very difficult word problem. 
Indeed, when applying the Magnus breakdown procedure to an input word of length $n$, one obtains as intermediate results words of the form $v_1^{x_1} \cdots v_m^{x_m}$ where $v_i \in \{a,b, bab^{-1} \}$, $x_i \in \Z$, and $m \leq n$. The issue is that the $x_i$ might grow up to $\tau(\log n)$; hence, this algorithm has non-elementary running time. 
However, as foreseen by the above-mentioned conjecture, Myasnikov, Ushakov and Won succeeded to show that the word problem of \BG is, indeed, decidable in polynomial time \cite{muw11bg}. Their crucial contribution was to introduce so-called \emph{power circuits} in \cite{MyasnikovUW12} for compressing the $x_i$ in the description above.

Roughly speaking, a \emph{power circuit} is a directed acyclic graph (a dag) where the edges are labelled by $\pm 1$. 
One can define an evaluation of a vertex $P$ as two raised to the power of the (signed) sum of the successors of $P$. 
 Note that this way the value $\tow(n)$ of the tower function can be represented by an $n + 1$-vertex power circuit~-- thus, power circuits allow for a non-elementary compression. The crucial feature for the application to the Baumslag group is that power circuits not only efficiently support the operations $+$, $-$, and $(x,y) \mapsto x \cdot 2^y$, but also the test whether $x = y$ or $x < y$ for two integers represented by power circuits can be done in polynomial time. The main technical part of the comparison algorithm is the so-called reduction process, which computes a certain normal form for power circuits.

Based on these striking results, Diekert, Laun and Ushakov \cite{DiekertLU13ijac} improved the algorithm for power circuit reduction and managed to decrease the running time for the word problem of the Baumslag group from $\Oh(n^7)$ down to $\Oh(n^3)$. They also describe a polynomial-time algorithm for the word problem of the famous Higman group $H_4$ \cite{higman51}. In \cite{muCRAG} these algorithms have been implemented in C++. Subsequently, more applications of power circuits to these groups emerged: in \cite{Laun14} a polynomial time solution to the word problem in generalized Baumslag and Higman groups is given, in \cite{DiekertMW2016alg} the conjugacy problem of the Baumslag group is shown to be strongly generically in \P and in \cite{Baker20} the same is done for the conjugacy problem of the Higman group. Here ``generically'' roughly means that the algorithm works for most inputs (for details on the concept of generic complexity, see \cite{KMSS1}).

Other examples where compression techniques lead to efficient algorithms in group theory can be found \eg in \cite{DisonER18} or \cite[Theorems~4.6, 4.8 and~4.9]{Lohrey14compressed}.
Finally, notice that in \cite{MiasnikovN20} the word search problem for the Baumslag group has been examined using parametrized complexity.

\vspace{-1mm}
\subparagraph*{Contribution.}
The aim of this work is to analyze power circuits and the word problem of the Baumslag group under the view of parallel (circuit) complexity. For doing so, we first examine so-called \emph{compact} representations of integers and show that ordinary binary representations can be converted into compact representations by constant depth circuits (\ie in \Ac0~-- see \cref{sec:compact}). We apply this result in the power circuit reduction process, which is the main technical contribution of this paper. While \cite{MyasnikovUW12,DiekertLU13ijac} give only polynomial time algorithms, we present a more refined method and analyze it in terms of parametrized circuit complexity. The parameter here is the depth $D$ of the power circuit. More precisely, we present threshold circuits of depth $\Oh(D)$ for power circuit reduction~-- implying our first main result: 

\theoremstyle{plain}
\newtheorem{propositionA}{Proposition}
\renewcommand*{\thepropositionA}{\Alph{propositionA}}
\newtheorem{theoremA}[propositionA]{Theorem}
\begin{propositionA}\label{prop:compIntro}
	The problem of comparing two integers given by power circuits of logarithmic depth is in \Tc1 (decidable by logarithmic-depth, polynomial-size threshold circuits).
\end{propositionA}
 We then analyze the word problem of the Baumslag group carefully. 
A crucial step is to show that all appearing power circuits have logarithmic depth. Using Proposition~\ref{prop:compIntro} we succeed to describe a \Tc1 algorithm for computing the Britton reduction of $uv$ given that $u$ and $v$ are already Britton-reduced (Britton reductions are the basic step in the Magnus breakdown procedure~-- see \cref{sec:BS} for  a definition). This leads to the following result:

\begin{theoremA}\label{thm:main}
	The word problem of the Baumslag group \BG is in \Tc2.
\end{theoremA}

In the final part of the paper we prove lower bounds on comparison in power circuits, and thus, on power circuit reduction. In particular, this emphasizes the relevance of Proposition A and shows that our parametrized analysis of power circuit reduction is essentially the best one can hope for. Moreover, Theorem~\ref{thm:compIntro} highlights the importance of the logarithmic depth bound for the power circuits appearing during the proof of Theorem~\ref{thm:main}.

\begin{theoremA}\label{thm:compIntro}
	The problem of comparing two integers given by power circuits is \P-complete.
\end{theoremA}

Power circuits can be seen in the broader context of arithmetic circuits and arithmetic complexity. Thus, results on power circuits also give further insight into these arithmetic circuits. Notice that the corresponding logic over natural numbers with addition and $2^x$ has been shown to be decidable by Semënov \cite{Semenov83}. In \cref{prop:convertToArithmeticCircuit} we show that, indeed, for every power circuit with a marking $M$ there is an arithmetic circuit of polynomial size with $+$-, $-$-, and $2^x$-gates evaluating to the same number and vice-versa. Moreover, the transformation between these two models can be done efficiently.

This work is the full and extended version of the conference publication \cite{MattesW21}. Besides giving full proofs of all results, here we explore the connections between power circuits and arithmetic circuits with $+$-, $-$-, and $2^x$-gates and in \cref{thm:PCsimulation} we give a refined variant of Theorem~\ref{thm:compIntro} which also yields hardness results for power circuits of logarithmic depth.

\section{Notation and preliminaries}\label{sec:prelims}
\subparagraph*{General notions.}
We use standard $\Oh$-notation for functions from $\N$ to 
non-negative reals $\R^{\geq 0}$, see \eg \cite{CLRS09}. 
Throughout, the logarithm $\log$ is with respect to base two.
The \ei{tower function} $\tow\colon\N \to \N$
is defined by $\tow(0) = 1$ and $\tow(i+1) = 2^{\tow(i)}$ for $i \geq 0$. 
It is primitive recursive, but $\tow(6)$ written in binary cannot be stored in the memory of any conceivable real-world computer.
Moreover, we set $\log^*(n) = \min \set{i}{\tau(i) \geq n}$. 

We denote the support of a function $f\colon X  \to \R$ by $\supp(f) = \set{x \in X}{f(x) \neq 0}$.
Furthermore, the interval of integers $\oneset{i, \dots, j} \sse \Z$ is denoted by $\interval{i}{j}$ and we define $\oneinterval{n} = \interval{0}{n-1}$.
We write $\Z[1/2] = \set{p/2^q\in \Q}{p,q \in \Z}$ for the set of dyadic fractions. 

Let $\Sigma$ be a set. The set of all words over $\Sigma$ is denoted by $\Sigma^* = \bigcup_{n\in \N}\Sig^n$. 
The length of a word $w \in \Sigma^*$ is denoted by $\abs{w}$.
A dag is a directed acyclic graph. For a dag $\Gamma$ we write $\depth(\Gamma)$ for its depth, which is the length (number of edges) of a longest path in $\Gamma$.

\subsection{Complexity}

We assume the reader to be familiar with the complexity classes \LOGSPACE and \P (polynomial time); see e.g.\ \cite{AroBar09} for details. 
Most of the time, however, we use circuit complexity within \NC.

Throughout, we assume that languages $L$ (resp.\ inputs to functions $f$) are encoded over the binary alphabet $\oneset{0,1}$.
Let $k \in \N$. A language $L$ (resp.\ function $f$)  is in \Ac{k} if there is a family of polynomial-size Boolean circuits of depth $\Oh(\log^kn)$ (where $n$ is the input length) deciding $L$ (resp.\ computing $f$). 
More precisely, a Boolean circuit is a dag (directed acyclic graph) where the vertices are either input gates $x_1, \ldots, x_n$, or \gnot, \gand, or \gor gates. There are one or more designated output gates (for computing functions there is more than one output gate~-- in this case they are numbered from $1$ to $m$). All gates may have unbounded fan-in (\ie there is no bound on the number of incoming wires).
A language $L \subseteq \{0,1\}^*$ belongs to $\Ac{k}$ if there exist a family $(C_n)_{n  \in \N}$ of Boolean circuits such that $x \in L \cap \{0,1\}^n$ if and only if the output gate of $C_n$ evaluates to $1$ when assigning $x = x_1 \cdots x_n$ to the input gates.
Moreover, $C_n$ may contain at most $n^{\Oh(1)}$ gates and have depth $\Oh(\log^kn)$. Here, the depth of a circuit is the length of the longest path from an input gate to an output gate.
Likewise $\Ac{k}$-computable functions are defined. 

The class \Tc{k} is defined analogously with the difference that also \gmaj gates are allowed (a \gmaj gate outputs $1$ if its input contains more $1$s than $0$s). Moreover, $\NC = \bigcup_{k\geq 0} \Tc{k} = \bigcup_{k\geq 0} \Ac{k}$. For more details on circuits we refer to \cite{Vollmer99}. 
Our algorithms (or circuits) rely on two basic building blocks which can be done in \TC:

\begin{example}\label{ex:itAdd}
	Iterated addition is the following problem: %
\compproblem{$n$ numbers $A_1, \dots, A_n$ each having $n$ bits}{$\sum_{i = 1}^{n}A_i$}%
\noindent	This is well-known to be in \TC.
\end{example}

\begin{example}\label{ex:sortTC}
	Let $(k_1, v_1), \dots, (k_n,v_n)$ be a list of $n$ key-value pairs $(k_i, v_i)$ equipped with a total order on the keys $k_i$ such that it can be decided in \TC whether $k_i < k_j$. 
	Then the problem of sorting the list according to the keys is in \TC: the desired output is a list $(k_{\pi(1)}, v_{\pi(1)}), \dots, (k_{\pi(n)},v_{\pi(n)})$ for some permutation $\pi$ such that $k_{\pi(i)} \leq k_{\pi(j)}$  for all $i < j$.

	We briefly describe a circuit family to do so: The first layer compares all pairs of keys $k_i,k_j$ in parallel. For all $i$ and $j$ the next layer computes a Boolean value $P(i,j)$ which is true if and only if $\abs{\set{\ell}{k_\ell < k_i}} = j$. The latter is computed by iterated addition. 
	As a final step the $j$-th output pair is set to $(k_i, v_i)$ if and only if $P(i,j)$ is true.
\end{example}

\begin{remark}\label{rem:uniformAc0}
	The class $\NC$ is contained in \P if we consider uniform circuits. A family of circuits is called \LOGSPACE-uniform (or simply uniform) if the function $1^n \mapsto C_n$ is computable in \LOGSPACE (where $1^n$ is the string consisting of $n$ ones and $C_n$ is given as some reasonable encoding). Be aware that for classes below \LOGSPACE usually even stronger uniformity conditions are imposed. In order not to overload the presentation, throughout, we state all our results in the non-uniform case~-- all uniformity considerations are left to the reader.
\end{remark}

\newcommand{\para}{\operatorname{par}}
\vspace{-1mm}
\subparagraph*{Parametrized circuit complexity.} In our work we also need some parametrized version of the classes \Tc{k}, which we call \emph{depth-parametrized}  \Tc{k}. Let $\para\colon  \{0,1\}^* \to \N$ (called the \emph{parameter}).
 Consider a family of circuits $(C_{n,D})_{n , D \in \N}$ such that $C_{n,D}$ contains at most $n^{\Oh(1)}$ gates (independently of $D$)\footnote{Here and in most other natural applications the parameter $D$ is bounded by the input size $n$. In this case, we could let the size of $C_{n,D}$ be a polynomial in both $n$ and $D$~-- without changing the actual class.} and has depth $\Oh(D \cdot \log^kn)$. 
 A language $L$ is said to be accepted by this circuit family if for all $n$ and $D$ and all $x \in\{0,1\}^n$ with $\para(x) \leq D$ we have $x \in L$ if and only if $C_{n,D}$ evaluates to $1$ on input $x$.
 Similarly, $f \colon  \{0,1\}^* \to \{0,1\}^* $ is computed by  $(C_{n,D})_{n , D \in \N}$ if for all $n$ and $D$ and all $x \in\{0,1\}^n$ with $\para(x) \leq D$ the circuit $C_{n,D}$ evaluates to $f(x)$ on input $x$. 
We define \PTc{k}{D} as the class of languages (resp.\ functions) for which there are such parametrizations $\para\colon  \{0,1\}^* \to \N$ and families of circuits $(C_{n,D})_{n , D\geq 0}$. 
 Note that this is not a standard definition~-- but it perfectly fits our purposes.
 
 \begin{lemma}\label{lem:paraToTC}
 	Let $C > 0$, $k, \ell \in \N $ and $\para\colon  \{0,1\}^* \to \N$  such that $\bigl\{w \in \{0,1\}^* \mathrel{\big|} \para(w) \leq C\cdot \floor{\log \abs{w}}^{\ell} \bigr\} \in \Tc{k + \ell}$ and $L \in \PTc{k}{D}$ (parametrized by $\para$). Then $\wt L = \set{w \in L }{\para(w) \leq C \cdot \floor{\log \abs{w}}^{\ell}}$ is in $\Tc{k+\ell}$.  	
\end{lemma}

\begin{proof}
	Let $w \in \oneset{0,1}^n$ be some input. First decide whether $ \para(w) \leq \log^{\ell} n$ (by the hypothesis this is in \Tc{k + \ell}). If yes, the circuit $C_{n,  C \cdot\floor{\log n}^{\ell} }$ can be used to decide whether $w \in \wt L$. Clearly, the combined circuit has polynomial size. Its depth is $\Oh( \log^{k + \ell } n)$ plus  $\Oh(  C \cdot\floor{\log n}^{\ell} \cdot  \log^{k } n) = \Oh( \log^{k + \ell } n)$ for $C_{n,  C \cdot\floor{\log n}^{\ell} }$. Hence, we have obtained a \Tc{k + \ell} circuit.
\end{proof}  

We introduce this parametrized \Tc{k} classes because later for computing reduced power circuits we apply a non-constant number of \TC computations $f$ one after each other. The number of these computations is the depth of the power circuit. The crucial step is to show that after any number of applications of $f$, the output is still polynomially bounded. Putting things together, we obtain a \PTc{0}{} computation parametrized by the depth of the power circuit. 
Let us formalize this idea: 

Denote the $i$-fold composition of $f$ by $f^{(i)}$ (\ie $f^{(0)}$ is the identity function and $f^{(i)} = f \circ f^{(i-1)}$ for $i \geq 1$).

In order to allow circuits to compute functions having outputs of different lengths for inputs of the same length, we can assume that each output gate also carries an enable bit (or equivalently we can think that there is an additional padding symbol in the output alphabet). 

\begin{lemma}\label{ex:composeNC1}
	Let $f \colon  \{0,1\}^* \to \{0,1\}^* $ be \Tc{k}-computable such that for all $x \in \{0,1\}^*$ there is some $\omega_x \leq \abs{x}$ with $f^{(\omega_x)}(x) = f^{(\omega_x + 1)}(x)$. Further, assume that there is some polynomial $p$ such that for all $x \in \{0,1\}^*$ and for all $i\in \N$ we have $\abs{f^{(i)}(x)} \leq p(\abs{x})$.

	Then $x \mapsto f^{(\omega_x)}(x)$ is in \PTc{k}{} where the parameter $\para\colon  \{0,1\}^* \to \N$ is defined by $x \mapsto \omega_x$. 
\end{lemma}
	
\begin{proof}
	Let $(C_n)_{n\in \N}$ be the family of \Tc{k} circuits computing $f$. We construct a new family of circuits $(C_{n, \omega})_{n , \omega \in \N}$. 
	Let $\tilde C_m$ be a circuit consisting of $C_i$ for all $i \in \interval{0}{m}$ in parallel. We can compose $\tilde C_{p(n)} \circ C_n$ by feeding the outputs of $C_n$ into the $C_i$ (as part of $\tilde C_{p(n)}$) with the appropriate number of input bits. By iterating this, we obtain a circuit $\tilde C_{p(n)} \circ \cdots \circ \tilde C_{p(n)} \circ C_n$ consisting of $C_n$ followed by $\omega - 1$ layers of $\tilde C_{p(n)}$. By the hypothesis of the lemma, we can assume $\omega \leq n$, so this circuit contains at most $n \cdot (p(n))^2$ gates. Moreover, the depth of $\tilde C_{p(n)}$ is $\Oh(\log^kp(n)) = \Oh(\log^kn)$, so the depth of $C_{n, \omega}$ is  $\Oh(\omega \cdot \log^k(n))$.	
\end{proof}	

\subsection{Power circuits} \label{sec: PowerCircuits}

Consider a pair $(\GG,\del)$ where $\GG$ is a set of $n$ vertices and $\del$ is a mapping $\del\colon  \GG \times \GG\to \OS$.
The support of $\del$ is the subset $\supp(\del) \sse \GG \times \GG$ 
consisting of those $(P,Q)$ with $ \del(P,Q) \neq 0$. 
Thus, $(\GG,\supp(\del))$ is a directed graph without multi-edges. Throughout we require that $(\GG, \supp(\del))$ is acyclic~-- \ie it is a dag. 
In particular, $\del(P,P) = 0$ for all vertices $P$. 
 A \ei{marking} is a mapping $M\colon  \GG \to \OS$. 
Each node $P\in \GG$ is associated in a natural way with a marking $\LL_P\colon  \GG \to \OS, \; Q \mapsto \del(P,Q)$ called its successor marking. The support of $\LL_P$ consists of the target nodes of outgoing edges from $P$.
We define the \ei{evaluation} $\e(P)$ of a node ($\e(M)$ of a marking resp.)
bottom-up in the dag by induction:
\begin{align*}
\e(\es) & = 0, \\
\e(P) & = 2^{\e(\LL_P)} &\text{for a node $P$}, \\
\e(M) & = \sum_{P}M(P)\e(P) &\text{for a marking  $M$}.
\end{align*}
We have $\eps(\LL_P) = \log_2(\eps(P))$, \ie the marking $\LL_P$ plays the role of a logarithm. Note that leaves (nodes of out-degree 0) evaluate to $1$ and every node evaluates to a positive real number. However, we are only interested in the case that all nodes evaluate to integers:
\begin{definition}\label{def:PC}
	A \ei{\PC} is a pair $(\GG,\del)$ with $\del\colon  \GG \times \GG\to \OS$ such that $(\GG, \supp(\del))$
	is a dag and all nodes evaluate to some positive natural number in $2^{\N}$. 
\end{definition}
The size of a power circuit is the number of nodes $\abs{\Gamma}$. By abuse of language, we also simply call $\Gamma$ a power circuit and suppress $\delta$ whenever it is clear. If $M$ is a marking on $\Gamma$ and $S \sse \Gamma$, we write $M|_S$ for the restriction of $M$ to $S$.
Let  $(\Gamma', \delta')$ be a power circuit, $\Gamma \sse \Gamma'$, $\delta = \delta'|_{\Gamma \times \Gamma}$, and  $\delta'|_{\Gamma \times (\Gamma'\setminus \Gamma)} = 0$. Then $(\Gamma, \delta)$ itself is a power circuit. We call it a \emph{sub-power circuit} and denote this by $(\Gamma, \delta) \leq (\Gamma', \delta')$ or, if $\delta$ is clear, by $\Gamma \leq \Gamma'$.

If $M$ is a marking on $ S \sse \Gamma$, we extend $M$ to $\Gamma$ by setting $M(P) = 0$ for $P \in \Gamma \setminus S$. With this convention, every marking on $\Gamma$ also can be seen as a marking on $\Gamma'$ if $\Gamma \leq \Gamma'$.

\begin{example}\label{ex:powtow}
	A \PC of size $n+1$ can realize $\tow(n)$ since a directed path of $n+1$ nodes represents
	$\tow(n)$ as the evaluation of the last node. The following power circuit realizes $\tow(5)$ using $6$ nodes: 	
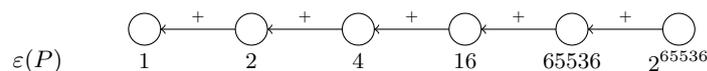
\begin{center}\vspace{-3mm}
			\tikzstyle{pcnode} = [minimum size = 12pt,circle,draw ]
				\begin{tikzpicture}[outer sep = 0pt, inner sep = 0.7pt, node distance = 40]
					{\footnotesize
						\node[pcnode] (1)    {};
						\node[pcnode, right of = 1] (2)   {};
						\node[pcnode, right of = 2] (3)    {};
						\node[pcnode, right of = 3] (4)   {};
						\node[pcnode, right of = 4] (5)   {};
						\node[pcnode, right of = 5] (6)   {};
						\node [below of=1, yshift=28](7){\footnotesize $1$};
						\node [below of=2, yshift=28](8) {\footnotesize $2$};
						\node [below of=3, yshift=28](9){\footnotesize $4$};
						\node [below of=4, yshift=28](10){\footnotesize $16$};
						\node [below of=5, yshift=28](11){\footnotesize $65536$};
						\node [below of=6, yshift=28](12){\footnotesize $2^{65536}$};
						\node [left of=1](13){};
						\node [left of=7](17){\small $\epsilon(P) $};
						
\path[->]         		(2) edge node[above, yshift=2]{\tiny +} (1)
						(3) edge node[above, yshift=2]{\tiny $+$}  (2)
						(4) edge node[above, yshift=2]{\tiny $+$} (3)
						(5) edge node[above, yshift=2]{\tiny $+$} (4)
						(6) edge node[above, yshift=2]{\tiny $+$} (5)	
						;
}
				\end{tikzpicture}
		
			\end{center}
\end{example}

\begin{example}\label{ex:binarybasis}
	We can represent every integer in the range 
	$[-2^n-1,2^n-1]$ as the evaluation of some marking in a \PC with node set $\oneset{P_{0} \lds P_{n-1}}$ with $\e(P_i) = 2^{i}$ for $i \in \oneinterval{n}$.
	Thus, we can convert the binary notation of an $n$-bit integer into a \PC
	with $n$ vertices, $\Oh( n\log n)$ edges (each successor marking requires at most $\floor{\log n} + 1$ edges) and depth at most $\log^*n$.
	 For an example of a marking representing the integer $23$, see \cref{fig: exmarking}. 
\end{example}

\begin{figure}[tbh]
	\begin{center}
		\tikzstyle{pcnode}=[minimum size= 12pt,circle,draw ]
		\vspace*{-0.2cm}
		\begin{tikzpicture}[scale=1, outer sep=0pt, inner sep = 0.7pt, node distance=1.2cm]
			{

				\node[pcnode] (0) at (0,0) {\textcolor{blue}{$-$}};
				\node[pcnode] (1) [right of=0] {};
				\node[pcnode] (2) [right of=1] {};
				\node[pcnode] (3) [right of=2] {\textcolor{blue}{+}};
				\node[pcnode] (4) [right of=3] {\textcolor{blue}{+}};
				\node[pcnode] (5) [right of=4] {};
				
				\node (6) [below of =0, yshift=10] {\footnotesize  1};
				\node (7) [below of =1,yshift=10] {\footnotesize  2};
				\node (8) [below of =2,yshift=10] {\footnotesize  4};
				\node (9) [below of =3,yshift=10] {\footnotesize 8};        
				\node (10) [below of =4,yshift=10] {\footnotesize  16};
				\node (11) [below of =5,yshift=10] {\footnotesize 32};

				\draw[->] (1) edge node[above, yshift=1] {\tiny+} (0)
				(2) edge node[above, yshift=1 ] {\tiny+} (1)
				(3) edge [bend left=28]node[below, xshift=-23, yshift=4] {\tiny+} (1)
				(3) edge[bend left=30] node[below left, xshift=-40, yshift=8] {\tiny+} (0)
				(4) edge[bend right=30] node[above,  xshift=23, yshift=-4] {\tiny+} (2)
				(5) edge[bend right=33] node[above, xshift=-20, yshift=-2] {\tiny+} (2)
				(5) edge[bend right=34] node[above, yshift=2] {\tiny+} (0)
				;
			}
		\end{tikzpicture}
	\end{center}
		\caption{Each integer $z \in \interval{-63}{63}$ can be represented by a marking in the following power circuit. The marking given in \textcolor{blue}{blue} is representing the number $23$.}\label{fig: exmarking}
\end{figure}
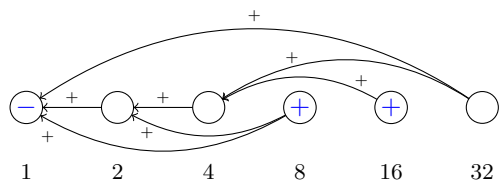

\begin{definition}\label{def:PCreduced}
We call a marking $M$ \emph{compact} if for all $P, Q \in \supp(M)$ with  $P \neq Q$ we have $\abs{\epsilon(\Lambda_{P})-\epsilon(\Lambda_{Q})} \geq 2$. 
A \emph{reduced power circuit} of size $n$ is a power circuit $(\Gamma,  \delta)$ with $\Gamma$ given as a sorted list $\Gamma = (P_{0}, \ldots, P_{n-1})$ such that all successor markings are compact and $\epsilon(P_{i}) < \epsilon(P_{j})$ whenever $ i < j$. In particular, all nodes have pairwise distinct evaluations. 
\end{definition}
It turns out to be crucial that the nodes in $\Gamma$ are sorted by their values. Still, sometimes it is convenient to treat $\Gamma$ as a set~-- we write $P \in \Gamma$ or $S \sse \Gamma$ with the obvious meaning. 
 Whenever convenient we assume that $\epsilon(\Lambda_{P_{i}}) = \infty$ for $i \geq n$. 

Notice that in \cite{DiekertLU13ijac} the data structure for a reduced power circuit also contains a bit-vector indicating which nodes have successor markings differing by one~-- we will compute this information on-the-fly whenever needed. 
For more details on power circuits see \cite{DiekertLU13ijac,MyasnikovUW12}.

\begin{remark}\label{rem:toposort}
	If $(\Gamma, \delta)$ is a reduced power circuit with $\Gamma = (P_{0}, \ldots, P_{n-1})$, we have $\delta(P_i, P_j) = 0$ for $j \geq i$. Thus, the order on $\Gamma$ by evaluations is also a topological order on the dag $(\Gamma, \supp(\delta))$.
\end{remark}

\section{Compact \sdr{}s}\label{sec:compact}

In this section we will show that for every binary number we can efficiently calculate a so-called unique compact representation. This will be crucial tool for the power circuit reduction process.

\begin{definition}
\begin{enumerate}[(i)]
\item A sequence $B = (b_{0}, \ldots, b_{m-1})$ with  $b_{i} \in \OS$ for $i \in \oneinterval{m} $ is called a \emph{\sdr} of $\val(B) = \sum_{i = 0}^{m-1}b_{i}\cdot2^{i} \in \Z$. 
\item  The digit-length of $B = (b_{0}, \ldots, b_{m-1})$ is the maximal $i$ with $b_{i-1} \neq 0$. 
\item 
The sequence $B = (b_{0}, \ldots, b_{m-1})$ is called \emph{compact} if $b_i b_{i-1} = 0$ for all $i \in \interval{1}{m-1}$ (\ie no two successive digits are non-zero).  
\end{enumerate}
\end{definition}
A non-negative binary number is the special case of a \sdr where all $b_i$ are $0$ or $1$ (note that, in general, they are not compact). 
Also negative binary numbers can be seen as special cases of \sdr{}s~-- though the precise form depends on the representation:
A negative number given as two's complement is a \sdr where the most-significant digit is a $\mOne$ and the other non-zero digits are $1$s; a negative signed magnitude representation can be viewed as \sdr where all non-zero digits are $\mOne$s.
In particular, every integer $k$ can be represented as a \sdr. However, in general, a \sdr for an integer $k$ is not unique.
 Still, we will prove in this section that each integer $k$, indeed, has a unique \csdr (see also \cite{MyasnikovUW12}). 

Note that by setting $b_i = 0$ for $i \geq m$, one can extend every \sdr $B = (b_{0}, \ldots, b_{m-1})$ to an arbitrarily long or infinite sequence. By doing so, $\val(B)$ and the digit-length of $B$ do not change. 

\subparagraph*{Computing \csdr{}s.}
In the following we will, amongst other things, show that for every binary number $A$ there exists such a compact signed digit representation $B$ of $A$ and that $B$ is unique with this property. We start with the existence and the complexity of calculating $B$.
While in \cite[Section 2.1]{MyasnikovUW12} a linear-time algorithm for calculating $B$ has been given, we aim for optimizing the parallel complexity.

\begin{theorem}\label{thm:compact}
	The following is in \Ac0: 

\compproblem{A binary number $A = (a_{0}, \ldots, a_{m-1})$.}{A compact signed-digit representation of $A$.}
	
\end{theorem}

Notice that \cref{thm:compact} implies that every integer has a \csdr.
Moreover, be aware that, clearly, the theorem is only true if we choose suitable encodings~-- in particular, we assume that the three values $\mOne,0,1$ are all encoded using two bits.

\begin{proof}
Let $A = (a_{0}, \ldots, a_{m-1})$ be a binary number. For $i\geq m$ we set $a_i = 0$.
We view the $a_i$ as Boolean variables and aim for constructing (almost) Boolean formulas for the compact representation. Since the digits of a compact representation are from the set $\oneset{\mOne,0,1}$, we treat the Boolean values $0,1$ as a subset of the integers and we will mix Boolean operations ($\land$, $\lor$, $\oplus$) with arithmetic operations ($+$, $\cdot$). Here $\oplus$ denotes the \emph{exclusive or}, which is addition modulo two.

For $i \geq 0$ we define 
\begin{align*}
c_{i}& = \bigvee_{1\leq j \leq i}\Bigl(\left(a_{j} \wedge a_{j-1}\right) \wedge \bigwedge_{j < k \leq i}\left(a_{k} \vee a_{k-1}\right)\Bigr),\\[3mm]
b_{i}& = \left(a_{i} \oplus c_{i}\right)\cdot(\mOne)^{a_{i+1}}.
\end{align*}
Moreover, we set $B = (b_{0},\ldots, b_{m-1}, b_{m})$. Observe that $c_{0} = 0$ and, hence, $b_0 = a_0 \cdot (\mOne)^{a_{1}}$. Furthermore, $b_{m} = c_{m}$ and $b_i = c_{i} = 0$ for $i \geq m+1$.

\begin{remark}\label{stz:ac0}
It is clear that the $c_{i}$ can be computed in \Ac0 and so the same holds for the $b_{i}$. This implies that on input of $A = (a_{0}, \ldots, a_{m-1})$, one can compute $B$ in \Ac0. By the very definition as a Boolean formula, it is clear, that it is actually in uniform \Ac0 (see \cref{rem:uniformAc0}). 
\end{remark}
Thus, in order to prove \cref{thm:compact}, it remains to show that $B = (b_{0}, \ldots, b_{m})$ is compact and that $\val(B) = \val(A)$. 
\begin{claim} \label{stz:formulac}
The $c_{i}$ satisfy the following recurrence: 
\begin{itemize}
\item $c_{0} = 0$
\item $c_{i} = (a_{i} \wedge a_{i-1}) \vee \bigl(c_{i-1} \wedge (a_{i} \vee a_{i-1})\bigr)$ for $i \geq 1$. 
\end{itemize}

\end{claim}
\begin{claimproof}
For $i = 0$, the claim holds because the empty disjunction is equal to $0$. 
Now we assume that $i \geq 1$ and that the recurrence holds for $i-1$. 
We set $X_{j} = a_{j} \wedge a_{j-1}$ and $Y_{j} = a_{j} \vee a_{j-1}$. Then we obtain
\begin{align*}
c_{i} = \bigvee_{1\leq j \leq i}\left(X_{j} \wedge \left(\bigwedge_{j < k \leq i} Y_{k} \right) \right) & = X_{i} \vee \left( \bigvee_{1\leq j \leq i-1} X_{j} \wedge  \left(\bigwedge_{j < k \leq i} Y_{k} \right)\right) \\[3mm]
& = X_{i}  \vee \left( \left( \bigvee_{1\leq j \leq i-1} X_{j} \wedge  \left(\bigwedge_{j < k \leq i-1} Y_{k} \right)\right) \wedge Y_{i} \right)\\[2mm]
& = X_{i}  \vee \left( c_{i-1}\wedge Y_{i} \right)
\end{align*}
 This proves the claim. 
\end{claimproof}

\begin{claim} \label{stz:Indstep}
Let $a_{i}, b_{i}$ and $c_{i}$ be as above. Then for all $k \geq 0$ we have
\[a_{k}+c_{k} = b_{k}+2c_{k+1}.\]
\end{claim}
\begin{claimproof}
	By \cref{stz:formulac} we have $c_{k+1} = \left(a_{k+1} \wedge a_{k} \right) \vee \left(c_{k} \wedge \left(a_{k+1} \vee a_{k} \right) \right)$. Thus, we can express both $b_{k}$ and $c_{k+1}$ in terms of $a_k$, $a_{k+1}$ and $c_k$. This leads us to the following table:
\begin{center}
\begin{tabular}[h]{|c|c|c||c|c|}
\hline
$a_{k}$ & $a_{k+1}$ & $c_{k} $ & $b_{k}$ & $c_{k+1}$ \\
\hline 
$0$ & $0$ & $0$ & $0$ & $0$ \\
\hline 
$0$ & $0$ & $1$ & $1$ & $0$ \\
\hline 
$0$ & $1$ & $0$ & $0$ & $0$ \\
\hline 
$0$ & $1$ & $1$ & $\mOne$ & $1$ \\
\hline 
$1$ & $0$ & $0$ & $1$ & $0$ \\
\hline 
$1$ & $0$ & $1$ & $0$ & $1$ \\
\hline 
$1$ & $1$ & $0$ & $\mOne$ & $1$ \\
\hline 
$1$ & $1$ & $1$ & $0$ & $1$ \\
\hline 
\end{tabular}
\end{center}
If we now take the values in the table as integer values and put them into the above equation, we see that the equation holds in all cases. 
\end{claimproof}

\begin{claim}\label{stz:induction}
Let $a_{i}, b_{i}$ and $c_{i}$ be as above. Then for all $k \geq 0$ we have 
\[\sum_{i = 0}^{k}2^{i}a_{i} = 2^{k+1}c_{k+1}+\sum_{i = 0}^{k} 2^{i}b_{i}.\]
\end{claim}
\begin{claimproof}
We use induction on $k$. 
Since $c_{0} = 0$ we have $a_{0} = 2 \cdot c_{1}+b_{0}$ by \cref{stz:Indstep}. Therefore, the equation holds for $k = 0$. 

Now let $k\geq 0 $. Then we obtain 
{\allowdisplaybreaks\begin{align*}
\sum_{i = 0}^{k+1}2^{i}a_{i}& = 2^{k+1}a_{k+1}+\sum_{i = 0}^{k}2^{i}a_{i}\\
& = 2^{k+1}a_{k+1} + 2^{k+1}c_{k+1}+\sum_{i = 0}^{k}2^{i}b_{i}\tag{by induction}\\
& = 2^{k+1} \left(a_{k+1} +c_{k+1} \right)+\sum_{i = 0}^{k}2^{i}b_{i}\\
& = 2^{k+1} \left(b_{k+1}+2 c_{k+2}\right) +\sum_{i = 0}^{k}2^{i}b_{i} \tag{by \cref{stz:Indstep}}\\
& = 2^{k+2} c_{k+2} +\sum_{i = 0}^{k+1}2^{i}b_{i}
\end{align*}}
This proves the claim. 
\end{claimproof}

\begin{claim}\label{stz: compact}
Let $B = (b_{0},\ldots, b_{m-1}, b_{m})$ be as defined above. Then $B$ is compact. 
\end{claim}
\begin{claimproof}
We have to make sure that there is no $i \in \oneinterval{m}$ such that $b_{i}\neq 0$ and $b_{i+1}\neq 0$. In order to do so, we express $b_{i}$ and $b_{i+1}$ in terms of $a_i$, $a_{i+1}$ and $c_i$.
Notice that $b_{i+1}$ is not fully determined by $a_i$, $a_{i+1}$ and $c_i$. Still these three values tell us whether $b_{i+1}$ is zero or not.
This leads us to the following table, which shows that $B$ is, indeed, compact:
\begin{center}
\begin{tabular}[h]{|c|c|c||c|c|c|}
\hline
$a_{i}$ & $a_{i+1}$ & $c_{i} $ & $c_{i+1} $ & $b_{i}$ & $b_{i+1}$ \\
\hline
$0$ & $0$ & $0$ & $0$ & $0 $ & $0$ \\
\hline
$0$ & $0$ & $1$ & $0$ & $1 $ & $0$ \\
\hline
$0$ & $1$ & $0$ & $0$ & $0 $ & $\pm 1$ \\
\hline
$0$ & $1$ & $1$ & $1$ & $\mOne$ & $0 $ \\
\hline
$1$ & $0$ & $0$ & $0$ & $ 1 $ & $0$ \\
\hline
$1$ & $0$ & $1$ & $1$ & $0 $ & $\pm 1$ \\
\hline
$1$ & $1$ & $0$ & $1$ & $\mOne $ & $0$ \\
\hline

$1$ & $1$ & $1$ & $1$ & $0$  & $0$ \\
\hline
\end{tabular}
\end{center}
\end{claimproof}

Now we are ready to finish the proof of \cref{thm:compact}.
Let $A = (a_{0}, \ldots, a_{m-1})$ and $B = (b_{0}, \ldots, b_{m-1}, b_{m})$ as above. By \cref{stz: compact}, $B$ is compact.  
Moreover, we have
\begin{align*}
\val(B)& = \sum_{i = 0}^{m}2^{i}b_{i}
 	 = 2^{m}c_{m} + \sum_{i = 0}^{m-1}2^{i}b_{i}\tag{since $c_{m} = b_{m}$}\\
 	& = \sum_{i = 0}^{m-1}2^{i}a_{i} = \val(A).\tag{by \cref{stz:induction}}
\end{align*}
Therefore, $B$ is a \csdr for $A$ as claimed in \cref{thm:compact}. By \cref{stz:ac0}, it can be computed in \Ac0. 
\end{proof}

\subparagraph*{Uniqueness of \csdr{}s.}

The following lemmas are crucial tools both for proving uniqueness of compact representations and for the power circuit reduction process, which we describe later.
In \cite[Section 2.1]{MyasnikovUW12} similar statements can be found. 

\begin{lemma}\label{lem:valcom}
	Let $A$ be a \csdr{}  and let $B = (b_{0}, \ldots, b_{n-1})$ be a \csdr of digit-length $n$ such that $b_i = n-i \bmod 2$ (\ie $b_{n-1} = 1$ and then $B$ alternates between $0$ and $1$).
	Then we have
	\begin{enumerate}[(i)]
		\item\label{maxCompactVal} $\val(B) = \floor{\frac{2^{n+1}}{3}}$,
		\item\label{compactCompareMax} $\val(A) \leq \val(B)$ if and only if $\val(A) \leq 0$ or the digit-length of $A$ is at most $n$. 
	\end{enumerate}
\end{lemma}

\begin{proof}	
	First, we want to calculate $\val(B)$. If $n$ is even, then 
	\begin{align*}
		\val(B) = \sum_{i = 0}^{\frac{n}{2}-1}2^{2i+1} = 2\sum_{i = 0}^{\frac{n}{2}-1}4^{i} = 2 \cdot \frac{1-4^{\frac{n}{2}}}{1-4} = \frac{2}{3}\cdot (2^{n}-1).
	\end{align*}
	If $n$ is odd, then 
	\begin{align*}
		\val(B) = \sum_{i = 0}^{\frac{n-1}{2}}2^{2i} = \sum_{i = 0}^{\frac{n-1}{2}}4^{i} = \frac{1-4^{\frac{n-1}{2}+1}}{1-4} = \frac{2^{n+1}-1}{3}
	\end{align*}
	showing that in any case $\val(B) = \floor{\frac{2^{n+1}}{3}}$. 
	
	In order to see (\ref{compactCompareMax}), we denote $A = (a_{0}, \ldots, a_{n-1})$. If $\val(A)\leq 0$, then clearly $\val(A)\leq \val(B)$. Hence, assume that the digit-length of $A$ is at most $n$ and consider the following operations:
	\begin{enumerate}
		\item If $a_{i} = \mOne$, then set $a_{i} = 0$. 
		\item If $a_{n-1} = 0$, then set $a_{n-1} = 1$ and set $a_{n-2} = 0$. 
		\item If $a_{i} = a_{i+1} = 0$ with $i \in \interval{1}{n-2} $, then set $a_{i} = 1$ and $a_{i-1} = 0$ (technically, this rule subsumes the previous rule). 
		\item If $a_0 = a_1 = 0$, set $a_0 = 1$. 
	\end{enumerate}
	Let $A'$ be the number we obtained after applying at least one of the above operations to $A$ (if this is possible). Then $A'$ is also a \csdr, the digit-length of $A'$ is at most $ n$, and $\val(A) < \val(A')$. Moreover, if $A\neq B$, then we always can apply one of these rules. This shows that $\val(A) \leq \val(B)$.

	On the other hand, assume that the digit-length of $A$ is $m$  with $m \geq n+1$. First, assume that $a_{m-1} = 1$ and set $A' = (a_{0}, \ldots, a_{m-3})$. Then, since $A$ is compact, we have $a_{m-2} = 0$ and, hence, $\val(A) = 2^{m-1} + \val(A')$.	
	By the previous implication and part (\ref{maxCompactVal}), we know that $\abs{\val(A')} \leq \floor{\frac{2^{m-1}}{3}}$. Therefore, $\val(A) \geq 2^{m-1} - \abs{\val(A')} \geq 2^{m-1} - \floor{\frac{2^{m-1}}{3}} > \floor{\frac{2^{m}}{3}} \geq \floor{\frac{2^{n+1}}{3}} = \val(B)$.	
	If $a_{m-1} = -1$, we obtain $ \val(A) < 0$ with the same argument.
\end{proof}

\begin{lemma}[\,\!{\cite[Lemma 4]{MyasnikovUW12}}]\label{lem:compareM}
	Let $A = (a_{0}, \ldots, a_{m-1})$, $B = (b_{0}, \ldots,b_{m-1})$ be \csdr{}s.
	Then: 
	\begin{enumerate}[(i)] 
		\item\label{compactCompareValEqual} $\val(A) = \val(B) $ if and only if  $a_{i} = b_{i}$ for all $i \in \oneinterval{m}$.
		\item \label{compactCompareVal} Assume there is some $i$ with $a_{i} \neq b_{i}$ and let $i_{0} = \max\set{i \in \oneinterval{m}}{a_{i} \neq b_{i}}$. Then $\val(A) < \val(B)$ if and only if $a_{i_{0}} < b_{i_{0}}$. 
	\end{enumerate}
\end{lemma}

\begin{proof}
	Notice that (\ref{compactCompareValEqual}) is an immediate consequence of (\ref{compactCompareVal}).
	In order to see (\ref{compactCompareVal}), observe that it suffices to show only one implication. Let $A' = (a_{0}, \ldots, a_{i_0})$ and $B' = (b_{0}, \ldots, b_{i_0})$ and assume that $0 = a_{i_{0}} < b_{i_{0}} = 1$ (the cases involving the value $\mOne$ follow with the same argument).
	Now, $A'$ and $B'$ are \csdr{}s, so by \cref{lem:valcom}, $\val(A') \leq \floor{\frac{2^{i_0+1}}{3}}$ and $\val(B') > \floor{\frac{2^{i_0+1}}{3}}$. Hence, $\val(A)<\val(B) $.
\end{proof}

From this lemma together with \cref{thm:compact} it follows that each $k \in \Z$ can be uniquely represented by a compact signed digit representation $\cor(k)$. Likewise for a signed digit representation $A$, we write $\cor(A)$ for its compact signed digit representation.

\begin{corollary}\label{cor:signedDigitOps} 
	The following problems are in \Ac0:
	\begin{enumerate}[(i)]
		\item\label{sdrToCompact} \compproblem{A \sdr{} $A$. }{$\cor(A)$.} 
		\item\label{sumCompact} \compproblem{Signed-digit representations $A$ and $B$. }{The compact signed-digit representation of $\val(A) + \val(B)$.}
		\item\label{sdrCompare} \problem{Signed-digit representations $A$ and $B$.}{Is $\val(A) < \val(B)$?}
	\end{enumerate}
\end{corollary}
\begin{proof}
	Given a \sdr $A = (a_{0}, \ldots, a_{m-1})$, we can split it into two non-negative binary numbers $B, C$ such that $\val(A) = \val(B) - \val(C)$ (\ie $b_i = \max\oneset{0,a_i}$ and $c_i = -\min\oneset{0,a_i}$). From these binary numbers we can compute the difference in \Ac0 and then make the result compact using \cref{thm:compact}. To see (\ref{sumCompact}), we proceed exactly the same way. For comparing two \sdr{}s, we compute their compact representations using part (\ref{sdrToCompact}) and then compare them in \Ac0 by evaluating the condition in \cref{lem:compareM}. 
\end{proof}

\section{Operations on power circuits}\label{sec:PCs}

\subsection{Basic operations}\label{sec:PCoperations}

Before we consider the computation of reduced power circuits, which is our main result in this section, let us introduce some more notation on power circuits and recall the basic operations from \cite{MyasnikovUW12,DiekertLU13ijac} under circuit complexity aspects.

\subparagraph*{Markings and chains.}

\begin{definition}\label{def:chain}
	
	Let $(\Gamma, \delta)$ be a reduced power circuit with $\Gamma$ given as the sorted list $\Gamma = (P_{0}, \ldots , P_{n-1})$. 
\begin{enumerate}[(i)]
\item A \emph{chain} $C$ of length $\abs{C}=\ell$ in $\Gamma$ starting at $P_{i} = \startc{C}$ is a sequence $(P_{i}, \ldots, P_{i+\ell-1})$ such that $\epsilon(P_{i+j+1}) = 2 \cdot \epsilon(P_{i+j})$ for all $j \in \oneinterval{\ell-1} $. 

	In particular, $\epsilon(P_{i+j}) = 2^{j}\cdot\epsilon(P_{i})$ for all $j \in \oneinterval{\ell} $.  As we do for $\Gamma$, we treat a chain both as a sorted list and as a set.

\item We call a chain $C$ \emph{maximal} if it cannot be extended in either direction. 
We denote the set of all maximal chains by $\mathcal{C}_{\Gamma}$.

	As a set, a reduced power circuit is the disjoint union of its maximal chains.

\item Let $M$ be a marking in the reduced power circuit $(\Gamma, \delta)$ and let $C = (P_{i}, \ldots, P_{i+\ell-1}) \in \mathcal{C}_{\Gamma}$ and define $a_{j} = M(P_{i+j})$ for $i \in \oneinterval{\ell} $. Then we write $\binM{M}{C} = (a_{0}, \ldots, a_{\ell-1})$.

\item There is a unique maximal chain $C_0$ containing the node $P_0$ of value $1$. We call $C_0$ the \emph{initial maximal chain} of $\Gamma$ and denote it by $C_0 = C_0(\Gamma)$.
\end{enumerate}

\end{definition}

For an example of a power circuit with three maximal chains, see \cref{fig: exchains}.

\begin{figure}[h]  
	\begin{center}
		\begin{tikzpicture}[scale=1.5, outer sep=0pt, inner sep = 0.7pt, node distance=1.2cm]{		
				\tikzstyle{pcnode}=[minimum size= 12pt,circle,draw ]
				\node[pcnode] (0) at (0,0) {};
				\node[pcnode] (1) [right of=0] {};
				\node[pcnode] (2) [right of=1] {};
				\node[pcnode] (3) [right of=2] {};
				\node[pcnode] (4) [right of=3] {};
				\node[pcnode] (5) [right of=4] {};
				\node[pcnode] (6) [right of=5] {};
				
				\node (7) [below of =0, yshift=17] {\footnotesize 1};
				\node (8) [below of =1,yshift=17] {\footnotesize 2};
				\node (9) [below of =2,yshift=17] {\footnotesize 4};
				\node (10) [below of =3,yshift=17] {\footnotesize 8};        
				\node (11) [below of =4,yshift=17] {\footnotesize $2^{8}$};
				\node (12) [below of =5,yshift=17] {\footnotesize $2^{9}$};
				\node (13) [below of =6,yshift=17] {\footnotesize $2^{2^{9}}$};

				\draw[->] (1) edge node[above, yshift=1] {\tiny+} (0)
				(2) edge node[above, yshift=1 ] {\tiny+} (1)
				(3) edge node[above, yshift=1] {\tiny+} (2)
				(3) edge[bend left=20] node[below left, xshift=-31, yshift=5.5] {\footnotesize$-$} (0)
				(4) edge node[above, yshift=1] {\tiny+} (3)
				(5) edge[bend right=30] node[above, xshift=-20, yshift=-2] {\tiny+} (3)
				(5) edge[bend right=28] node[above, yshift=2] {\tiny+} (0)
				(6) edge node[above, yshift=1] {\tiny+} (5)
				;
				
			}

		\end{tikzpicture}
		\caption{This power circuit is an example for a reduced power circuit with three maximal chains: The first one consists of the nodes of values $1, 2,4,8$, the next one is formed by the nodes of values $2^{8}$ and $2^{9}$ and the node of value $2^{2^{9}}$ is a maximal chain of length $1$. }\label{fig: exchains}
	\end{center}
\vspace{-5mm}
\end{figure}

We will show how to computationally find the maximal chains in \cref{cor:findChains}.
The following facts are clear from the definition of maximal chains:

\begin{fact}\label{lem:representM}
	Let  $(\Gamma, \delta)$ be a reduced power circuit and let $M$ be a marking on $\Gamma$. Then the following holds:
	\begin{enumerate}[(i)]	
	\item\label{MultByStart} $\epsilon(M|_{C}) = \epsilon(\startc{C})\cdot \val(\binM{M}{C})$ for every chain $C$ in $\Gamma$ (even if $C$ is not maximal). 
	\item $\epsilon(M) = \sum_{C \in \cC_\Gamma} \epsilon(\startc{C})\cdot \val(\binM{M}{C})$.
	\item The marking $M$ is compact if and only if $\binM{M}{C}$ is compact for all $C \in \mathcal{C}_{\Gamma}$. 
	\end{enumerate}
\end{fact}

\begin{lemma}\label{stz: smalldiffs}
Let $(\Gamma, \delta)$ be a reduced power circuit. Let $L$ and $M$ be compact markings in $\Gamma$ such that $\epsilon(L) > \epsilon(M)$ and let $0\leq k \leq \floor{\frac{2^{\abs{C_{0}}+1}}{3} }$.  Then $\epsilon(L) \leq \epsilon(M) +k$ if and only if the following holds:
\begin{itemize}
\item $\epsilon(M|_{\Gamma \setminus C_{0}}) = \epsilon(L|_{\Gamma \setminus C_{0}})$ and 
\item $\epsilon (L|_{C_{0}}) \leq \epsilon(M|_{C_{0}}) + k$. 
\end{itemize}
\end{lemma}
\begin{proof}
We first assume that $\epsilon(M|_{\Gamma \setminus C_{0}}) \neq \epsilon(L|_{\Gamma \setminus C_{0}})$ and aim for showing that $\epsilon(L) > \epsilon(M)+k$:
 By \cref{lem:valcom}, we have $\abs{\epsilon(L|_{C_{0} })}, \abs{\epsilon(M|_{C_{0} })} \leq \floor{\frac{2^{ \abs{C_{0}} +1 }}{3}}$. Hence,
 \[\absbig{\epsilon(M|_{C_{0} }) + k -\epsilon(L|_{C_{0} })} \leq \abs{\epsilon(M|_{C_{0} })} + k + \abs{\epsilon(L|_{C_{0} })}  \leq 3 \floor{\frac{2^{\abs{C_{0}}+1}}{3} } \leq 2^{\abs{C_{0} }+1}-1.\] 
Furthermore, $\epsilon(L|_{\Gamma \setminus C_{0}}) - \epsilon(M|_{\Gamma \setminus C_{0}})$ is a multiple of $2^{\abs{C_{0}}+1}$. Therefore, by the assumption $\epsilon(L) > \epsilon(M)$ and \cref{lem:compareM}, we obtain $\epsilon(L|_{\Gamma \setminus C_{0}}) - \epsilon(M|_{\Gamma \setminus C_{0}}) \geq 2^{\abs{C_{0}}+1}$.
It follows that
\[\epsilon(L|_{\Gamma \setminus C_{0}})- \epsilon(M|_{\Gamma \setminus C_{0}})+\epsilon(L|_{C_{0}})-\epsilon(M|_{C_{0}})-k \geq 1\]
and so $\epsilon(L) > \epsilon(M)+k$.

Now assume that $\epsilon(M|_{\Gamma \setminus C_{0}}) = \epsilon(L|_{\Gamma \setminus C_{0}})$. It remains to show that under this assumption we have $\epsilon(L) \leq \epsilon(M) +k$ if and only if $\epsilon (L|_{C_{0}}) \leq \epsilon (M|_{C_{0}})+k $.
However, this follows immediately from the fact that
\[\epsilon(L) = \epsilon(L|_{\Gamma \setminus C_{0}})+\epsilon (L|_{C_{0}}) \quad \text{ and }\quad \epsilon(M)+k = \epsilon(M|_{\Gamma \setminus C_{0}})+\epsilon (M|_{C_{0}})+k. \]
This finishes the proof of the lemma. 
\end{proof}

\subparagraph*{Comparison of markings.}

\begin{lemma}\label{lem:findC0}
	Given a reduced power circuit $(\Gamma, \delta)$ and a node $P \in \Gamma$ one can decide in \Ac0 whether $P \in C_0$.
\end{lemma}

\begin{remark}
	Since membership in \Ac0 often highly depends on the encoding of the input, in the following we always assume that power circuits are given in a suitable way. In particular, we may assume that an $n$-node power circuit is given by the $n\times n$ matrix representing $\delta$ where each entry from $\oneset{0,\pm1}$ is encoded using two bits. Moreover, in order to represent power circuits with fewer nodes within the same data structure, we can allow one \emph{deleted} bit for every row and column of the matrix.
	Markings can be encoded the same way by a sequence of $n$ symbols from $\oneset{0,\pm1}$. Moreover, if the power circuit is reduced, we also assume that the matrix representing $\delta$ is already in the sorted order (in particular, the ordering is not given by some separate data structure).
	
	In the following, we do not further consider these encoding issues. Moreover, as soon as we are dealing with \TC circuits, there is a lot of freedom how to encode inputs. 
\end{remark}

\begin{proof}[Proof of \cref{lem:findC0}]
	Let $\Gamma = (P_{0}, \ldots, P_{n-1})$. For each $i$ we define a \sdr $A_i = (a_{i,0}, \dots, a_{i,n-1})$ by $a_{i,j} = \Lambda_{P_i}(P_j)$. These \sdr{}s might not be compact, but, if $P_i \in C_0$, then $A_i$ is compact (this is because, by \cref{rem:toposort}, $P_i$ has only successors in $C_0$). Using \cref{cor:signedDigitOps}, we can compute the maximal $i_{max}$ such that $A_{i_{max}}$ is compact and for all $i < {i_{max}}$ also $A_i$ is compact and $ \val(A_{i+1}) = \val(A_i) + 1 = i+1$ (checking whether $A_i$ is compact, clearly, can be done in \Ac0). By a straightforward induction, we obtain that for all $i \leq i_{max}$ we have $\val(A_i) = \eps(\Lambda_{P_i})$ and $P_i \in C_0$. On the other hand, clearly, $P_{i_{max}+1} \not\in C_0$. Hence, we have computed $C_0$.
\end{proof}

 \begin{proposition} \label{lem:compareCompactMarkings}
	Let ${\comparator} \in \compOpSet$. The following problems are in \Ac0:
	\begin{enumerate}[(a)]
		\item\label{compMarkings} \problem{A reduced power circuit $(\Gamma, \delta)$ and compact markings $L$ and $M$ on $\Gamma$.}{Is $\epsilon(L) \comparator \epsilon(M)$? }
		\item\label{compMarkingsPlusK} \problem{A reduced power circuit $(\Gamma, \delta)$ with compact markings $L, M$ and $k \in \interval{0}{\bigl\lfloor\frac{2^{\abs{C_{0}}+1}}{3} \bigr\rfloor } $ given in binary.}{ Is $\epsilon(L) \comparator \epsilon(M)+ k $?} 
	\end{enumerate}
\end{proposition}
\begin{proof} 
	Let us choose $\leq$ as $\comparator$ (the other cases follow from this case in a straightforward way).

Let $\Gamma = (P_{0}, \ldots, P_{n-1})$. By \cref{lem:compareM} (\ref{compMarkings}) we can check in \Ac0 if $\epsilon(L)=\epsilon(M)$. If this is not the case, then by \cref{lem:compareM} (\ref{compactCompareVal})  we have $\epsilon(M) < \epsilon(L)$ if and only if $M(P_{i_{0}}) < L(P_{i_{0}})$ for $i_{0} = \max\set{i \in \oneinterval{n}}{M(P_{i}) \neq L(P_{i})}$. Now, $i_0$ can be found in \Ac0 and, hence, the whole check is in \Ac0. This proves part (\ref{compMarkings}).

For part (\ref{compMarkingsPlusK}) we first check whether $\epsilon(L) \leq \epsilon(M)$. If yes, then $\epsilon(L) \leq \epsilon(M)+k$. According to part (\ref{compMarkings}), this check is possible in \Ac0. 
Now assume that  $\epsilon(L) > \epsilon(M)$. By \cref{lem:findC0}, we can compute $C_0$ in \Ac0.
 By \cref{stz: smalldiffs} we know that $\epsilon(L) \leq \epsilon(M)+k$ if and only if $\epsilon(M|_{\Gamma \setminus C_{0}}) = \epsilon(L|_{\Gamma \setminus C_{0}})$ and $\epsilon (L|_{C_{0}}) \leq \epsilon(M|_{C_{0}}) + k$. The markings $M|_{\Gamma \setminus C_{0}}$ and $L|_{\Gamma \setminus C_{0}}$ are still compact markings in a reduced power circuit, and so we are able to decide in \Ac0 if that equality holds by part (a). So it remains to check if $\epsilon (L|_{C_{0}}) \leq \epsilon(M|_{C_{0}}) + k$. This amounts to an addition and a comparison of signed-digit representations of digit-length at most $\abs{C_{0}}+1$ (according to \cref{lem:valcom}), which both can be done in \Ac0 (see \cref{cor:signedDigitOps}). Thus, $\epsilon (L) \comparator \epsilon(M) + k$ can be checked in \Ac0. 
\end{proof}

\begin{corollary}\label{cor:findChains}
	We can decide in \Ac0, given a reduced power circuit $(\Gamma, \delta)$ and nodes $P,Q \in \Gamma$, whether $P$ and $Q$ belong to the same maximal chain of $\Gamma$.
\end{corollary}
\begin{proof}
	Let $P = P_i$ and $Q = P_j$ with $i < j$. Then $P$ and $Q$ belong to the same maximal chain if and only if $\eps(\Lambda_{P_{\ell} + 1}) = \eps(\Lambda_{P_{\ell} })+ 1$ for all $\ell \in \interval{i}{j - 1}$. The latter can be checked in \Ac0 using \cref{lem:compareCompactMarkings}.
\end{proof}

\subparagraph*{Calculations with markings.}

\begin{lemma}\label{lem: OpOnMarkings}
The following problems are all in $\TC$:
\begin{enumerate}[(a)]
\item  \label{addmark} \compproblem{A power circuit $(\Pi, \delta_{\Pi})$ together with markings $K$ and $L$.}
{A power circuit $(\Pi', \delta_{\Pi'})$ with a marking $M$ such that $\epsilon(M) = \epsilon(K) + \epsilon(L)$ and $(\Pi, \delta_{\Pi}) \leq (\Pi', \delta_{\Pi'})$,  $\abs{\Pi'} \leq 2\cdot \abs{\Pi}$ and $\depth(\Pi') = \depth(\Pi)$.}  

\item \label{submark} \compproblem{A power circuit $(\Pi, \delta_{\Pi})$ together with a marking $L$.}
{A marking $M$ in the power circuit $(\Pi, \delta_{\Pi})$ such that $\epsilon(M) = - \epsilon(L)$.} 

\item \label{powermark} \compproblem{A power circuit $(\Pi, \delta_{\Pi})$ together with markings $K$ and $L$ such that $\epsilon(L) \geq 0$.}
{A power circuit $(\Pi', \delta_{\Pi'})$ with a marking $M$ such that  $\epsilon(M) = \epsilon(K)\cdot 2^{\epsilon(L)}$ and  $(\Pi, \delta_{\Pi}) \leq (\Pi', \delta_{\Pi'})$, $\abs{\Pi'} \leq 3\cdot \abs{\Pi}$ and $\depth(\Pi') \leq \depth(\Pi)+1$.} 	
\end{enumerate}
\end{lemma}

The proof of this lemma uses the following construction  (see also \cite{DiekertLU13ijac}): 
\begin{definition}
	Let $(\Pi, \delta)$ be a power circuit and let $M$ be a marking on $\Pi$. 
	\begin{enumerate}[(a)]
		\item Let $P \in \Pi$. We define a new power circuit $\Pi \cup \oneset{\clone(P)}$ where $\clone(P)$ is a new node with $\Lambda_{\clone(P)} = \Lambda_{P}$. 
		\item We define a marking $\clone(M)$ as follows: First we clone all the nodes in $\supp(M)$. Then we set $\clone(M)(\clone(P)) = M(P)$ for $P \in \supp(M)$ and $\clone(M)(P) = 0$ otherwise. 
	\end{enumerate}
\end{definition}

It is clear that the problem, given a power circuit $(\Pi, \delta)$ and a marking $M$, compute a new power circuit $(\Pi', \delta')$ containing $\clone(M)$ is in \TC~-- and even in \Ac0 when defining the underlying data structure properly. Notice that $\abs{\Pi'} \leq 2 \cdot \abs{\Pi}$ and $\depth(\Pi') = \depth(\Pi)$.  

\begin{proof}
	We apply the constructions described in \cite[Section 7]{MyasnikovUW12} and \cite[Section 2]{DiekertLU13ijac}.
	
	Part (\ref{addmark}): First, we clone the marking $K$ leading to a power circuit $(\Pi', \delta_{\Pi'})$ of size at most $2 \cdot \abs{\Pi}$. Now $\clone(K)$ and $L$ certainly have disjoint supports. Then we define
	\begin{align*}
		M(P) = 
		\begin{cases}
			 \clone(K)(P), & P \in \supp(\clone(K)), \\ 
			 L(P),& P \in \supp(L), \\
			 0,   &~ \text{otherwise.} 
		\end{cases}
	\end{align*}
	Clearly, $\epsilon(M) = \epsilon(K) + \epsilon(L)$, and $M$ can be output in \TC.
	To show (\ref{submark}), we set
	\begin{align*}
		M(P) = 
		\begin{cases}
			- L(P),& P \in \supp(L), \\
			 0,    &~ \text{otherwise.} 
		\end{cases}
	\end{align*}
	As for (a), to define $M(P)$ we only have to look up $L(P)$ and change the sign and we do not have to create any new nodes or edges~-- so this can be done even in \Ac0. 
	
	To (c): To obtain $M$, we follow a similar approach as described in \cite[Section 2]{DiekertLU13ijac}. We first clone the markings $K$ and $L$ and so obtain markings $\clone(K)$ and $\clone(L)$.  At this point, the size of $\Pi$ increased by a factor of at most three. 
	
	Next we create new edges from every node $P \in\supp( \clone( K ))$ to every node $Q \in \supp(\clone(L))$ such that $\delta(P,Q) = \clone(L)(Q)$. This operation does not change the size of the power circuit, but it increases the depth by at most $1$ since there are no incoming edges to nodes in $\supp( \clone( K )$. Then the marking $\clone(K)$ is the marking we search for. 	
\end{proof}

Notice that the construction in (\ref{powermark}) also yields $\epsilon(M) = \epsilon(K) \cdot 2^{\epsilon(L)}$ in the case that $\epsilon(L) < 0$. However, then the resulting graph might not be a power circuit anymore since it might have nodes of non-integral evaluation.
Note that \cite{DiekertLU13ijac} is not very precise here: it is actually not sufficient that $\epsilon(K) \cdot 2^{\epsilon(L)} \in \Z$ in order to assure that there are no nodes of non-integral evaluation.

\newcommand{\eval}{\operatorname{eval}}
\newcommand{\ptxCircuit}{$(0,+,-,2^x)$-circuit\xspace}

\subsection{Relation to arithmetic circuits with + and $2^x$ gates.}\label{app:arithmeticCircuits}

Before we proceed to the power circuit reduction, our main result on power circuits, let us elaborate on the relation of power circuits to more general arithmetic circuits.
A (constant) \ptxCircuit is a dag where each node is either a $0$- (\ie a constant-), +-, $-$- or a $2^x$-gate. $0$-gates have zero inputs, +-gates two and $-$- and $2^x$-gates have one input. There is one designated output gate. The evaluation $\eval(\cC)$ of such a circuit $\cC$ is defined in a straightforward way (as a real number~-- in general, it might not be an integer). The $2^x$-depth of a circuit $\cC$ denoted by $\mathrm{depth}_{2^x}(\cC)$ is the maximal number of $2^x$-gates on any path in the circuit.

\begin{proposition}\label{prop:convertToArithmeticCircuit}
	For every power circuit $(\Pi,\delta)$ with a marking $M$, there is a  \ptxCircuit $\cC$ with $\eval(\cC) = \eps(M)$ such that
	\begin{itemize}
		\item  $\abs{\cC} \leq 2\abs{\sigma(\delta)} + 3\abs{\Pi} + 1$,
		\item $\mathrm{depth}(\cC) \leq (\depth(\Pi)+2) \cdot (\ceil{\log(\abs{\Pi})}+2)$ and
		\item $\mathrm{depth}_{2^x}(\cC) = \mathrm{depth}(\Pi) + 1$. 
	\end{itemize}
	Moreover, $\cC$ can be computed in \TC.
	
	\medskip
	Conversely, for every \ptxCircuit $\cC$ there is a power circuit\footnote{Note that technically speaking it is not a power circuit as defined in \cref{def:PC} since it might have nodes not evaluating to integers. Since we have not introduced a terminology for power circuits without this integrality condition, we use the term ``power circuit'' here.} $(\Pi,\delta)$ with a marking $M$ such that $\eval(\cC) = \eps(M)$  and
	\begin{itemize}
		\item $\mathrm{depth}(\Pi) \leq \mathrm{depth}_{2^x}(\cC)$ and 
		\item $\abs{\Pi} \leq \abs{\cC}^2 + \abs{\cC}$. 
	\end{itemize}
	Moreover, $(\Pi,\delta)$ and $M$ can be computed in \Nc2. 
	
	If the input of every $2^x$-gate of $\cC$ is non-negative, then $(\Pi,\delta)$ is, indeed, a power circuit~-- \ie all nodes evaluate to positive integers.
	
\end{proposition}
\begin{proof}
	In order to transform $(\Pi,\delta)$ with a marking $M$ into a \ptxCircuit $\cC$, we proceed as follows: we create one $0$-gate; for every leaf (node of out-degree zero) of $\Pi$ we create a $2^x$-gate with input coming from the $0$-gate, for every other node of $\Pi$, we create a $2^x$-gate whose input we describe next. For every marking (both $M$ and the successor markings $\Lambda_{P}$) we create a tree of $+$-gates (possibly with some $-$-gates) of logarithmic depth where the leaves correspond to some of the (already created) $2^x$- or $0$-gates and the last $+$-gate (\ie the root) evaluates to $\epsilon(M)$ (resp.\ $\epsilon(\Lambda_p)$). Now, the $2^x$-gate corresponding to a node $P \in \Pi$ receives its input from the $+$-gate corresponding to $\Lambda_{P}$. It is straightforward that this construction can be done in \TC.
	
	Clearly, this process introduces at most one $+$-gate and one $-$-gate for every pair in the support of $\delta$ and every node in the support of $M$. So we have at most $2\abs{\sigma(\delta)} + 2\abs{\Pi}$ many $+$ and $-$-gates.	Since there is one $0$-gate and $\abs{\Pi}$ many $2^x$-gates, the total number of gates is at most $2\abs{\sigma(\delta)} + 3\abs{\Pi} + 1$. It is clear that $\mathrm{depth}_{2^x}(\cC) = \mathrm{depth}(\Pi) + 1$ (note that the depth increases by one because leaves of $\Pi$ are replaced by $2^x$-gates with input from the $0$-gate). 
	Moreover, the depth of each $+$-tree is bounded by $\ceil{\log(\abs{\Pi})}$; introducing the $-$-gates and connecting to the $2^x$-gates increases the depth further by 2 (note that for the $2^x$-gates with single input from a $0$-gate this is a huge over-estimate). Since $\mathrm{depth}_{2^x}(\cC) = \mathrm{depth}(\Pi) + 1$  and we have one additional $+$-tree for the marking $M$, the total depth is at most $(\depth(\Pi)+2) \cdot (\ceil{\log(\abs{\Pi})}+2)$. 
	
	\medskip
	Now consider a \ptxCircuit $\cC$ with $n$ gates.
	Let $h_1, \dots, h_k$ be the $2^x$-gates of $\cC$. As a first step, we replace each $2^x$-gate $h_i$ by an ``input'' gate $X_i$ and cut its incoming wire. Thus, we obtain an arithmetic circuit over $\Z$ with $+$- and $-$-gates. Each $+$- or $-$-gate $g$ computes a linear combination $\sum_{i=1}^{k} a_{g,i}X_i$ with $a_{g,i} \in \Z$. By \cite[Theorem 21]{Travers06}, the $a_{g,i}$ can be computed in \GapL, and hence, in \Nc{2} (\cite[Theorem 4.1]{AlvarezJ93}). Notice that $\abs{a_{g,i}} < 2^{n}$ for all $g$ and $i$.
	
	Now, to construct the power circuit $(\Pi,\delta)$, we proceed as follows: we start with $n$ singleton nodes. 
	For each $2^{x}$-gate $h_j$ in $\cC$ we construct nodes $Q_{j,0}, \dots, Q_{j,n - 1}$. The aim is to define $\Lambda_{ Q_{j,\ell}}$ such that $\eps( Q_{j,\ell}) = \eval(h_j) \cdot 2^{\ell}$; in particular, $\eps( Q_{j,0}) = \eval(h_j)$ and $C_{j} = (Q_{j,0}, \dots, Q_{j,n - 1})$ is a chain (by a slight abuse of the notation of \cref{def:chain} since now the power circuit is not reduced).

	Let $h_j$ be some $2^x$-gate and $g$ the gate from where $h_j$ receives its input. If $g$ is a $0$-gate, we define $a_{j,i} = 0$  for all $i\in \interval{1}{k}$; if $g$ is a $2^x$-gate $h_m$, we set $a_{j,m} =1$  and $a_{j,i} = 0$ for all $i \neq m$. Otherwise, $g$ is a $+$ or $-$-gate. In this case we define $a_{j,i} = a_{g,i}$ where  $a_{g,i} \in \Z $ is as above. 
	Then for all $\ell \in \oneinterval{n}$ we define $\Lambda_{Q_{j,\ell}}$ on each chain $C_i$ such that $\binM{\Lambda_{Q_{j,\ell}}}{C_{i}}$ is the binary representation of $a_{j,i}$ (notice that $a_{j,i}$ requires only $n$ bits and all the chains $C_{i}$ for different $i$ are disjoint, so this is well-defined). Moreover, we add a $+$ edge from $Q_{j,\ell}$ to $\ell$ many of the singleton nodes. We do this for all $2^x$-gates in parallel. 
	By induction we see that, indeed, $\eps( Q_{j,\ell}) = \eval(h_j) \cdot 2^{\ell}$. 
	
	If the output gate of $\cC$ is a $2^x$-gate $h_j$, we obtain a marking evaluating to the same value by simply marking $Q_{j,0}$ with one; if the output gate is a $+$- or $-$-gate, we obtain a corresponding marking in the same fashion as for the $\Lambda_{Q_{j,0}}$ described above.
	
	Clearly, the whole computation also can be done in \Nc{2}.
	The bound $\abs{\Pi} \leq n^2 + n$ is straightforward: we introduced at most $n$ singleton nodes and then for every of the at most $n$ $2^x$-gates we introduced $n$ additional nodes. The bound on the depth is because we can have an edge from  $Q_{j,\ell}$ to $Q_{j',\ell'}$ only if there is a path from $h_j$ to $h_{j'}$ in $\cC$. Adding the edges from $Q_{j,\ell}$ to the singleton nodes only increases its depth if the depth without these edges was zero, \ie if $h_j$ is a $2^x$-gate whose input is a sum of $0$-gates. However, we counted the depth of such $2^x$-gates already as one~-- so also in this case the depth does not increase.
\end{proof}

\subsection{Power circuit reduction}\label{sec:PCreduction}

While compact markings on a reduced power circuit yield unique representations of integers, in an arbitrary power circuit $(\Pi, \delta_{\Pi})$ we can have two markings $L$ and $M$ such that $L \neq M$ but $\epsilon(L) = \epsilon(M)$.
Therefore, given an arbitrary power circuit, we wish to produce a reduced power circuit for comparing markings. This is done by the following theorem, which is our main technical result on power circuits.

\begin{theorem}\label{thm:pcred}
The following is in \PTc{0}{D} parametrized by $\depth(\Pi)$: 
	\compproblem{A power circuit $(\Pi, \delta_{\Pi})$ together with a marking $M$ on $\Pi$.}
	{A reduced power circuit $(\Gamma, \delta)$ together with a compact marking $\wt M$ on $\Gamma$ such that
		$\epsilon(\wt M) = \epsilon(M)$.
	}
\end{theorem}
For a power circuit $(\Pi, \delta_{\Pi})$ with a marking $M$ we call the power circuit $(\Gamma, \delta)$ together with the marking $\wt M$ obtained by \cref{thm:pcred} the \emph{reduced form} of $\Pi$. 

The proof of \cref{thm:pcred} consists of several steps, which we introduce on the next pages. The high-level idea is as follows: Like in \cite{MyasnikovUW12,DiekertLU13ijac}, we keep the invariant that there is an already reduced part and a non-reduced part (initially the non-reduced part is $\Pi$). The main difference is that in one iteration we insert \emph{all} the nodes of the non-reduced part that have only successors in the reduced part into the reduced part. Each iteration can be done in \TC; after $\depth(\Pi)+1$ iterations we obtain a reduced power circuit.
\vspace{-1mm}
\subparagraph*{Insertion of new nodes.} 

 The following procedure, called \proc{InsertNodes}, is a basic tool for the reduction process.
Let $(\Gamma, \delta)$ be a reduced power circuit and $I$ be a set of nodes with $\Gamma \cap I = \emptyset$. Assume that for every $P \in I$ there exists a marking $\Lambda_{P}\colon \Gamma \rightarrow \oneset{\mOne, 0, 1}$ satisfying: 
\begin{itemize}
\item $\epsilon(\Lambda_{P})\geq 0$ for all $P \in I$, 
\item $\Lambda_{P}$ is compact for all $P \in I$, and
\item $\epsilon(\Lambda_{P}) \neq \epsilon(\Lambda_{Q})$ for all $P, Q \in I \cup \Gamma$, $P \neq Q$. 
\end{itemize}

 We wish to add $I$ to the reduced power circuit $(\Gamma, \delta)$. For this, we set $\Gamma' = \Gamma \cup I$ and define $\delta'\colon \Gamma' \times \Gamma' \rightarrow \oneset{\mOne, 0, 1}$ in the obvious way: 
  $\delta'|_{\Gamma \times \Gamma} = \delta$, $\delta'|_{\Gamma' \times I} = 0$ and $\delta'(P,Q) = \Lambda_{P}(Q)$ for $(P, Q) \in I \times \Gamma$. 
Now, $(\Gamma', \delta')$ is a power circuit with $(\Gamma, \delta) \leq (\Gamma', \delta')$ and for every $P \in I$ the map $\Lambda_{P}$ is the successor marking of $P$. 
Moreover, each node of $\Gamma'$ has a unique value. 
In order to obtain a reduced power circuit, we need to sort the nodes in $\Gamma'$ according to their values:
Since for every node $P \in \Gamma'$ the marking $\Lambda_{P}$ is a compact marking on the reduced power circuit $\Gamma$, by \cref{lem:compareCompactMarkings}, for $P, Q \in \Gamma'$ we are able to decide in \Ac0 whether $\epsilon(\Lambda_{Q}) \leq \epsilon(\Lambda_{P})$. Therefore, by \cref{ex:sortTC} we can sort $\Gamma'$ according to the values of the nodes in \TC and, hence, assume that $\Gamma' = (P_{0}, \ldots, P_{\abs{\Gamma'} - 1})$ is in increasing order.

Observe that $\abs{\Gamma'} = \abs{\Gamma \cup I} = \abs{\Gamma} + \abs{I}$. In addition, inserting a new node either extends an already existing maximal chain, joins two existing maximal chains, or increases the number of maximal chains by one. Therefore, $\abs{\mathcal{C}_{\Gamma'}} \leq \abs{\mathcal{C}_{\Gamma}} + \abs{I}$. So we have proven the following:

\begin{lemma}[\proc{InsertNodes}]\label{lem: newnodes}
The following problem is in \TC:
\compproblem{A power circuit $(\Gamma,\delta)$ and a set $I$ with the properties described above.}{A reduced power circuit $(\Gamma', \delta')$ such that $(\Gamma, \delta) \leq (\Gamma', \delta')$ and such that for every $P \in I$ there is a node $Q$ in $\Gamma'$ with $\Lambda_{Q} = \Lambda_{P}$. In addition,
	\begin{itemize}
		\item $\abs{\Gamma'} = \abs{\Gamma}+\abs{I}$, and 
		\item $\abs{\mathcal{C}_{\Gamma'}} \leq \abs{\mathcal{C}_{\Gamma}}+\abs{I}$.
	\end{itemize} 
	}
\vspace{-.8cm}
\end{lemma} 

\vspace{-1mm}
\subparagraph*{The three steps of the reduction process.}

The reduction process for a  power circuit $(\Pi, \delta_{\Pi})$ with a marking $M$ consists of several iterations. Each iteration starts with a power circuit $(\Gamma_i \cup \Xi_i, \delta_i)$ such that $\Gamma_i$ is a reduced sub-power circuit and a marking $M_i$ with $\epsilon(M_i) = \epsilon(M)$. The aim of one iteration is to integrate the vertices $\Min(\Xi_i) \sse \Xi_i$ into $\Gamma_i$ where $\Min(\Xi_i)$ is defined by
	\[\Min(\Xi_i) = \set{P \in \Xi_i }{\supp(\Lambda_{P}) \subseteq \Gamma_i}\]

\noindent and to update the marking $M_i$ accordingly.
Each iteration consists of the three steps \stepone, \steptwo, and \stepthree, which can be done in \TC. We have
 $\Xi_{i+1} = \Xi_i \setminus \Min(\Xi_i)$. Thus, the full reduction process consists of $\depth(\Pi)+1$ many \TC computations.

 Let us now describe these three steps in detail and also show that they can be done in \TC. After that we present the full algorithm for power circuit reduction.

We write $(\Gamma \cup \Xi, \delta) = (\Gamma_i \cup \Xi_i, \delta_i)$ for the power circuit at the start of one iteration (for simplicity we do not write the indices).
Let us fix its precise properties: $\Gamma \cap \Xi = \emptyset$, $(\Gamma, \delta|_{\Gamma \times \Gamma}) \leq (\Gamma \cup \Xi, \delta)$ is a reduced power circuit and $\Lambda_{P}|_{\Gamma}$ is a compact marking for every $P \in \Xi$. Moreover, we assume that $\abs{C_0(\Gamma)}\geq\ceil{\log(\abs{\Xi})}+1 $.

\begin{lemma}[\stepone]\label{stz:insertminU}
The following problem is in \TC: 
\compproblem{A power circuit  $(\Gamma \cup \Xi, \delta)$ as above.}
{ A reduced power circuit $(\Gamma', \delta')$  such that
\begin{itemize}
\item $(\Gamma, \delta|_{\Gamma \times \Gamma}) \leq (\Gamma', \delta')$,
\item for every node $ Q \in \Min(\Xi)$ there exists a node $P \in \Gamma'$ with $\epsilon(P) = \epsilon(Q)$,
\item $\abs{\Gamma'}\leq \abs{\Gamma}+\abs{\Min(\Xi)}$, and 
\item $\abs{\mathcal{C}_{\Gamma'}}\leq \abs{\mathcal{C}_{\Gamma}}+\abs{\Min(\Xi)}$.
\end{itemize} 
}
\vspace{-.5cm}
\end{lemma}
For the proof, we define the following equivalence relation  $\equival$ on $\Gamma \cup \Min(\Xi)$: 
\[P \equival Q \text{ if and only if } \epsilon(P) = \epsilon(Q).\]
For $P \in \Gamma \cup \Min(\Xi)$ we write $[P]_{\epsilon}$  for the equivalence class containing $P$.

\begin{proof}
	Consider the equivalence relation $\equival$ as defined above on $\Gamma \cup \Min(\Xi)$. Define a set $I\sse \Min(\Xi)$ by taking one representative of each $\equival$-class not containing a node of $\Gamma$. Such a set $I$ can be computed in \TC: Clearly, $\Min(\Xi)$ can be computed in \TC. The $\equival$-classes can be computed in \Ac0 by \cref{lem:compareCompactMarkings}. 
	Finally, for defining $I$ one has to pick representatives. For example, for every $\equival$-class which does not contain a node of $\Gamma$ one can pick the first node in the input which belongs to this class. These representatives also can be found in \TC.
	Now, we can apply \cref{lem: newnodes} to insert $I$ into $\Gamma$ in \TC. This yields our power circuit  $(\Gamma', \delta')$. The size bounds follow now immediately from those in \cref{lem: newnodes} (notice that $\abs{I} \leq \abs{\Min(\Xi)}$).	
\end{proof}

\newcommand{\prolongate}{\mu}

\newcommand{\startGamma}{\Gamma'}
\newcommand{\innerGamma}{\tilde\Gamma}
\newcommand{\finalGamma}{\Gamma''}

\newcommand{\startDelta}{\delta'}
\newcommand{\innerDelta}{\tilde\delta}
\newcommand{\finalDelta}{\delta''}

\newcommand{\innerP}{\tilde P}
\newcommand{\innerB}{\tilde b}

\begin{lemma}[\steptwo] \label{stz:extendChains}
	The following problem is in \TC:
	\compproblem{A reduced power circuit $(\Gamma', \delta')$ and $\prolongate \in \mathbb{N}$ such that $\prolongate \leq\floor{\frac{2^{\abs{C_{0}}+1}}{3}} $ (where, as before, $C_0 = C_0(\Gamma')$ is the initial maximal chain of $\Gamma'$)}{ A reduced power circuit $(\Gamma'',\delta'')$ such that 
		\begin{itemize}
			\item $(\Gamma', \delta') \leq (\Gamma'',\delta'')$, 
			\item for each $P \in \Gamma'$ and each $i \in \interval{0}{\prolongate} $ there is a node $Q \in \Gamma''$ with $\epsilon(\Lambda_{Q}) = \epsilon(\Lambda_{P})+i$,
			\item $\abssmall{\Gamma''} \leq \abs{\Gamma'}+\abs{\mathcal{C}_{\Gamma'}}\cdot\prolongate$, and
			\item $\abssmall{\mathcal{C}_{\Gamma''}} \leq \abs{\mathcal{C}_{\Gamma'}} $.
	\end{itemize}}\vspace{-.5cm}
\end{lemma}

\begin{proof}
	First assume that $\abs{C_{0}} = 1$. Then $\abs{\startGamma} = 1$ and $\prolongate \leq 1$. If $\prolongate = 1$, then just one node has to be created, namely the one of value $2$ and we are done. Thus, in the following we can assume that $\abs{C_{0}} \geq 2$. 
	Now, the proof of \cref{stz:extendChains} consists of two steps: first, we extend only the chain $C_0$ to some longer (and long enough) chain in order to make sure that the values of the (compact) successor markings of the nodes we wish to introduce can be represented within the power circuit; only afterwards we add the new nodes as described in the lemma.	
	
	\proofsubparagraph{Step 1:}
	We first want to extend the chain $C_{0}$ to the chain $\tilde{C}_{0}$ of minimal length such that $ \tilde{C}_{0}$ is a maximal chain, $C_{0} \subseteq \tilde{C}_{0}$, and the last node of $\tilde{C}_{0}$ is not already present in $\startGamma$. The resulting power circuit will be denoted by $\innerGamma$. We define
	\[i_{0} = \min \set{ i \in\oneinterval{\abssmall{\startGamma}}}{\epsilon(\Lambda_{P_{i+1}})-\epsilon(\Lambda_{P_{i}}) > 2}.\]
	Here, we use the convention that $P_{\abs{\startGamma}}$ has value infinity, so $i_0$ indeed exists. 
	Furthermore, we define  
	\begin{align*}
		I & = \set{ i \in \interval{0}{i_0}}{ \epsilon(\Lambda_{P_{i+1}}) - \epsilon(\Lambda_{P_{i}}) \geq 2}.
	\end{align*}
	Thus, in order to obtain $\innerGamma$, we need to insert a new node between $P_{i}$ and $P_{i+1}$ into $\startGamma$ for each $i \in I$ (resp.\ one node above $P_{i_{0}}$). Since the successor markings of these new nodes might point to some of the other new nodes, we cannot apply \cref{lem: newnodes} as a black-box. Instead, we need to take some more care: the rough idea is that, first, we compute all positions $I$ where new nodes need to be introduced ($I$ is as defined above), then we compute \csdr{}s for the respective successor markings, and, finally, we introduce these new nodes all at once knowing that all nodes where the successor markings point to are also introduced at the same time. In order to map the positions of nodes in $\startGamma$ to positions of nodes in $\innerGamma$, we introduce a function $\lambda\colon \oneinterval{\abssmall{\startGamma}} \to \N$ with 
	\begin{align*}
		\lambda(i) & = i + \abs{I \cap \interval{0}{i-1}}.
	\end{align*}
	Observe that $\lambda(i) = i$ for $i \in \oneinterval{\abs{C_{0}}}$, and $\lam(i+1) = \lam(i) + 2$ for $i \in I$, and $\lambda(j) = j + \abs{I}$ for $j \geq i_{0}+1$.
	
	For each $i \in I$ we introduce a node $Q_i$ whose successor marking we will specify later such that $\epsilon(Q_{i}) = 2\epsilon(P_{i})$.  We define the new power circuit $\innerGamma = (\innerP_{0}, \ldots, \innerP_{\abs{\startGamma}+ \abs{I}-1})$ by 
	\[ \innerP_j = \begin{cases}
		P_i & \text{if } j = \lambda(i)  \\
		Q_i & \text{if } j = \lambda(i) + 1 \text{ and } i \in I. 
	\end{cases} \]
	Notice that, if $j = \lambda(i) + 1$ for some  $i \in I$, then $j \neq \lambda(i)$ for any $i$~-- hence, $\innerP_j$ is well-defined in any case.
	
	The nodes $\innerP_0, \ldots, \innerP_{\lambda(i_{0}) + 1} $ will form the chain $\tilde{C}_{0}$ as claimed above. Moreover, we have $\startGamma \subseteq \innerGamma$ and $\innerGamma$ is sorted increasingly. The successor markings of nodes from $\startGamma$ remain unchanged (\ie $\Lambda_{\innerP_{\lambda(i)}}\! (\innerP_{\lambda(j)}) = \Lambda_{P_i}(P_j)$ for $i, j \in \oneinterval{\abs{\startGamma}}$ and $\Lambda_{\innerP_{\lambda(i)}}(Q_j) = 0$ for $j \in I$).

	For every $i \in I$ we define the successor marking of the node $Q_{i}$ by
	\[\binM{\Lambda_{ Q_{i}}}{\tilde{C}_{0}} = \cor\left(\eps(\Lambda_{ P_i})+1\right)\quad \text{and}\quad \Lambda_{ Q_{i}}|_{\innerGamma \setminus \tilde{C}_{0}} = 0.\]
	Be aware that, since $Q_{i} \in \tilde C_0$, also the successor marking of $Q_{i}$ (of value $\eps(\Lambda_{ P_{i}})+1$) can be represented using only the nodes from $\tilde C_0$ (see \cref{rem:toposort}), so this is, indeed, a meaningful definition (be aware that to represent $\eps(\Lambda_{ P_{i}})+1$, we might need some of the additional nodes $Q_i$, but never a node that is not part of the chain $\tilde C_0$). Clearly, this yields $\epsilon(\Lambda_{Q_{i}}) = \epsilon(\Lambda_{P_{i}})+1$ as desired.
	
	We obtain a reduced power circuit $(\innerGamma, \innerDelta)$ with $(\startGamma, \startDelta) \leq (\innerGamma, \innerDelta)$ where the map $\innerDelta\colon \innerGamma \rightarrow \oneset{\mOne, 0, 1}$ is defined by the successor markings. 
	Moreover, $\tilde{C}_{0} \sse \innerGamma$ has the required properties.

	It remains to show that $\innerGamma$ can be computed in \TC: As $\abs{C_{0}} \geq 2$, according to  \cref{lem:compareCompactMarkings}, we are able to decide in \Ac0 whether the markings $\Lambda_{P_{i}}$ and $\Lambda_{P_{i+1}}$ differ by $1$, $2$, or more than 2~-- for all $i \in\oneinterval{\abs{\startGamma}}$ in parallel. Now, $i_{0}$ can be determined in \TC via its definition as above. Likewise $I$ and the function $\lambda$ can be computed in \TC.
	 By \cref{cor:signedDigitOps}, $\cor\left(\eps(\Lambda_{ P_i})+1\right)$ for $i \in I$ can be computed in \Ac0 (since $\abssmall{\tilde{C}_{0}} \leq 2 \cdot \abs{\Gamma'}$) showing that altogether $\innerGamma$ can be computed in \TC.

	\proofsubparagraph{Step 2:}	The second step is to add nodes above each chain of $\innerGamma$ as required in the Lemma. The outcome will be denoted by $(\finalGamma, \finalDelta)$. We start by defining
	\begin{align*}
	&&	d_{i}	& = \min\{\epsilon(\Lambda_{\innerP_{i+1}}) - \epsilon(\Lambda_{\innerP_{i}})-1, \prolongate\} &&\text{ for } 	i\in \oneinterval{\abssmall{\innerGamma}} \setminus \bigl\{ \abssmall{\tilde C_0 }-1\bigr\} \quad \text{and}\\
	&&	d_{i}	& = \min\{\epsilon(\Lambda_{\innerP_{i+1}}) - \epsilon(\Lambda_{\innerP_{i}})-1, \prolongate -1\} && \text{ for } i = \abssmall{\tilde C_0 } - 1. 
	\end{align*}
	In order to obtain $(\finalGamma, \finalDelta)$ from $(\innerGamma,\innerDelta)$,
	 for every $i \in \oneinterval{\abssmall{\innerGamma}}$ and every  $h \in \interval{1}{d_i}$ we have to insert a node $R^{(i,h)}$ such that \[\epsilon(\Lambda_{R^{(i,h)}}) = \epsilon(\Lambda_{\innerP_{i}})+h.\]
	Observe that the numbers $d_i$ can be computed in \TC: since 
	\[\prolongate + 1 \leq \floor{\frac{2^{\abs{C_{0}}+1}}{3}} + 1 \leq \floor{\frac{2^{\abssmall{\tilde C_0 }}}{3}} + 1 \leq \floor{\frac{2^{\abssmall{\tilde C_0 }+1}}{3}},\] by \cref{lem:compareCompactMarkings}, we can check in \Ac0 whether $ \epsilon(\Lambda_{\innerP_{i+1}}) \leq \epsilon(\Lambda_{\innerP_{i}})+k$ with $k \leq \prolongate +1$. If $i = \abssmall{\tilde C_{0}}-1$ we choose $k = \prolongate$, otherwise $k = \prolongate+1$. If the respective inequality holds, we obtain by \cref{stz: smalldiffs} that $ \epsilon(\Lambda_{\innerP_{i+1}}) - \epsilon(\Lambda_{\innerP_{i}}) - 1 = \epsilon(\Lambda_{\innerP_{i+1}}|_{\tilde{C}_{0}}) - \epsilon(\Lambda_{\innerP_{i}}|_{\tilde{C}_{0}})-1$. For the latter we have \sdr{}s of digit-length at most $\abssmall{\tilde{C}_{0}}$. Hence, this difference can be computed in \TC.

	Since $\innerP_{\abssmall{\tilde{C}_{0}}-1 } \not \in \startGamma$ and in Step 1 we have not introduced any vertex above $\innerP_{\abssmall{\tilde{C}_{0}}-1 } $, we know that $\innerP_{\abssmall{\tilde{C}_{0}}-1 }$ is not marked by $\Lambda_{\innerP}$ for any $ \innerP \in \innerGamma$. Therefore,  for all $i\in \oneinterval{\abssmall{\innerGamma}}$ we have $\eps(\Lambda_{\innerP_{i}}|_{\tilde{C}_{0}}) + \prolongate \leq \floor{\frac{2^{\abssmall{\tilde{C}_{0}}}}{3}} + \floor{\frac{2^{\abs{C_{0}}+1}}{3}} \leq2\floor{\frac{2^{\abssmall{\tilde{C}_{0}}}}{3}}$ and, hence, by \cref{lem:valcom}, $\eps(\Lambda_{\innerP_{i}}|_{\tilde{C}_{0}})+h$ can be represented as a compact marking using only nodes from $\tilde{C}_{0}$ for every $h \in \interval{1}{d_i}$.
	Thus, for every $d_{i} \neq 0$ and every  $h \in \interval{1}{d_i}$ we define a successor marking of $R^{(i,h)}$ by
	\[\binM{\Lambda_{R^{(i,h)}}}{\tilde{C}_{0}} = \cor(\eps(\Lambda_{\innerP_{i}}|_{\tilde{C}_{0}})+h) \quad	
	\text{ and }\quad \Lambda_{R^{(i,h)}}|_{\innerGamma \setminus \tilde{C}_{0}} = \Lambda_{\innerP_{i}}|_{\innerGamma \setminus \tilde{C}_{0}}.\]  
	Again, we know that $\abssmall{\tilde{C}_{0}} \leq 2\abs{\startGamma}$. So according to \cref{cor:signedDigitOps} we are able to calculate $\cor(\eps(\Lambda_{\innerP_{i}}|_{\tilde{C}_{0}})+h) $ in \Ac0.

	Now we set $I = \set{R^{(i,h)}}{d_{i}\neq 0, h \in \interval{1}{d_{i}} }$. According to \cref{lem: newnodes} we are able to construct in \TC a reduced power circuit $(\finalGamma, \finalDelta)$ such that $(\innerGamma, \innerDelta) \leq (\finalGamma, \finalDelta)$ and such that for each $R \in I$ there exists a node $Q \in \finalGamma$ with $\epsilon(Q) = \epsilon(R)$. 
	
	Considering the size of $\finalGamma$, observe that during the whole construction, for every node $P_{i} \in \startGamma$ we create at most $\prolongate$ new nodes between $P_i$ and $P_{i+1}$. 
	Moreover, we only create new nodes between $P_{i}$ and $P_{i+1}$ if $P_{i}$ is the last node of a maximal chain of $\startGamma$. 
	 Furthermore, notice that the only node of $\startGamma$ above which we have introduced new nodes in both Step 1 and Step 2 is the second largest node of $\tilde{C}_{0}$: in Step 1 we have created one new node and in Step 2 we have created at most $\prolongate - 1$ new nodes above it. Thus, for every chain of $\startGamma$ we have introduced at most $\prolongate$ new nodes.
	Thus, $\abssmall{\finalGamma} \leq \abs{\startGamma}+\abs{\mathcal{C}_{\startGamma}}\cdot \prolongate$. Finally, the new nodes we create only prolongate the already existing chains, so we do not create any new chains. This finishes the proof of the lemma. 
\end{proof}

In the following, $(\Gamma', \delta')$ denotes the power circuit obtained by \stepone when starting with $(\Gamma \cup \Xi, \delta)$, and
 $(\Gamma'', \delta'')$ denotes the power circuit obtained by \steptwo with $\prolongate=\ceil{\log(\abs{\Min(\Xi)})}+1$  on input of the power circuit $(\Gamma', \delta')$ (observe that, by the assumption $\abs{C_0(\Gamma)}\geq\ceil{\log(\abs{\Xi})}+1 $, the condition on $\mu$ in \cref{stz:extendChains} is satisfied).
The value of $\mu$ is chosen to make sure that in the following lemma one can make the markings compact. Indeed, if $ \Min(\Xi) = \{ P_1, \dots, P_{k} \}$ and all $P_i$ have the same evaluation and are marked with $1$ by $M$, then we might need a node of value $2^\mu\cdot \eps(P_1)$ in order to make $M$ compact.

\begin{lemma}[\stepthree] \label{stz:updateM}
The following problem is in \TC:
\compproblem{The power circuit $(\Gamma'', \delta'')$ as a result of \steptwo{} with $\prolongate = \ceil{\log(\abs{\Min(\Xi)})}+1$ and a marking $M$ on $\Gamma \cup \Xi$. }{
 A marking $\wt M$ on $\Gamma'' \cup (\Xi \setminus \Min(\Xi))$ such that $\epsilon(\wt M) = \epsilon(M)$ and $\wt M|_{\Gamma''}$ is compact.}
\end{lemma}

\begin{proof}
	Consider again the equivalence relation $\equival$ as defined above on $\Gamma'' \cup \Min(\Xi)$. For the equivalence class of a node $P \in \Gamma'' \cup \Min(\Xi)$ we write $[P]_\eps$.
 We will define the marking $\wt M$ on $\Gamma''$ by defining it on each maximal chain. Recall that we can view $M$ as a marking on $\Gamma'' \cup \Xi$ by defining $M(P) = 0$ if $P \not\in\Gamma \cup \Xi$.

Let $C = (P_{i}, \ldots, P_{i+h-1}) \in \mathcal{C}_{\Gamma''}$ be a maximal chain of length $h$ and let \[S = \bigcup_{P \in C}[P]_\eps = \bigcup_{P \in C} \set{Q \in \Gamma'' \cup \Min(\Xi)}{\eps(Q) = \eps(P)}\; \sse\; \Gamma'' \cup \Min(\Xi).\] 
We wish to find a compact marking $\tilde M_C$ with support contained in $C \sse \Gamma''$ and evaluation $\eps(\tilde M_C) = \eps(M|_S)$. First define the integer
\begin{align*}
	Z_{M,C} = \sum_{r = 0}^{h-1}\left( \sum_{Q \in [P_{i+r}]_\eps} M(Q)\right)2^{r}.
\end{align*}
Then we have
\begin{align*}
	Z_{M,C}\cdot \epsilon(\startc{C})& = \sum_{r = 0}^{h-1} \,\sum_{Q \in [P_{i+r}]_\eps} M(Q) 2^{r}\cdot \epsilon(\startc{C})\\[.3em]
	& = \sum_{Q \in S} M(Q) \eps(Q)\\[.3em]
	& = \epsilon(M|_{S}).
\end{align*}
Thus, defining $\tilde M_C$ by $\binM{\tilde M_C}{C} = \cor(Z_{M,C})$ gives our desired marking.

However, be aware that, for this, we have to show that the digit-length of $\cor(Z_{M,C})$ is at most $\abs{C} = h$. Let $k$ be maximal such that $P_{i+k} \in \Gamma'$. 
Then, in particular, no node in $S$ with higher evaluation than $P_{i+k}$ is marked by $M$. Moreover, by the properties of \steptwo{}$(\ceil{\log(\abs{\Min(\Xi)})}+1)$, we have $h - 1- k  \geq \ceil{\log(\abs{\Min(\Xi)})}+1$.	
Therefore, 
\begin{align*}
	Z_{M,C}	&\leq \val(\binM{M}{C}) + \abs{\Min(\Xi)} \cdot 2^{k} \\
			&\leq \frac{1}{3} \cdot 2^{k+2} + 2^{k+\log(\abs{\Min(\Xi)})}          \tag{by \cref{lem:valcom}}\\
			&\leq \frac{4}{3} \cdot \left(2^{k} + 2^{k+\log(\abs{\Min(\Xi)})}\right)\\
			&\leq \frac{2}{3} \cdot 2^{k + \ceil{\log(\abs{\Min(\Xi)})} + 2}.
\end{align*}
Thus, by \cref{lem:valcom}, the digit-length\ of $\cor(Z_{M,C})$ is at most $k + \ceil{\log(\abs{\Min(\Xi)})} + 2 \leq h$.

By \cref{cor:findChains}, the maximal chains can be determined in \TC.
Now, for every maximal chain $C$ the (binary) number $Z_{M,C}$ can be computed in \TC using iterated addition and made be compact in \Ac0 using \cref{thm:compact}. Thus, the marking $\tilde M_C$ can be computed in \TC.
The marking $\wt M$ as desired in the lemma is simply defined by $\wt M|_{\Xi\setminus\Min(\Xi)} = M|_{\Xi\setminus\Min(\Xi)}$ and  $\wt M|_C = \tilde M_C|_C$ for $C \in \mathcal{C}_{\Gamma''}$  ~-- all the markings $\tilde M_C$ can be computed in parallel.
\end{proof}

\begin{proof}[Proof of \cref{thm:pcred}.]
 Now we are ready to describe the full reduction process based on the three steps described above. We aim for a \PTc{0}{D} circuit where the input is parametrized by the depth of the power circuit.
The input is some arbitrary power circuit $(\Pi, \delta_{\Pi})$ together with a marking $M$ on $\Pi$. We start with some initial reduced power circuit $(\Gamma_{0}, \delta_{0})$ and some non-reduced part $\Xi_0 = \Pi$ and successively apply the three steps to obtain power circuits $(\Gamma_{i} \cup \Xi_{i}, \delta_{i})$ and markings $M_i$ for $i = 0, 1\dots$ while keeping the following invariants:
\begin{itemize}
\item  $(\Gamma_i, \delta_i|_{\Gamma_i \times \Gamma_i}) \leq (\Gamma_{i} \cup \Xi_{i}, \delta_{i})$ (\ie there are no edges from $\Gamma_i$ to $\Xi_i$),
\item $\Gamma_i$ is reduced,
\item $\Gamma_{i-1} \leq \Gamma_{i}$ and $\Xi_{i} \sse \Xi_{i-1}$,
\item $\eps(M_i) = \eps(M)$.
\end{itemize}
 Moreover, as long as $\Xi_{i-1} \neq \emptyset$ we assure that $\depth(\Xi_i) < \depth(\Xi_{i-1})$.

We first construct the initial reduced power circuit $(\Gamma_{0}, \tilde \delta_{0})$ which consists exactly of a chain of length $\ell = \ceil{\log(\abs{\Pi})}+1$. This can be done as follows:
Let $\Gamma_{0} = (P_{0}, \ldots, P_{\ell-1}) = C_0$ 
and define successor markings by $\binM{\Lambda_{P_{i}}}{C_{0}} = \cor(i)$ for $i \in \oneinterval{\ell}$. This defines $\tilde \delta_0$.  
Now we set $\Xi_{0} = \Pi$ and we define $\delta_{0} \colon  (\Gamma_{0} \cup \Xi_{0}) \times (\Gamma_{0} \cup \Xi_{0}) \rightarrow \{\mOne, 0, 1\}$ by $\delta_{0}|_{\Gamma_{0} \times \Gamma_{0}} = \tilde{\delta}_{0}$, $\delta_{0}|_{\Xi_{0} \times \Xi_{0}} = \delta_{\Pi}$ and $\delta = 0$ otherwise. We extend the marking $M$ to $\Gamma_{0} $ by setting $M(P) = 0$ for all $P \in \Gamma_{0} $. So we obtain a power circuit of the form $(\Gamma_{0}\cup \Xi_{0}, \delta_{0})$ with the properties described above.  

Now let the power circuit $(\Gamma_{i} \cup \Xi_{i}, \delta_{i})$ together with the marking $M_{i}$ be the input for the $i+1$-th iteration meeting the above described invariants. We write $\tilde{\delta}_{i} = \delta_{i}|_{\Gamma_{i} \times \Gamma_{i}}$. Now we apply the three steps from above:
\begin{enumerate}
	\item Using \stepone (\cref{stz:insertminU}) we compute a reduced power circuit $(\Gamma'_{i}, \delta'_{i})$ with $(\Gamma_{i}, \tilde{\delta}_{i}) \leq (\Gamma'_{i}, \delta'_{i})$ such that for every $P \in \Min(\Xi_{i})$ there is some $Q \in \Gamma'_{i}$ with $\epsilon(Q) = \epsilon(P)$. 
	\item Using \steptwo (\cref{stz:extendChains}) with $\prolongate = \ceil{\log(\abs{\Min(\Xi_{i})})}+1$ we extend each maximal chain in $(\Gamma'_{i}, \delta'_{i})$ by at most $\ceil{\log(\abs{\Min(\Xi_{i})})}+1$ nodes. Notice that $\ceil{\log(\abs{\Min(\Xi_{i})})}+1 \leq \ceil{\log(\abs{\Pi})}+1$ and so,  as $\Gamma_0 \leq \Gamma_i'$, the condition $\prolongate \leq \floor{\frac{2^{\abssmall{C_0(\Gamma_i')}+1}}{3}}$ in \cref{stz:extendChains} is satisfied.	
	The result of this step is denoted by $(\Gamma''_{i}, \delta''_{i})$. 
	\item We apply \stepthree (\cref{stz:updateM}) to obtain markings $\wt M_{i}$ and $\tilde{\Lambda}_{P}$ for $P \in \Xi_{i} \setminus \Min(\Xi_{i})$ on $\Gamma''_{i} \cup (\Xi_{i} \setminus \Min(\Xi_{i}))$ such that $\epsilon(\wt M_{i}) = \epsilon(M_{i})$ and $\epsilon(\tilde{\Lambda}_{P}) = \epsilon(\Lambda_{P})$. 
	Observe that these markings restricted to $\Gamma''_{i}$ are compact.
	
	\item 
	Each iteration ends by setting $\Gamma_{i+1} = \Gamma''_{i}$ and $\Xi_{i+1} = \Xi_{i} \setminus \Min(\Xi_{i})$ and $M_{i+1} = \wt M_{i}$. Finally, $\delta_{i+1}$ is defined as $\delta_i''$ on $\Gamma_{i+1}$ and via the successor markings $\tilde{\Lambda}_{P}$ for $P \in \Xi_{i+1}$.
\end{enumerate}

After exactly $\depth(\Pi) + 1$ iterations we reach $\Xi_{d+1} = \Xi_{d} \setminus \Min(\Xi_{d}) = \emptyset$ where $d = \depth(\Pi)$. In this case we do not change the resulting power circuit any further.
It is clear from \cref{stz:insertminU}, \cref{stz:extendChains} and \cref{stz:updateM} that throughout the above-mentioned invariants are maintained.
  Thus, $(\Gamma, \delta) = (\Gamma_{d +1 } ,\delta_{d + 1})$ is a reduced power circuit and for every node $P \in \Pi$ there exists a node $Q \in \Gamma_{d+1}$ such that $\epsilon(Q) = \epsilon(P)$ and  $M_{d+1}$ is a compact marking on $\Gamma_{d+1}$ with $\epsilon(M_{d+1}) = \epsilon(M)$. 
 
\begin{claim}\label{lem:pcredsize}
Let $d = \depth(\Pi)$ and $\Gamma_{0}, \dots, \Gamma_{d+1}$ be as constructed above. Then for all $i$ we have
		\begin{itemize}
			\item $\abs{\mathcal{C}_{\Gamma_{i}}} \leq \abs{\Pi}+1$,  
			\item $\abs{\Gamma_{i}}  \leq (\abs{\Pi}+1)^{2}\cdot (\log(\abs{\Pi})+2)$.
		\end{itemize}
\end{claim}

\begin{claimproof}
According to \cref{stz:insertminU} and \cref{stz:extendChains} we have 
	$\abs{\mathcal{C}_{\Gamma_{i+1}}} \leq \abs{\mathcal{C}_{\Gamma_{i}}} + \abs{\Min(\Xi_{i})}$. Further observe that $\Pi$ is the disjoint union of the $ \Min(\Xi_{j})$ for $j \in \interval{0}{d}$.  
 Since $\abs{\mathcal{C}_{\Gamma_{0}}} = 1$, we obtain for all $i \in \interval{0}{d}$ that
\begin{align}
	\abs{\mathcal{C}_{\Gamma_{i+1}}} \leq \abs{\mathcal{C}_{\Gamma_{i}}} + \abs{\Min(\Xi_{i})} \leq 1+\sum_{0 \leq j \leq i} \abs{\Min(\Xi_{j})} \leq \abs{\Pi}+1. \label{eq:chainbound}
\end{align}
Again by \cref{stz:insertminU} and \cref{stz:extendChains} we have
\begin{align*}
	\abs{\Gamma_{i+1}} &\leq \abs{\Gamma'_{i}} + \absbig{C_{\Gamma'_{i}}} \cdot \left( \ceil{\log(\abs{\Min(\Xi_{i})})} + 1 \right) \tag{by \cref{stz:extendChains}}\\
	 &\leq \abs{\Gamma_{i}} + \abs{\Min(\Xi_{i})} + \left(\abssmall{\mathcal{C}_{\Gamma_{i}}} + \abs{\Min(\Xi_{i})} \right) \cdot \left( \ceil{\log(\abs{\Min(\Xi_{i})})} + 1 \right) \tag{by \cref{stz:insertminU}}\\
	 &\leq \abs{\Gamma_{i}} + \abs{\Min(\Xi_{i})} + (\abs{\Pi}+1) \cdot \left( \ceil{\log(\abs{\Pi})} + 1 \right). \tag{by (\ref{eq:chainbound})}
\end{align*}

Since $\abs{\Gamma_0} = \ceil{\log(\abs{\Pi})} +1$, we obtain by induction that 
\begin{align*}
	\abs{\Gamma_{i}} &\leq  \abs{\Gamma_0}+\sum_{0 \leq j \leq i-1} \abs{\Min(\Xi_{j})} + i  \cdot(\abs{\Pi}+1)\cdot (\log(\abs{\Pi})+2) \\
	& \leq (\ceil{\log(\abs{\Pi})}+1)+\abs{\Pi}+ i  \cdot(\abs{\Pi}+1)\cdot (\log(\abs{\Pi})+2) \\
	& \leq (i+1)\cdot(\abs{\Pi}+1)\cdot (\log(\abs{\Pi})+2) 
\end{align*}
for all $i \in \interval{1}{d+1}$.
The last inequality is due to the fact that $\abs{\Pi}+1 \geq 2$ and $\log(\abs{\Pi})+2 \geq 2$. Since $d+1 \leq \abs{\Pi}$, we obtain $\abs{\Gamma_{i}}  \leq (\abs{\Pi}+1)^{2}\cdot (\log(\abs{\Pi})+2)$.
\end{claimproof}

Let $D\in \N$ and assume that $\depth(\Pi) \leq D$. 
By \cref{stz:insertminU}, \cref{stz:extendChains} and \cref{stz:updateM} each iteration of the three steps above can be done in \TC.
Notice here that the construction of the markings $\wt M_{i}$ and $\tilde{\Lambda}_{P}$ during \stepthree can be done in parallel~-- so it is in \TC, although \cref{stz:updateM} is stated only for a single marking.
Now, the crucial observation is that, due to \cref{lem:pcredsize}, the input size for each iteration is polynomial in the original input size of $(\Pi, \delta_{\Pi})$. Therefore, we can compose the individual iterations and obtain a circuit of polynomial size and depth bounded by $\Oh(D)$ as described in \cref{ex:composeNC1}.
Thus, we have described a \PTc{0}{D} circuit (parametrized by $\depth(\Pi)$) for the problem of computing a reduced form for $(\Pi, \delta_{\Pi})$.
This completes the proof of \cref{thm:pcred}.
\end{proof}

\begin{remark}\label{rem:startchain}
\begin{enumerate}[(1)]
	\item While \cref{thm:pcred} is only stated for one input marking, the construction works within the same complexity bounds for any number of markings on $(\Pi, \delta_{\Pi})$ since during \stepthree these all can be updated in parallel.
\item \label{startmaxchain} Moreover, note that for every  maximal chain $C \in \mathcal{C}_{\Gamma}$ there exists a node $Q \in \Pi$ (\ie in the original power circuit) such that $\epsilon(Q) = \epsilon(\startc{C})$. This is because new chains are only created during \stepone, the other steps only extend already existing chains. 
\item \label{sizeCompactM} Further observe that $\abssmall{\supp(\wt M)} \leq \abs{\supp(M)}$. Looking at the construction of $\wt M$ we see that we first make sure that $M$ does not mark two nodes of the same value, then we make the marking compact. Both operations do not increase the number of nodes in the support of the marking. 
\end{enumerate}
\end{remark}

\begin{example}\label{ExStepsDescr}
	In \cref{ExSteps} we illustrate what happens in the steps \stepone, \steptwo and \stepthree during the reduction process. Picture a) shows our starting situation. In b) we already inserted the nodes of value $2^{3}$ and $2^{32}$ into the reduced part. Now the reduced part consists of three chains: one starting at the node of value $1$ and the nodes $2^{5}$ and $2^{32}$ as chains of length $1$.  Because $\abs{\Min(\Xi)}=3$, we have to extend each chain by three nodes or until two chains merge. So in c) we obtain two chains, one from $1$ to $2^{8}$ and the one from $2^{32}$ to $2^{35}$. In d) we then updated the markings and discarded the nodes from $\Min(\Xi)$. 
\end{example}
\begin{figure}
	\vspace*{-0.2cm}
	\begin{minipage}[b]{.45\linewidth}			
		\begin{center}
			\tikzstyle{pcnode}=[minimum size= 12pt,circle,draw ]
			\begin{tikzpicture}[scale=0.9, outer sep=0pt, inner sep = 0.7pt, node distance=1cm]
				{
					\node[pcnode, blue] (0) at (0,0) {\tiny 1};
					\node[pcnode, blue] (1) [right of=0] {\tiny 2};
					\node[pcnode, blue] (2) [right of=1] {\tiny 4};
					\node[pcnode, blue] (3) [right of=2] {\tiny $2^{5}$};
					\node[pcnode, cyan] (4) [above right of=0, yshift=10] {\tiny $2^{3}$};
					\node[pcnode, cyan] (5) [right of=4] {\tiny $2^{3}$};
					\node[pcnode] (6) [above of=5, yshift=-10] {};
					\node[pcnode, cyan] (7) [above of=3, yshift=0] {\tiny $2^{\text{\scalebox{.7}{32}}}$};

				\node[purple] (17) [above left of =5, yshift=-20, xshift=12]{\tiny +} ;
					\node[purple] (18) [above left of =6, yshift=-15, xshift=10]{\tiny +} ;
					\node[purple] (19) [above left of =4, yshift=-20, xshift=12]{\tiny +} ;
					\node[purple] (20) [below of =1, yshift=20, xshift=0]{\tiny +} ;

					\draw[->] (1) edge node[above, yshift=2] {\tiny+} (0)
					(2) edge node[above, yshift=2 ] {\tiny+} (1)
					(3) edge [bend left=30]node[below, xshift=-25, yshift=3] {\tiny+} (0)
					(3) edge node[above, yshift=2 ] {\tiny+} (2)
					(4) edge node[right, xshift=3 ] {\tiny+} (2)
					(4) edge node[left, xshift=-5 ] {\tiny $-$} (0)
					(5) edge node[left, xshift=0, yshift=1] {\tiny $+$} (0)
					(5) edge node[right] {\tiny $-$} (2)
					(6) edge[bend left=20] node[right] {\tiny $+$} (2)
					(6) edge node[left, xshift=0, yshift=3] { \tiny $+$} (4)
					(7) edge node[right] {\tiny $+$} (3)
					;
				}
			\end{tikzpicture}
			
			a) \small Starting situation
		\end{center}
		
	\end{minipage}
	\hspace*{1cm}
	\begin{minipage}[b]{.45\linewidth}	
		\begin{center}
			\tikzstyle{pcnode}=[minimum size= 12pt,circle,draw ]
			\begin{tikzpicture}[scale=0.9, outer sep=0pt, inner sep = 0.7pt, node distance=1cm]
				{
					
					\node[pcnode, blue] (0) at (0,0) {\tiny 1};
					\node[pcnode, blue] (1) [right of=0] {\tiny 2};
					\node[pcnode, blue] (2) [right of=1] {\tiny 4};
					\node[pcnode, blue] (4) [right of=2] {\tiny $2^{3}$};
					\node[pcnode, blue] (3) [right of=4] {\tiny $2^{5}$};

					\node[pcnode, cyan] (5) [above right of=0, yshift=10] {\tiny $2^{3}$};
					\node[pcnode, cyan] (6) [right of=5, yshift=0] {\tiny $2^{3}$};
					\node[pcnode] (7) [above of=6, yshift=-10] {};
					\node[pcnode, cyan] (8) [above of=3, yshift=0] {\tiny $2^{\text{\scalebox{.7}{32}}}$};
					\node[pcnode, blue] (9) [right of=3, yshift=0] {\tiny $2^{\text{\scalebox{.7}{32}}}$};

					\node[purple] (17) [above left of =5, yshift=-20, xshift=12]{\tiny +} ;
					\node[purple] (18) [above left of =6, yshift=-20, xshift=12]{\tiny +} ;
					\node[purple] (19) [above left of =7, yshift=-15, xshift=10]{\tiny +} ;
					\node[purple] (20) [below of =1, yshift=20, xshift=0]{\tiny +} ;

					\draw[->] (1) edge node[above, yshift=2] {\tiny+} (0)
					(2) edge node[above, yshift=2 ] {\tiny+} (1)
					(3) edge [bend left=35]node[below, xshift=-25, yshift=3] {\tiny+} (0)
					(3) edge[bend right=30] node[above, yshift=2 ] {\tiny+} (2)
					(4) edge node[above, yshift=2 ] {\tiny+} (2)
					(4) edge [bend left=30] node[below, yshift=0, xshift=10 ] {\tiny $-$} (0)
					
					(5) edge node[right, xshift=5 ] {\tiny+} (2)
					(5) edge node[above ] {\tiny $-$} (0)
					(6) edge node[above] {\tiny $-$} (0)
					(6) edge node[right] {\tiny $+$} (2)
					(7) edge[bend left=35] node[right] {\tiny $+$} (2)
					(7) edge node[left, xshift=0, yshift=3] { \tiny $+$} (5)
					(8) edge node[right] {\tiny $+$} (3)
					(9) edge node[above] {\tiny $+$} (3)
					
					;
				}
			\end{tikzpicture}
			
			b) \small After \stepone	
		\end{center}
		
	\end{minipage}
	\vspace*{0.2cm}
	
	\begin{minipage}[t]{.45\linewidth}	
		\begin{center}
			\tikzstyle{pcnode}=[minimum size= 12pt,circle,draw ]
			\begin{tikzpicture}[scale=0.7, outer sep=0pt, inner sep = 0.7pt, node distance=1.2cm]
				{
					\node[pcnode, blue] (0) at (0,0) {\tiny 1};
					\node[pcnode, blue] (1) [right of=0, xshift=-10] {\tiny 2};
					\node[pcnode, blue] (2) [right of=1, xshift=-10] {\tiny 4};
					\node[pcnode, blue] (4) [right of=2, xshift=-10] {\tiny $2^{3}$};
					\node[pcnode, blue] (3) [right of=4] {\tiny $2^{5}$};

					\node[pcnode, blue] (7) [right of=3] {\tiny $2^{8}$};
					
					\node[pcnode, cyan] (8) [above of=1] {\tiny $2^{3}$};
					\node[pcnode, cyan] (9) [above of=2] {\tiny $2^{3}$};
					\node[pcnode] (10) [above right of=8, yshift=-6, xshift=2] {};
					\node[pcnode, blue] (11) [below of=4, yshift=10] {\tiny $2^{\text{\scalebox{.7}{32}}}$};
					\node[pcnode, blue] (12) [right of=11, xshift=-10] {\tiny $2^{\text{\scalebox{.7}{33}}}$};
					\node[pcnode, blue] (13) [right of=12, xshift=-10] {\tiny $2^{\text{\scalebox{.7}{34}}}$};
					\node[pcnode, blue] (14) [right of=13, xshift=-10] {\tiny $2^{\text{\scalebox{.7}{35}}}$};
					\node[pcnode, cyan] (15) [above of=3] {\tiny $2^{\text{\scalebox{.7}{32}}}$};
					
					\node[purple] (17) [above left of =8, yshift=-23, xshift=15]{\tiny +} ;
					\node[purple] (18) [above left of =9, yshift=-23, xshift=15]{\tiny +} ;
					\node[purple] (19) [above left of =10, yshift=-23, xshift=15]{\tiny +} ;
					\node[purple] (20) [below of =1, yshift=25, xshift=0]{\tiny +} ;

					\node (21) [left of=7, xshift=17]{\dots};
					\node (22) [left of=3, xshift=17]{\dots};
										
					\draw[->]		(1) edge node[above, yshift=2] {\tiny+} (0)
					(2) edge node[above, yshift=2 ] {\tiny+} (1)
					(4) edge node[above, yshift=2 ] {\tiny+} (2)
					(4) edge [bend left=35] node[below, yshift=0, xshift=10 ] {\tiny $-$} (0)
					(9) edge node[right] {\tiny $-$} (0)
					(9) edge node[right] {\tiny $+$} (2)
					(8) edge node[left, xshift=-1, yshift=5] { \tiny $-$} (0)
					(8) edge node[right] {\tiny $+$} (2)
					(10) edge node[above] {\tiny $+$} (8)
					(10) edge[bend left=35]  node[right, xshift=2] {\tiny $+$} (2)
					(11) edge node[right, xshift=4] {\tiny $+$} (3)
					(15) edge node[right] {\tiny $+$} (3)
					
					;				
					
					\draw[-,dashed]		(12) edge (11)		
					(13) edge (12)	
					(14) edge (13)	
					
					;
					
				}
				
			\end{tikzpicture}
			
			c) \small	After \steptwo.
		\end{center}
		
	\end{minipage}
	\hspace*{1cm}
	\begin{minipage}[t]{.45\linewidth}	
		\begin{center}
			\tikzstyle{pcnode}=[minimum size= 12pt,circle,draw ]
			\begin{tikzpicture}[scale=0.7, outer sep=0pt, inner sep = 0.7pt, node distance=1.2cm,]
				{
					
					\node[pcnode, blue] (0) at (0,0) {\tiny 1};
					\node[pcnode, blue] (1) [right of=0, xshift=-10] {\tiny 2};
					\node[pcnode, blue] (2) [right of=1, xshift=-10] {\tiny 4};
					
					\node[pcnode, blue] (4) [right of=2, xshift=-10] {\tiny $2^{3}$};
					\node[pcnode, blue] (5) [right of=4, xshift=-10] {\tiny $2^{4}$};
					\node[pcnode, blue] (3) [right of=5, xshift=-10] {\tiny $2^{5}$};
					\node[pcnode, blue] (7) [right of=3] {\tiny $2^{7}$};
					\node[pcnode,cyan] (8) [above right of=4, yshift=10] {\tiny $2^{\text{\scalebox{.7}{12}}}$};
					
					\node[pcnode, blue] (11) [below of=4, yshift=10] {\tiny $2^{\text{\scalebox{.7}{32}}}$};
					\node[pcnode, blue] (12) [right of=11, xshift=-10] {\tiny $2^{\text{\scalebox{.7}{33}}}$};
					\node[pcnode, blue] (13) [right of=12, xshift=-10] {\tiny $2^{\text{\scalebox{.7}{34}}}$};
					\node[pcnode, blue] (14) [right of=13, xshift=-10] {\tiny $2^{\text{\scalebox{.7}{35}}}$};

					\node[purple] (17) [above of =5, yshift=-25, xshift=-5]{\tiny +} ;
					\node[purple] (18) [above left of =8, yshift=-23, xshift=15]{\tiny +} ;
					\node[purple] (19) [below of =1, yshift=25, xshift=0]{\tiny +} ;
					
					\node (21) [left of=7, xshift=17]{\dots};		        
					
					\draw[->]		(1) edge node[above, yshift=2] {\tiny+} (0)
					(2) edge node[above, yshift=2 ] {\tiny+} (1)
					(4) edge node[above, yshift=2, xshift=4 ] {\tiny+} (2)
					(4) edge [bend left=35] node[below, yshift=0, xshift=10 ] {\tiny $-$} (0)		        
					(8) edge node[right, yshift=2] {\tiny+} (5)
					(8) edge node[above, yshift=2] {\tiny $-$} (2)
					(11) edge node[above, yshift=-6, xshift=-10] {\tiny+} (3)
					
					;
					
					\draw[-,dashed]	(3) edge (5)
					(5) edge  (4)	
					(12) edge  (11)
					(13) edge  (12)
					(14) edge  (13)
					
					;
				}
				
			\end{tikzpicture}
			
			d) \small After \stepthree 
		\end{center}
		
	\end{minipage}
	\caption{The three steps of power circuit reduction. 
		The already reduced part consist of \textcolor{blue}{blue nodes} and $\Min(\Xi_i)$ is colored in  \textcolor{cyan}{cyan}. The \textcolor{red}{red signs} indicate a marking. Three dots $\cdots$ in between two nodes mean that we omitted some nodes. A dashed edge - -  means that we actually omitted the outgoing edges of the right node. }\label{ExSteps}
\end{figure}
\begin{example}
	In \cref{ExLevel} we give an example of the complete power circuit reduction process by showing the result after each iteration. We start with a non-reduced power circuit of depth $2$ in a). This power circuit has size $5$, so we first construct the starting chain of length $4$ in b). Part c) and d) show the result after inserting layer 0 and layer 1, respectively. In e) we finally inserted all layers and thus have constructed the reduced power circuit. 
\end{example}
\begin{figure}[t]
		\vspace*{-0.1cm}
	\begin{minipage}[b]{.45\linewidth}	
		\begin{center}
			\tikzstyle{pcnode}=[minimum size= 12pt,circle,draw ]
			\begin{tikzpicture}[scale=0.7, outer sep=0pt, inner sep = 0.7pt, node distance=.8cm]
				{
					
					\node[pcnode] (0) at (0,0) {\tiny 1};
					\node[pcnode] (1) [right of=0] {\tiny 1};
					\node[pcnode] (2) [above of=1] {\tiny $2^{2}$};
					\node[pcnode] (3) [above of=0] {\tiny $2^{1}$ };
					\node[pcnode] (4) [above of=2] {\tiny  $2^{7}$};
					
					\node[purple] (15) [below of =0, yshift=15, xshift=0]{\tiny +} ;
					\node[purple] (16) [below of =1, yshift=15, xshift=0]{\tiny $-$} ;
					\node[purple] (17) [right of =2, yshift=0, xshift=-15]{\tiny $+$} ;
					\node[purple] (18) [right of =3, yshift=0, xshift=-15]{\tiny $-$} ;
					\node[purple] (19) [above of =4, yshift=-15, xshift=0]{\tiny $+$} ;

					\draw[->]		(2) edge node[left, yshift=5, xshift=3] {\tiny+} (0)
					(2) edge node[left, yshift=0, xshift=5] {\tiny+} (1)
					(3) edge node[right] {\tiny+} (0)
					(4) edge[bend left=40] node[right] {\tiny+} (1)
					(4) edge node[right] {\tiny+} (2)
					(4) edge node[right] {\tiny+} (3)

					;			
				}
				
			\end{tikzpicture}
			
			a) \small Non-reduced circuit
		\end{center}
		
	\end{minipage}
	\hspace*{1cm}
	\begin{minipage}[b]{.45\linewidth}	
		\begin{center}
			\tikzstyle{pcnode}=[minimum size= 12pt,circle,draw ]
			\begin{tikzpicture}[scale=0.8, outer sep=0pt, inner sep = 0.7pt, node distance=0.8cm]
				{

					\node[pcnode, cyan] (0) at (0,0) {\tiny 1};
					\node[pcnode, cyan] (1) [right of=0] {\tiny 1};
					\node[pcnode] (2) [above of=1] {\tiny $2^{2}$};
					\node[pcnode] (3) [above of=0] {\tiny $2^{1}$ };
					\node[pcnode] (4) [above of=2] {\tiny  $2^{7}$};
					
					\node[pcnode, blue] (5) [right of=1, ] {\tiny 1};
					\node[pcnode, blue] (6) [right of=5, ] {\tiny 2};
					\node[pcnode, blue] (7) [right of=6, ] {\tiny 4};
					\node[pcnode, blue] (8) [right of=7, ] {\tiny $2^{3}$};
					
					\node[purple] (15) [below of =0, yshift=15, xshift=0]{\tiny +} ;
					\node[purple] (16) [below of =1, yshift=15, xshift=0]{\tiny $-$} ;
					\node[purple] (17) [right of =2, yshift=0, xshift=-15]{\tiny $+$} ;
					\node[purple] (18) [right of =3, yshift=0, xshift=-15]{\tiny $-$} ;
					\node[purple] (19) [above of =4, yshift=-15, xshift=0]{\tiny $+$} ;
					
					\draw[->]		(2) edge node[left, yshift=5, xshift=3] {\tiny+} (0)
					(2) edge node[left, yshift=0, xshift=5] {\tiny+} (1)
					(3) edge node[right] {\tiny+} (0)
					(4) edge[bend left=40] node[right] {\tiny+} (1)
					(4) edge node[right] {\tiny+} (2)
					(4) edge node[right] {\tiny+} (3)
					(6) edge node[above] {\tiny+} (5)
					(7) edge node[above] {\tiny+} (6)
					(8) edge node[above] {\tiny+} (7)
					(8) edge[bend left=25] node[below] {\tiny $-$} (5)
					
					;
					
				}
				
			\end{tikzpicture}
			
			b) \small With initial chain $\Gamma_0$.		
			
		\end{center}
		
	\end{minipage}
	\vspace*{0.1cm}
	
	\begin{minipage}[t]{.45\linewidth}	
		\begin{center}
			\tikzstyle{pcnode}=[minimum size= 12pt,circle,draw ]
			\begin{tikzpicture}[scale=0.8, outer sep=0pt, inner sep = 0.7pt, node distance=1cm]
				{
					\node[pcnode, blue] (0) at (0,0) {\tiny 1};
					\node[pcnode, blue] (1) [right of=0] {\tiny 2};
					\node[pcnode, blue] (2) [right of=1] {\tiny 4};
					\node[pcnode, blue] (3) [right of=2] {\tiny $2^{3}$};
					\node[pcnode, blue] (4) [right of=3] {\tiny $2^{5}$};
					\node[pcnode, cyan] (5) [above of=0, yshift=-8] {\tiny $2^{1}$};
					\node[pcnode, cyan] (6) [above of=1,yshift=-8] {\tiny $2^{2}$};
					\node[pcnode] (7) [above left of=6] {\tiny $2^{7}$};
					
					\node (10) [left of =4, xshift=15] {\dots};	
					
					\node[purple] (17) [right of =6, yshift=0, xshift=-20]{\tiny $+$} ;
					\node[purple] (18) [right of =5, yshift=0, xshift=-20]{\tiny $-$} ;
					\node[purple] (19) [above of =7, yshift=-20, xshift=0]{\tiny $+$} ;

					\draw[->]		(1) edge node[above] {\tiny+} (0)
					(2) edge node[above] {\tiny+} (1)
					(3) edge node[above] {\tiny+} (2)
					(3) edge[bend left=25] node[below] {\tiny $-$} (0)
					(5) edge node[left] {\tiny+} (0)
					(6) edge node[left] {\tiny+} (1)
					(7) edge[bend right=60] node[left] {\tiny+} (0)
					(7) edge node[left] {\tiny+} (5)
					(7) edge node[left] {\tiny+} (6)
					
					;
				}
				
			\end{tikzpicture}
			
			c) \small After inserting layer $0$		
			
		\end{center}
	\end{minipage}
	\hspace*{1cm}
	\begin{minipage}[c]{.45\linewidth}	
		\vspace*{-0.2cm}
		\begin{center}
			\tikzstyle{pcnode}=[minimum size= 12pt,circle,draw ]
			\begin{tikzpicture}[scale=0.8, outer sep=0pt, inner sep = 0.7pt, node distance=1cm]
				{
					\node[pcnode, blue] (0) at (0,0) {\tiny 1};
					\node[pcnode, blue] (1) [right of=0] {\tiny 2};
					\node[pcnode, blue] (2) [right of=1] {\tiny 4};
					\node[pcnode, blue] (3) [right of=2] {\tiny $2^{3}$};
					\node[pcnode, blue] (4) [right of=3] {\tiny $2^{7}$};
					\node[pcnode, cyan] (5) [above right of=0] {\tiny $2^{7}$};

					\node (10) [left of =4, xshift=15] {\dots};

					\node[purple] (19) [above of =5, yshift=-20, xshift=0]{\tiny $+$} ;
					\node[purple] (20) [below of =1, yshift=20, xshift=0]{\tiny $+$} ;
					
					\draw[->]		(1) edge node[above] {\tiny+} (0)
					(2) edge node[above] {\tiny+} (1)
					(3) edge node[above] {\tiny+} (2)
					(3) edge[bend left=28] node[below] {\tiny $-$} (0)	
					(5) edge node[left] {\tiny $-$} (0)
					(5) edge node[right, xshift=8] {\tiny+} (3)
					;
				}
				
			\end{tikzpicture}
			
			d) \small  After inserting layer $1$				
			
		\end{center}
		\begin{center}
			\tikzstyle{pcnode}=[minimum size= 12pt,circle,draw ]
			\begin{tikzpicture}[scale=0.8, outer sep=0pt, inner sep = 0.7pt, node distance=1cm]
				{
					\node[pcnode, blue] (0) at (0,0) {\tiny 1};
					\node[pcnode, blue] (1) [right of=0] {\tiny 2};
					\node[pcnode, blue] (2) [right of=1] {\tiny 4};
					\node[pcnode, blue] (3) [right of=2] {\tiny $2^{3}$};
					\node[pcnode, blue] (4) [right of=3] {\tiny $2^{7}$};
					\node[pcnode, blue] (5) [right of=4] {\tiny $2^{8}$};
					
					\node (10) [left of =4, xshift=15] {\dots};

					\node[purple] (19) [below of =4, yshift=38, xshift=0]{\tiny $+$} ;
					\node[purple] (20) [below of =1, yshift=38, xshift=0]{\tiny $+$} ;		
					
					\draw[->]		(1) edge node[above] {\tiny+} (0)
					(2) edge node[above] {\tiny+} (1)
					(3) edge node[above] {\tiny+} (2)
					(3) edge[bend left=25] node[below] {\tiny $-$} (0)
					;
					
					\draw[-,dashed]	(5) edge (4)
					;
				}
				
			\end{tikzpicture}
						
						e) \small After inserting layer $2$
		\end{center}
	\end{minipage}
		\caption{The complete process of power circuit reduction~-- inserting layer after layer. For an explanation of the colors, see \cref{ExSteps}.}\label{ExLevel}
\end{figure}

For comparing two markings $L$ and $M$ on an arbitrary power circuit, we can proceed as follows: first compute the difference (\cref{lem: OpOnMarkings}), then reduce the power circuit
(\cref{thm:pcred}) and, finally, compare the resulting compact marking with zero (\cref{lem:compareCompactMarkings}). 
This shows:
\begin{corollary}\label{cor:compareMarkings}
	Let ${\comparator} \in \compOpSet$. The following is in \PTc{0}{D} parametrized by $\depth(\Pi)$: 
	\problem{A power circuit $(\Pi, \delta_{\Pi})$ together with markings $L,M$ on $\Pi$.}
	{Is $\eps(L)  \comparator \eps(M)$?
	}
\end{corollary}
\begin{proof}[Proof of Proposition~\ref{prop:compIntro}]
	When assuming $\depth(\Pi) \leq C \cdot \log\abs{\Pi}$, by \cref{lem:paraToTC}, we obtain Proposition~\ref{prop:compIntro} as an immediate consequence of \cref{cor:compareMarkings}. 
\end{proof}

\begin{remark}
	By \cref{cor:compareMarkings} comparing two numbers $m_1$ and $m_2$ represented by a \ptxCircuit $\cC$ can be done in 
 \PTc{0}{D} parametrized by $\depth_{2^x}(\cC) + \log^2\abs{\cC}$ if the input of every $2^x$-gate of $\cC$ is non-negative: by \cref{prop:convertToArithmeticCircuit} we can find a power circuit $(\Pi, \delta)$ with $\depth(\Pi) \leq \depth_{2^x}(\cC)$ and markings $M_1$ and $M_2$ evaluating to $m_1$ and $m_2$ in $\Nc2\sse \Tc2$. It remains to compare $\eps(M_1)$ and $\eps(M_2)$.

\end{remark}

\subsection{Operations with floating point numbers} \label{sec: morePCoperations}

In the following, we want to represent a number $r \in \Z[1/2]$ using markings in a power circuit. For this, we use a floating point representation.
Observe that for each such $r \in \Z[1/2] \setminus \{0\}$ there exist unique $u,e \in \mathbb{Z}$ with $u$ odd such that $r = u \cdot 2^{e}$.

\begin{lemma} \label{lem: mTOr}
The following problem is in \PTc{0}{D} parametrized by $\depth(\Pi)$:
\compproblem{A power circuit $(\Pi, \delta_{\Pi})$ with a marking $M$ on $\Pi$.}{A power circuit $(\tilde{\Pi}, \delta_{\tilde{\Pi}})$ with markings $E, U$ on $\tilde{\Pi}$ such that $\epsilon(M) = \epsilon(U)\cdot 2^{\epsilon(E)}$ with $\epsilon(U)$ odd.}
In addition, $(\Pi, \delta_{\Pi}) \leq (\tilde{\Pi}, \delta_{\tilde{\Pi}})$, $\depth(\tilde{\Pi}) \leq \max(\depth(\Pi),2)$ and $\abssmall{\tilde{\Pi}} \in \Oh(\abs{\Pi})$.%

\end{lemma}

\begin{proof} 
First note that we are searching for a marking representing the maximal $e \in \Z$  with $2^{e}~|~ \epsilon(M)$. For finding $e$, we need the compact representation of $M$. 
Therefore, we construct the reduced form $(\Gamma, \delta)$ of $\Pi$ and a compact marking $\wt M$ on $\Gamma$ such that $\epsilon(\wt M) = \epsilon(M)$. According to \cref{thm:pcred} this is possible in \PTc{0}{D}.
Now we have $\epsilon(M) = \sum_{i = 1}^{k} \wt M(Q_{i}) \cdot 2^{\epsilon(\Lambda_{Q_{i}})}$ where $\supp(\wt M) = \oneset{Q_{1}, \ldots, Q_{k}} \sse \Gamma$. We assume that the $Q_i$ are ordered according to their value, \ie $\eps(Q_{i}) < \eps(Q_{j}) $ for $i < j$. 
 Hence, $e = \epsilon(\Lambda_{Q_{1}})$. 

Before we can define the markings $U$ and $E$, we have to introduce some new nodes. First we add $\floor{\log(\abs{\Gamma})}$ new nodes  to $\Pi$ each of value $1$ (\ie with empty successor marking). Then for each $j \in \interval{0}{ \floor{\log(\abs{\Gamma})}}$ we create a node of value $2^{j}$ and depth $1$ in the following way: the successor marking of such a node marks exactly $j$ nodes of value $1$ with $+1$ and all the other nodes with $0$. 

In order to define $U$, we aim for adding a node $S_{i}$ to $\Pi$ with $\epsilon(\Lambda_{S_{i}}) = \epsilon(\Lambda_{Q_{i}}) - e$ for each $i \in \interval{1}{k}$. We proceed as follows: For each $i \in \interval{1}{k}$, let $C_i \in \mathcal{C}_{\Gamma}$ denote the maximal chain to which $Q_i$ belongs. 
 Note that for different $i$ these chains could be equal. By \cref{rem:startchain}, we know that there exist nodes $R_1, \dots, R_k \in \Pi$ such that $\epsilon(R_{i}) = \epsilon(\startc{C_{i}})$ for $i \in \interval{1}{k}$. To find the nodes $R_{i}$, we can for example remember the equivalence classes we obtain during the reduction process. 
 Now there exist $m_{i} \in \mathbb{N} $ with $m_i \in \interval{0}{\abs{\Gamma}} $ such that $\epsilon(\Lambda_{Q_{i}}) = \epsilon(\Lambda_{R_{i}}) + m_{i}$. We can find $m_{i}$ as the difference of the indices of $Q_{i}$ and $\startc{C_{i}}$ in the sorted order of $\Gamma$, and so we can find all the $m_{i}$ in \Ac0. Note that the binary representation of $m_i$ uses at most $\floor{\log(\abs{\Gamma})} + 1$ bits. We define markings $M_{i}$ on the newly defined nodes of depth $1$ using the binary representation of $m_{i}$ such that $\epsilon(M_{i}) = m_{i}$ for $i \in \interval{1}{k}$.
Now we are ready to define the marking $E$ by $E = \Lambda_{R_{1}} + M_{1}$. 

We now want to define a marking $U$, with $\epsilon(U) = \epsilon(M)\cdot 2^{-\epsilon(E)}$. For every $i \in \interval{1}{k}$ we create a node $S_{i}$ with $\Lambda_{S_{i}} = \Lambda_{R_{i}}+M_{i}-E$, so in particular, $\epsilon(S_{i}) = \epsilon(Q_{i}) \cdot 2^{-\epsilon(E)}$ (notice that $\eps(S_1) = 1$). Because $E$ and $M_{i}$ could have supports with non-trivial intersection (as well as $E$ and $\Lambda_{R_{i}}$), we have to clone the nodes in $\supp(E)$ for the addition. Then the marking $U$ with $U(S_{i}) = \wt M(Q_{i})$ for $i \in \interval{1}{k}$  is the marking $U$ we searched for. 

Regarding the size of $\tilde{\Pi}$, observe that to define the markings $M_{i}$, we insert $2 \cdot \floor{\log(\abs{\Gamma})} + 1$ nodes. By cloning the nodes in $\supp(E)$, we add at most $  \floor{\log(\abs{\Gamma})} + 1 + \abs{\Pi}$ additional nodes. By \cref{rem:startchain} (\ref{sizeCompactM}), we know that $\abssmall{\supp(\wt M)} \leq \abs{\supp(M)} \leq \abs{\Pi}$, so we insert at most $\abs{\Pi}$ nodes when inserting the nodes $S_{i}$. According to \cref{lem:pcredsize}, we have $\log(\abs{\Gamma}) \in \Oh(\log(\abs{\Pi}))$. Hence, $\abssmall{\tilde{\Pi}} \in \Oh(\abs{\Pi})$.

Considering the depth, when inserting the new nodes of depth $1$, the depth only increases if $\depth(\Pi) = 0$. When inserting a node $S_{i}$ the depth increases only if $\depth(\Pi) \leq 1$.
\end{proof}

\begin{definition}\label{def: pcrep}
A power circuit representation of $r \in \Z[1/2]$ consists of a power circuit $(\Pi, \delta_{\Pi})$ together with a pair of markings $(U, E)$ on $\Pi$ such that $\epsilon(U)$ is either zero or odd and $r = \epsilon(U) \cdot 2^{\epsilon(E)}$. 
\end{definition}

\begin{lemma} \label{lem: OpOnFP}
\begin{enumerate}[(a)]
\item\label{TCOpOnFP} The following problems are in \TC:
\compproblem{A power circuit representation for $r \in \Z[1/2] $ over a power circuit $(\Pi, \delta_{\Pi})$ and a marking $M$ on $\Pi$. }{A power circuit representation of $r \cdot 2^{\epsilon(M)}$ over a power circuit $(\tilde{\Pi}, \delta_{\tilde{\Pi}})$.}
\compproblem{A power circuit representation for $r \in \Z[1/2] $ over a power circuit $(\Pi, \delta_{\Pi})$. }{A power circuit representation of $-r$ over $(\Pi, \delta_{\Pi})$.}
\compproblem{Power circuit representations for $r,s \in \Z[1/2] $ over a power circuit $(\Pi, \delta_{\Pi})$ such that $\frac{r}{s}$ is a power of two.}{A marking $L$ in a power circuit $(\tilde{\Pi}, \delta_{\tilde{\Pi}})$ such that $\epsilon(L) = \log(\frac{r}{s})$.}

\item\label{PTCOpOnFP} The following problems are in \PTc{0}{D}  parametrized by $\depth(\Pi)$: 
\compproblem{A power circuit $(\Pi, \delta_{\Pi})$ and a marking $M$ on $\Pi$.}{A power circuit representation of $\epsilon(M) \in  \Z[1/2] $ over a power circuit $(\tilde{\Pi}, \delta_{\tilde{\Pi}})$.}
\compproblem{$r,s \in  \Z[1/2] $ given as power circuit representations over a power circuit $(\Pi, \delta_{\Pi})$.}{A power circuit representation of $r+s$ over a power circuit $(\tilde{\Pi}, \delta_{\tilde{\Pi}})$. }
\problem{A power circuit representation for $r \in \Z[1/2] $ over a power circuit $(\Pi, \delta_{\Pi})$}{Is $r \comparator 0$ for ${\comparator} \in \compOpSet$?}
\compproblem{A power circuit representation for $r \in \Z[1/2] $ over a power circuit $(\Pi, \delta_{\Pi})$.}{Is $r \in \Z$? If yes, a marking $M$ in a power circuit $(\tilde{\Pi}, \delta_{\tilde{\Pi}})$ such that $\epsilon(M) = r$.}
\end{enumerate}
\noindent In all cases we have  $(\Pi, \delta_{\Pi}) \leq (\tilde{\Pi}, \delta_{\tilde{\Pi}})$, $\abssmall{\tilde{\Pi}} \in \Oh(\abs{\Pi})$, and $\depth(\tilde{\Pi}) = \depth(\Pi) + \Oh(1)$.

\end{lemma}

\begin{proof} 
During the whole proof, let  $U,V,E,F$ be markings in $\Pi$ such that $\epsilon(U)$,  $\epsilon(V)$ are odd, $r = \epsilon(U)\cdot 2^{\epsilon(E)} $ and $s = \epsilon(V)\cdot 2^{\epsilon(F)}$. 

\proofsubparagraph{Part (a):}
 We have $r \cdot 2^{\epsilon(M)} = \epsilon(U)\cdot 2^{\epsilon(E)+\epsilon(M)}$. According to  \cref{lem: OpOnMarkings} point (\ref{addmark}), the marking $E+M$ can be obtained in \TC as marking in a power circuit $(\tilde{\Pi}, \delta_{\tilde{\Pi}})$ that satisfies all the required properties. The computation of $-r$ is clear by \cref{lem: OpOnMarkings}.

If $\frac{r}{s}$ is a power of two, we know that $\epsilon(U) = \epsilon(V)$, and so $\frac{r}{s} = 2^{\epsilon(E)-\epsilon(F)} $. Now again \cref{lem: OpOnMarkings} finishes the proof of part (a). 

\proofsubparagraph{Part (b):} 
The first point is due to \cref{lem: mTOr}. 
For the addition, first observe that
\[r+s = \epsilon(U) \cdot 2^{\epsilon(E)}+\epsilon(V) \cdot 2^{\epsilon(F)} = 2^{\epsilon(E)}\cdot (\epsilon(U)+\epsilon(V) \cdot 2^{\epsilon(F)-\epsilon(E)})\] 
We can decide in \PTc{0}{D} whether $\epsilon(E) \leq \epsilon(F)$ using \cref{cor:compareMarkings}. \Wlog let $\epsilon(E) \leq \epsilon(F)$ (otherwise we switch the roles of $r$ and $s$).
Next, we construct a marking $K$ in a power circuit $(\Pi', \delta_{\Pi'})$ such that $\epsilon(K) = \epsilon(U)+\epsilon(V) \cdot 2^{\epsilon(F)-\epsilon(E)}$. According to \cref{lem: OpOnMarkings} and because $\epsilon(F)-\epsilon(E) \geq 0$, this is possible in \TC and such that $\abs{\Pi'} \in \Oh(\abs{\Pi})$ and $\depth(\Pi')\leq \depth(\Pi)+1$. Now, according to \cref{lem: mTOr} we can construct markings $W$ and $G$ in a power circuit $(\Pi'', \delta_{\Pi''})$ such that $\epsilon(W)$ is odd and $\epsilon(K) = \epsilon(W) \cdot 2^{\epsilon(G)}$ in \PTc{0}{D}. In addition, $\abs{\Pi''} \in \Oh(\abs{\Pi'} )$ and $\depth(\Pi'') \leq \max(\depth(\Pi'),2)$.  Now according to \cref{lem: OpOnMarkings} the marking $E + G$ can be obtained as marking in a power circuit $(\tilde{\Pi}, \delta_{\tilde{\Pi}})$ with $\abssmall{\tilde{\Pi}} \in \Oh(\abs{\Pi})$ and $\depth(\tilde{\Pi}) = \depth(\Pi) + \Oh(1)$. Then the power circuit $(\tilde{\Pi}, \delta_{\tilde{\Pi}})$ together with the markings $W$ and $E+G$ is the power circuit representation for $r+s$, satisfying the required properties. 

\medskip
To decide if $r \comparator 0$, we just have to check if $\epsilon(U) \comparator 0$. According to  \cref{cor:compareMarkings} this is possible in \PTc{0}{D}.
To decide if $r \in \Z$, since $\eps(U)$ is odd, we just need to decide if $\epsilon(E) \geq 0$. Again, this can be done using \cref{cor:compareMarkings}. In the affirmative case, we just have to apply \cref{lem: OpOnMarkings} point (\ref{powermark}) to produce the desired output.
\end{proof}

\section{The word problem of the Baumslag group}\label{sec:BS}

Before we start solving the word problem of the Baumslag group, let us fix our notation from group theory.
\subparagraph*{Group presentations.}
A group $G$ is \emph{finitely generated} if there is some finite set $\Sigma$ and a surjective monoid homomorphism $\eta\colon  \Sigma^* \to G$ (called a presentation). 
Usually, we do not write the homomorphism $\eta$ and treat words over $\Sigma$ both as words and as their images under $\eta$. We write $v =_G w$ with the meaning that $\eta(v) = \eta(w)$. 
If $\Sigma = S \cup S^{-1}$ where $S^{-1}$ is some disjoint set of formal inverses and $R\sse \Sigma^* \times \Sigma^*$ is some set of relations, we write $\Gen{\Sigma}{R}$ for the group  $\Sigma^*/C(R)$ where $C(R)$ is the congruence generated by $R$ together with the relations $aa^{-1} = a^{-1} a = 1$ for $a \in \Sigma$. If $R$ is finite, $G$ is called \emph{finitely presented}.

\medskip
\noindent The word problem for a fixed group $G$ with presentation $\eta\colon  \Sigma^* \to G$ is as follows:
\problem{A word $w \in \Sigma^*$}{Is $w =_G 1$?}
\noindent For further background on group theory, we refer to \cite{LS01}.

\subparagraph*{The Baumslag-Solitar group.}

The Baumslag-Solitar group is defined by 
\begin{align*}
	\BS12& = \Gen{a,t}{tat^{-1} = a^2}.
\end{align*}
We have $\BS12 \cong \Z[1/2] \rtimes \Z$ via the isomorphism $a\mapsto (1,0)$ and $t\mapsto (0,1)$. 
	Recall that $\Z[1/2] = \set{p/2^q\in \Q}{p,q \in \Z}$ is the set of dyadic fractions with addition as group operation. 
The multiplication in $\Z[1/2] \rtimes \Z$ is
defined by  $(r,m) \cdot (s,n) = (r + 2^m s, m + n)$. 
	Inverses can be computed by the formula $(r,m)^{-1} = (-r \cdot 2^{-m}, -m)$. 
In the following we use $\BS12$ and  $\Z[1/2] \rtimes \Z$ as synonyms.

\subparagraph*{The Baumslag group.}
A convenient way to understand the Baumslag group $\BG$ is as an HNN extension\footnote{Named after
	Graham Higman, Bernhard H.~Neumann and Hanna Neumann. For a precise definition, we refer to \cite{LS01}. This is also the way how the Magnus breakdown procedure works.} of the Baumslag-Solitar group: 
\begin{align*}
	\begin{split}
		\BG  & = \Gen{\BS12,b}{bab^{-1} = t}\\ 
		& = \Gen{a,t,b}{tat^{-1} = a^2, bab^{-1} = t}.
	\end{split}
\end{align*}
Indeed, due to $bab^{-1} = t$, we can remove $t$ and we obtain exactly the presentation $\Gen{a,b}{bab^{-1} a = a^2bab^{-1}}$. 
Moreover, \BS12 is a subgroup of \BG via the canonical embedding and we have $b(q,0)b^{-1} = (0,q)$, so a conjugation  by $b$ ``flips'' the two components of the semi-direct product if possible. Henceforth, we will use the alphabet $\Sigma = \{1,a, a^{-1}, t, t^{-1}, b, b^{-1}\}$ to represent elements of \BG (the letter $1$ represents the group identity; it is there for padding reasons).

\vspace{-1mm}
\subparagraph*{Britton reductions.}
Britton reductions are a standard way to solve the word problem in HNN extensions. 
Here we define them for the special case of \BG. Let
\[\Delta = \BS12 \cup \oneset{b, b^{-1}}\]
be an infinite alphabet (note that $\Sigma \sse \Delta$). A word $w \in \Delta^*$ is called \emph{Britton-reduced} if it is of the form 
\[w = (s_{0},n_{0})  \beta_{1} (s_{1},n_{1}) \cdots  \beta_{\ell} (s_{\ell}, n_{\ell} )\]
with $\beta_i \in \oneset{b, b^{-1}}$ and $(s_{i},n_{i}) \in \BS12$ for all $i$ (\ie $w$ does not have two successive letters from $\BS12$) and there is no factor of the form $b(q,0)b^{-1}$ or $b^{-1}(0,k)b$ with $q, k \in \Z$. If $w$ is not Britton-reduced, one can apply one of the rules
\begin{align*}
	(r,m)(s,n) 	 &\:\to\: (r + 2^m s, m + n)\\
	b(q,0)b^{-1} &\:\to\: (0,q)\\
	b^{-1}(0,k)b &\:\to\: (k,0)
\end{align*}
in order to obtain a shorter word representing the same group element. The following lemma is well-known (see also \cite[Section IV.2]{LS01}).

\begin{lemma}[Britton's Lemma for \BG{}\!  \cite{britton63}]\label{lem:Britton}
	Let $w \in \Delta^*$ be Britton-reduced. Then $w \in \BS12$ as a group element if and only if $w$ does not contain any letter $b$ or $b^{-1}$. In particular, $w \eqBG 1$ if and only if $w = (0,0)$ or $w = 1$ as a word.
\end{lemma}

\begin{example}\label{ex:blowup}
	Define  words 
	$w_0 = t$ and $w_{n+1} = b\,  w_{n}\,  a\, w_{n}^{-1}\, b^{-1}$
	for $n \geq 0$. Then we have $\abs{w_{n}} = 2^{n+2} -3$ but 
	$w_{n} \eqBG t^{\tow(n)}$. While the length of the word $w_n$ is only exponential in $n$, the length of its \Breduced 
	form is $\tow(n)$. 
\end{example}

\subsection{Conditions for Britton reductions}\label{sec:WP}

The idea to obtain a parallel algorithm for the word problem is to compute a Britton reduction of $uv$ given that both $u$ and $v$ are Britton-reduced. For this, we have to find a maximal suffix of $u$ which cancels with a prefix of $v$. The following lemma is our main tool for finding the longest canceling suffix. It is important to note that for all suffixes the conditions can be checked in parallel.

\newcommand{\modLog}{\log}

\begin{lemma}\label{lem:Bred_conditions}
	
	Let $ w = \beta_{1}(r,m)\beta_{2}\, x\, \beta_{2}^{-1}(s,n)\beta_{1}^{-1} \in \Delta^*$ with  $\beta_{1}, \beta_{2} \in \oneset{b,b^{-1}}$  such that $\beta_{1}(r,m)\beta_{2}$ and $\beta_{2}^{-1}(s,n)\beta_{1}^{-1}$ both are Britton-reduced and $\beta_{2}x\beta_{2}^{-1} \eqBG (q,k) \in \BS12$ (in particular, $k = 0$ and $q \in \Z$, or $q = 0$).

	Then $w \in \BS12$ if and only if the respective condition in the following table is satisfied. Moreover, if $w \in \BS12$, then $w \eqBG \hat w$ according to the last column of the table.	
	\begin{center}
			\renewcommand{\arraystretch}{1.4}
				\begin{tabular}[h]{cc|rll|l}
		
			$\beta_{1}$ & $\beta_{2}$
			& \multicolumn{3}{c|}{Condition}  & \multicolumn{1}{c}{$\hat w$} \\

			\hline
			$b$ & $b$  & $r+2^{m+k}s$\!\!&$\in \mathbb{Z}$, &$m+n +k=0$& $\left(0,\; r+2^{-n}s\right)$ \\
			
			$b$ & $b^{-1}$ & $r+2^{m}(q+s)$\!\!&$\in \mathbb{Z}$, &$\hphantom{m+n +k}\mathllap{ m+n} = 0$ & $\left(0,\; r+2^{m}(q+s)\right)$\\
			
			$b^{-1}$ & $b$ &  $r + 2^{m+k}s$\!\!&$ = 0$&& $\left(n+ \modLog(\frac{-r}{s}),\; 0\right)$\\
			
			$b^{-1}$ & $b^{-1}$ &  $ r + 2^{m}(q+s)$\!\!&$ = 0$& &$\left(m+n,\; 0\right)$\\
			
		\end{tabular}
		\end{center}
	Notice that in the case $\beta_{1} = \oi{b}$ and $\beta_{2} = b$, we have $r \neq 0 $ and $s \neq 0$.	
\end{lemma}

\begin{example} \label{ex: formulas} Let us illustrate with two examples how to read \cref{lem:Bred_conditions}.
	For this, let $ w = \beta_{1}(r_{1},m_{1})\beta_{2}\, x\, \beta_{2}^{-1}(s_{1},n_{1})\beta_{1}^{-1} \in \Delta^*$ with the same properties as in \cref{lem:Bred_conditions}, in particular, $ \beta_{2}\, x\, \beta_{2}^{-1} \eqBG  (q,k)\in \BS12$. 
	
	We first consider the case that $\beta_{1} = \beta_{2} = b$. To check if $w \in \BS12$ we have to check if $m_{1}+n_{1} + k = 0$ and if $r_{1}+2^{ m_{1}+k}\cdot s_{1} \in \mathbb{Z}$ according to \cref{lem:Bred_conditions}. In order to obtain a formula for $k$, we apply \cref{lem:Bred_conditions} to $\beta_{2}\, x\, \beta_{2}^{-1}$ using the rightmost column. We write 
	\begin{align}\label{beta3}
		(q,k) \eqBG \beta_{2}\, x\, \beta_{2}^{-1} = \beta_{2}(r_{2}, m_{2})\beta_{3}\, x'\,  \beta_{3}^{-1}(s_{2},n_{2})\beta_{2}^{-1}.
	\end{align}	
	For our example let us assume that $\beta_{3} = b$. Then according to \cref{lem:Bred_conditions}, $(q,k) = (0, r_{2}+2^{-n_{2}}\cdot s_{2})$. Hence, $w \in \BS12$ if and only if $m_{1} + n_{1} + ( r_{2}+2^{-n_{2}}\cdot s_{2}) = 0$  and $r_{1}+2^{ m_{1}+r_{2}+2^{-n_{2}}\cdot s_{2}}\cdot s_{1} \in \mathbb{Z}$. If both conditions are satisfied, we know that $w \eqBG \left(0,\; r_{1}+2^{-n_{1}}s_{1}\right)$. 
	
	\medskip
	Now let us consider a more difficult case. Assume that $ \beta_{1} = b^{-1}$ and $\beta_{2} = b$. According to \cref{lem:Bred_conditions}, $w \in \BS12$ if and only if $r_{1} + 2^{m_{1}+k}s_{1} = 0$. To obtain a formula for $k$,  we apply again \cref{lem:Bred_conditions} using the rightmost column. We again write $(q,k) \eqBG \beta_{2}\, x\, \beta_{2}^{-1} $ as in (\ref{beta3}), but here we assume that $\beta_{3} = b^{-1}$. As we assumed $\beta_{2} = b$, we have by the last column in \cref{lem:Bred_conditions} that
	\[(q,k) \eqBG \beta_{2}\, x\, \beta_{2}^{-1} = \beta_{2} (r_{2}, m_{2}) \beta_{3}\,  x'\,  \beta_{3}^{-1} (s_{2},n_{2})\beta_{2}^{-1} \eqBG (0,r_{2}+2^{m_{2}}(q'+s_{2}))\]
	if $\beta_{3}x'\beta_{3}^{-1} \eqBG (q',k') \in \BS12$.
	This means that $k = r_{2}+2^{m_{2}}(q'+s_{2})$. Now we have to find a formula for $q'$. We write 
	\[(q',k') \eqBG \beta_{3}\, x'\, \beta_{3}^{-1} = \beta_{3}(r_{3},m_{3})\beta_{4}\, x''\,  \beta_{4}^{-1}(s_{3},n_{3})\beta_{3}^{-1}.\] 
	 We assume $\beta_{4} = b$. Then according to \cref{lem:Bred_conditions} we obtain that $q' = n_{3} + \modLog\left(\frac{-r_{3}}{s_{3}}\right)$.
	This implies that $k = r_{2} + 2^{m_{2}} \cdot (n_{3} + \modLog(\frac{-r_{3}}{s_{3}})+s_{2})$. So, to check whether $w \in \BS12$, we have to check if 
	\[r_{1} + 2^{m_{1} + r_{2} + 2^{m_{2}} \cdot (n_{3} + \modLog(\frac{-r_{3}}{s_{3}})+s_{2})} \cdot s_{1} = 0.\] 
	If this is the case, then $w \eqBG \left(n_{1} + \modLog\left(\frac{-r_{1}}{s_{1}}\right),0\right)$. 	
\end{example}

\begin{proof}[Proof of \cref{lem:Bred_conditions}]
	We distinguish the two cases $\beta_{2} = b$ and $\beta_{2} = \oi{b}$. Each case consists of two sub-cases depending on $\beta_1$.
	
	\proofsubparagraph{Case $\beta_{2} = b$:} Since $\beta_{2}x\beta_{2}^{-1} \in \BS12$, we have $\beta_{2} x \beta_{2}^{-1} \eqBG (0,k)$ for some $k \in \Z$. Therefore, we obtain
	\begin{align*}
		(r,m)\beta_{2}\, x\, \beta_{2}^{-1}(s,n) 	& \eqBG (r,m) (0,k) (s,n) \\
		& \eqBG (r,m + k) (s,n)	\\ 
		& \eqBG (r + 2^{m+k}s,\, m + k + n).
	\end{align*}
	Thus, if $\beta_{1} = b$, we have $w \in \BS12$ if and only if $r + 2^{m+k}s \in \Z$ and $m + n + k = 0$. Moreover, if the latter conditions are satisfied, we have $w \eqBG b (r + 2^{m+k}s,\, 0)\oi{b} = b (r + 2^{-n}s,\, 0)\oi{b} \eqBG (0,\,r + 2^{-n}s )$.
	
	On the other hand, if $\beta_{1} = \oi{b}$, it follows that $w \in \BS12$ if and only if $r + 2^{m+k}s = 0$. Notice that in this case, by the assumption that $\beta_{1}(r,m)\beta_{2}$ and $\beta_{2}^{-1}(s,n)\beta_{1}^{-1}$ are Britton-reduced, we have $r\neq 0 $ and $s\neq 0$. Therefore, if the condition $r + 2^{m+k}s = 0$ is satisfied, we have $k = \modLog(\frac{-r}{2^{m}s})$. Hence, under this condition, we have $w \eqBG \oi{b} (0,\, m + k + n)b = \oi{b} (0,\, m + \modLog(\frac{-r}{2^{m}s}) + n)b \eqBG (n + \modLog(\frac{-r}{s}) ,0)$ (because $\modLog(\frac{-r}{2^{m}s}) = \modLog(\frac{-r}{s})-m$).
	
	\proofsubparagraph{Case $\beta_{2} = b^{-1}$:} In this case, we do a similar computation:	
	\begin{align*}
		(r,m)\beta_{2}\, x\, \beta_{2}^{-1}(s,n) 	& \eqBG (r,m) (q,0) (s,n) \\
		& \eqBG (r,m ) (q+s,n)	\\
		& \eqBG (r + 2^{m}(q+s),\,  m + n)
	\end{align*}
	with $q \in \Z$. Again, let us consider the case $\beta_{1} = b$ first. In this case we have $w \in \BS12$ if and only if $r + 2^{m}(q+s) \in \Z$ and $m + n = 0$.
	If these conditions are satisfied, we have $w \eqBG b (r + 2^{m}(q+s),\, 0)\oi{b} \eqBG (0,\,r + 2^{m}(q+s) )$.
	
	Finally, if $\beta_{1} = \oi{b}$, it follows that $w \in \BS12$ if and only if $r + 2^{m}(q+s) = 0$. If this applies, we have $w \eqBG \oi{b} (r + 2^{m}(q+s),\, m + n)b \eqBG ( m + n,0)$.
\end{proof}

Let us fix the following notation for elements $v, w \in \BG$ written as words over $\Delta$:  
\begin{align}
u = (r_{h},m_{h}) \beta_{h}  \cdots(r_{1},m_{1})  \beta_{1}(r_{0},m_{0}),  \qquad v = (s_{0},n_{0})  \tilde{\beta}_{1} (s_{1},n_{1}) \cdots  \tilde{\beta}_{\ell}(s_{\ell},n_{\ell} )\label{betaFact}
\end{align}
 with $(r_{j},m_{j}),  (s_{j},n_{j})\in \sdZ$ and $\beta_{j},\tilde{\beta}_{j} \in \oneset{b, b^{-1}}$. 
 We define 
  \[uv[i,j] =  \beta_{i} (r_{i-1},m_{i-1}) \cdots \beta_{1}(r_{0},m_{0}) \; (s_{0},n_{0})  \tilde{\beta}_{1}  \cdots  (s_{j-1},n_{j-1})\tilde{\beta}_{j}.\]  
 Notice that as an immediate consequence of \hyperref[lem:Britton]{Britton's Lemma} we obtain that, if $u$ and $v$ as in (\ref{betaFact}) are Britton-reduced and $uv[i,i] \in \BS12$ for some $i$, then also $uv[j,j] \in \BS12$ for all $j \leq i$. Moreover, $uv$ is Britton-reduced if and only if $\beta_{1}(r_{0},m_{0})  (s_{0},n_{0})  \tilde{\beta}_{1} \not \in \BS12$.

\medskip
\newcommand{\PowExp}{\mathrm{PowExp}}

For $\ell\in \N$ let $\cX_\ell$ denote some set of $\ell$ variables. Denote by $\PowExp(\cX_\ell)$ the set of expressions which can be made up from the variables $\cX_\ell$ using the operations $+$, $-$, $(r,s) \mapsto r\cdot2^s$ if $s \in \Z$ (and undefined otherwise), and $(r,s) \mapsto \modLog(r/s)$ if $\modLog(r/s) \in \Z$  (and undefined otherwise).

\begin{lemma}\label{lem:BredCond}
	For every $\vec \beta \in \{b,b^{-1},\bot\}^4$ there are expressions $ \theta_{\vec \beta}, \xi_{\vec \beta},\phi_{\vec \beta},\psi_{\vec \beta} \in \PowExp(\cX_{12})$ such that the following holds:	
	Let  $u,v \in \BG$ as in (\ref{betaFact}) be Britton-reduced and assume that $uv[i-1,i-1] \in \BS12$ and $\beta_{i} = \tilde \beta_{i}^{-1}$ and let $V_{i} = \set{r_j,s_j,m_j,n_j}{j \in \{i-1,i-2, i-3\}}$. If $\vec \beta = (\beta_{i},\beta_{i-1},\beta_{i-2},\beta_{i-3})$ (where $\beta_j = \bot$ for $j\leq 0$), then	
	\smallskip
	\begin{enumerate}
		\item $uv[i,i] \in \BS12\; $ if and only if $\; \theta_{\vec \beta}(V_{i}) \in \Z$ and $ \xi_{\vec \beta}(V_{i}) = 0$,\smallskip
		\item if  $uv[i,i] \in \BS12$, then  $uv[i,i] \eqBG \bigl(\phi_{\vec \beta}(V_{i}), \psi_{\vec \beta}(V_{i})\bigr)$.
	\end{enumerate}
\end{lemma}

Be aware that here we have to read the set $V_{i}$ of cardinality (at most) 12 as assignment to the variables $\cX_{12}$. In particular, given that $uv[i-1,i-1] \in \BS12$, one can decide whether $uv[i,i] \in \BS12$ by looking at only constantly many letters of $uv$~-- this is the crucial observation we shall be using for describing an \NC algorithm for the word problem of \BG (see \cref{prop:Bredstep} below).

\begin{proof} \Wlog $i\geq 4$. We follow the approach of \cref{ex: formulas}. By assumption we know that there exist $q, k \in \mathbb{Z}$ such that $uv[i-1,i-1] \eqBG (q,k) \in \BS12$.  
According to the conditions in \cref{lem:Bred_conditions}, to show \cref{lem:BredCond} it suffices to find expressions $\phi_{\vec \beta}(V_{i})$, $\psi_{\vec \beta}(V_{i})$ for $q$ and $k$ respectively. If $(\beta_{i-1}, \beta_{i-2})\neq (b, b^{-1})$, this follows directly from the rightmost column in \cref{lem:Bred_conditions}. Otherwise, we know that  $(\beta_{i-2}, \beta_{i-3})\neq (b, b^{-1})$ and so we obtain the expressions for $q$ and $k$ by applying \cref{lem:Bred_conditions} to $uv[i-2,i-2]$ (note that $uv[i-2,i-2] \in \BS12$ because $uv[i-1,i-1] \in \BS12$). This proves the lemma. 
\end{proof}

\subsection{The algorithm}\label{sec:algorithm}

A power circuit representation of $u \in \BG$ written as in (\ref{betaFact}) consists of the sequence $\cB = ( \beta_{h}, \ldots,\beta_{1} )$ and a power circuit $(\Pi, \delta_{\Pi})$ with markings $U_{i}, E_{i}, M_{i}$ for $i \in \interval{0}{h}$ such that $(U_{i},E_i)$ is a power circuit representation of $r_i$ (see \cref{def: pcrep}) and $m_{i} = \epsilon(M_{i})$.

\begin{lemma}\label{prop:Bredstep}
The following problem is in \PTc{0}{D} parametrized by $\max_i \depth(\Pi_i)$: 
	\compproblem
	{Britton-reduced power circuit representations of $u, v \in \BG$ over power circuits $\Pi_1$, $\Pi_2$.}	
	{A Britton-reduced power circuit representation of $w \in \BG$ over a power circuit $\Pi'$ such that $w \eqBG uv$ and $\depth(\Pi') =  \max_i \depth(\Pi_i)+\Oh(1)$ and $\abs{\Pi'} \in  \Oh(\abs{\Pi_1}+\abs{\Pi_2})$.} 
\end{lemma}
\begin{proof}
	Let $\Pi$ be the disjoint union of $\Pi_1$ and $\Pi_2$.
We need to find the maximal $i$ such that $uv[i,i] \in \BS12$. This can be done as follows: By \cref{lem: OpOnFP} one can evaluate the expressions $\theta_{\vec \beta}(V_{i})$ and $\xi_{\vec \beta}(V_{i})$ of \cref{lem:BredCond} and test the conditions $\theta_{\vec \beta}(V_{i}) \in \Z$ and $\xi_{\vec \beta}(V_{i}) = 0$ in \PTc{0}{D}.
For every $i$ this can be done independently in parallel giving us Boolean values indicating whether $uv[i-1,i-1] \in \BS12$ implies $uv[i,i] \in \BS12$. Now, we have to find only the maximal $i_0$ such that for all $j \leq i_0$ this implication is true. Since $uv[0,0] = 1 \in \BS12$, it follows inductively that  $uv[i,i] \in \BS12$ for all $i\leq i_0$. Moreover, as the implication $uv[i_0 ,i_0 ]  \in \BS12 \implies uv[i_0 +1, i_0 +1] \in \BS12$ fails, we have $uv[j,j] \notin \BS12$ for $j \geq i_{0}+1$. 

Now, using the expressions $\phi_{\vec \beta},\psi_{\vec \beta}$ from \cref{lem:BredCond} one can compute again using \cref{lem: OpOnFP} $(q, k) = (\phi_{\vec \beta}(V_{i_0}), \psi_{\vec \beta}(V_{i_0})) \eqBG uv[i_0,i_0] $ in \PTc{0}{D}. Again using \cref{lem: OpOnFP}, we can compute in \PTc{0}{D} $(s,m) = (r_{i_{0}},m_{i_{0}})(q, k)(s_{i_{0}},n_{i_{0}})$ as a power circuit representation over a power circuit $(\Pi', \delta_{\Pi'})$ with $(\Pi, \delta_{\Pi}) \leq (\Pi', \delta_{\Pi'})$, $\abs{\Pi'} \in \Oh(\Pi)$ and $\depth(\Pi')\in \depth(\Pi)+\Oh(1)$. Now, the output is
\[ (r_{h},m_{h}) \beta_{h}  \cdots(r_{i_0+1},m_{i_0+1})  \beta_{i_0+1}\; (s,m)\;   \tilde{\beta}_{i_0+1} (s_{i_0+1},n_{i_0+1}) \cdots  \tilde{\beta}_{\ell}(s_{\ell},n_{\ell} ).\qedhere \]
\end{proof}

Before showing Theorem~\ref{thm:main}, we prove the following slightly more general result. Recall that $\Sigma = \{1,a, a^{-1}, t, t^{-1}, b, b^{-1}\}$.
\begin{theorem}\label{thm:Bred}
	The following problem is in \Tc{2}:
\compproblem{A word $w \in \Sigma^{*}$.}{ A power circuit representation for a Britton-reduced word $w_{\text{red}} \in \Delta^{*}$ such that $w \eqBG w_{\text{red}} $ and the underlying power circuit has  depth $\Oh(\log \abs{w})$.}
\end{theorem}

\begin{proof}
Let $w = w_{1}\cdots w_{n}$ with $w_{j} \in \Sigma$ be some input. Since we can pad with the letter $1$, we can assume $n = 2^{m}$ for $m \in \mathbb{N}$.
The idea for the proof is simple: First, we transform each letter $w_{j}$ into a power circuit representation. After that, the first layer computes the Britton reduction of two-letter words using  \cref{prop:Bredstep}, the next layer takes always two of these Britton-reduced words and joins them to a new Britton-reduced word and so on. After $m = \log n$ layers we have obtained a single Britton-reduced word. By the bound in \cref{prop:Bredstep}, the size of the resulting power circuits stays polynomial in $n$ and their depth in $\Oh(\log n)$. In particular, each application of \cref{prop:Bredstep} is in \Tc1 and, hence, the whole computation is in \Tc{2}. 

Let us detail this high-level description a bit further: For $j \in \interval{1}{n}$
 we set $w_{j} = w^{(1)}_{j}$. Now for each word $w^{(1)}_{j}$ we construct its power circuit representation as follows: Let $(\Pi^{(1)}_{j}, \delta^{(1)}_{j})$ be a power circuit such that $\abssmall{\Pi^{(1)}_{j}} = 1$ for $j \in \interval{1}{n} $. The  we define  markings $U_{i}$, $E_{i}$ and $M_{i}$ as follows: If $w^{(1)}_{j} = a^{\alpha}$ for $\alpha \in \oneset{\mOne, 1}$ then $\epsilon(U_{i}) = \alpha$ and $\epsilon(E_{i}) = \epsilon(M_{i}) = 0$. If $w^{(1)}_{j} = t^{\alpha}$, then $\epsilon(U_{i}) = \epsilon(E_{i}) = 0$ and $\epsilon(M_{i}) = \alpha$. If $w^{(1)}_{j} = \beta$ with $\beta \in \{b, b^{-1}\}$, then all markings evaluate to $0$ and we set $\cB^{(1)}_{j} = (\beta)$ (in all other cases we define $\cB^{(1)}_{j}$ to be the empty sequence). If $w^{(1)}_{j} = 1$, then all markings evaluate to $0$ and $\cB^{(1)}_{j}$ is the empty sequence. 

\medskip
Now let $k \in  \interval{2}{m+1} $  and  $j \in \interval{1}{\frac{n}{2^{k-1}}}$ and assume that the words $w^{(k-1)}_{2j-1}$ and $w^{(k-1)}_{2j}$ are Britton-reduced with power circuit representations over $(\Pi^{(k-1)}_{2j-1}, \delta^{(k-1)}_{2j-1})$ and $(\Pi^{(k-1)}_{2j}, \delta^{(k-1)}_{2j})$, respectively. By \cref{prop:Bredstep} we can construct a power circuit representation for a Britton-reduced word  $w^{(k)}_{j}$ over a power circuit $(\Pi^{(k)}_{j}, \delta^{(k)}_{j})$  such that $w^{(k)}_{j} \eqBG w^{(k-1)}_{2j-1}w^{(k-1)}_{2j}$ and 
\begin{align*}
	\abssmall{\Pi^{(k)}_{j}} &\leq c_{s} \cdot \left(\abssmall{\Pi^{(k-1)}_{2j-1}}+\abssmall{\Pi^{(k-1)}_{2j}} \right)  \text{ and }\\
	 \depth(\Pi^{(k)}_{j}) &\leq  c_{d} +  \max(\depth(\Pi^{(k-1)}_{2j-1}),\depth( \Pi^{(k-1)}_{2j}))
\end{align*}
for constants $c_{s}$ and $c_{d}$. In order to bound the size and depth of these power circuits inductively, define $\Pi^{(k)}$ to be the disjoint union of the $\Pi_{j}^{(k)}$ for $j \in \interval{1}{\frac{n}{2^{k-1}}}$. It follows that 
\[\abssmall{\Pi^{(k)}} \leq c_{s} \cdot \abssmall{\Pi^{(k-1)}} \qquad \text{ and } \qquad \depth(\Pi^{(k)}) \leq \depth(\Pi^{(k-1)})+c_{d}.\]
Let $\nu \in \N$ with  $c_{s} \leq  2^{\nu}$. 
With $\abs{\Pi^{(1)}} = n$  and $\depth(\Pi^{(1)}) = 1$ we obtain that
\begin{align}
\depth(\Pi^{(k)}) &\leq \depth(\Pi^{(1)})+(k-1) \cdot c_{d} \leq 1+ m \cdot c_{d} = 1+\log(n)\cdot c_{d},	 	 \label{depthbound}\\
\nonumber\abssmall{\Pi^{(k)}} &\leq c_{s}^{k-1} \cdot \abssmall{\Pi^{(1)}} \leq  c_{s}^{m} \cdot n \leq  \left(2^{\nu}\right)^{\log(n)} \cdot n = n^{\nu+1}.
\end{align}
Therefore, the size of each $\Pi^{(k)}$ is polynomial in the input size $n$ and its depth is logarithmic in $n$. In particular, the same applies to the power circuits $\Pi^{(k)}_{j}$. Therefore, by \cref{lem:paraToTC}, \cref{prop:Bredstep} yields a \Tc{1} circuit for computing $w^{(k)}_{j}$ from $ w^{(k-1)}_{2j-1}$ and $w^{(k-1)}_{2j}$.
 For each $k$, all the power circuit representations of the $w^{(k)}_{j}$ for $j \in \interval{1}{\frac{n}{2^{k-1}}}$ can be computed in parallel with the bound on the depth given by (\ref{depthbound}).  Since we have $\Oh(\log n)$ of these stages, the over all complexity is \Tc{2}. 
\end{proof}

\begin{proof}[Proof of Theorem~\ref{thm:main}]
	In order to decide whether  $w \eqBG 1$, we first compute its Britton reduction $\hat w$ using \cref{thm:Bred}. If $\hat w$ still contains some $b$ or $b^{-1}$, by \hyperref[lem:Britton]{Britton's Lemma}, we know that $w \neqBG 1$. Otherwise, $\hat w = (r,m)\in \Z[1/2] \rtimes \Z$ where $r,m$ are given as there power circuit representations over a power circuit $\Pi$ of depth $\Oh(\log \abs{w})$.
According to \cref{lem: OpOnFP}, we can check in \Tc{1} whether $r = m = 0$.  
\end{proof}

\subparagraph{The compressed word problem.}
The \emph{compressed word problem} of a group is similar to the ordinary word problem. However, the input element is not given directly but as a straight-line program. A \emph{straight-line program} is a context-free grammar which generates exactly one word.
The compressed word problem for a group $G$ with presentation $\eta\colon  \Sigma^* \to G$ is as follows:
\problem{A straight-line program generating a word $w \in \Sigma^*$}{Is $w =_G 1$?}

\begin{corollary}
	The compressed word problem of the Baumslag group is in \PSPACE.
\end{corollary}
\begin{proof}
It is an easy exercise that all the circuit families we described are, indeed, \LOGSPACE-uniform. In particular, the word problem of \BG is in \LOGSPACE-uniform \Tc{2}. Since \LOGSPACE-uniform \Tc{2} is contained in $\DSPACE(\log^3n)$, we can apply \cite[Lemma 34]{BartholdiFLW20} in order to obtain the corollary.
\end{proof}

\section{Hardness of comparison in power circuits}

The main result of this section is to show how Boolean circuits can be simulated by power circuits. This leads to \P-completeness of the comparison problem for power circuits (Theorem \ref{thm:compIntro}). In this section we consider functions computable in \DLOGTIME-uniform \Ac0. The reader unfamiliar with the precise definitions might simply think of \LOGSPACE-computable.  We start by introducing some normalization steps for Boolean circuits.

For a Boolean circuit $\cC$ with  input gates $x_1, \ldots, x_n$ and some $\vec{a} \in \{0,1\}^n$ we write $\eval_{\vec{a}}(\mathcal{C})$ for the evaluation of the output gate of $\cC$ when assigning $\vec{a}$ to the inputs.

\subparagraph{Elimination of \gand-gates.} By de Morgan's rule we can simulate each \gand-gate by a circuit of depth 3 using an $\gor$-gate and $\gnot$-gates. So for each $\Ac{}$-circuit $\cC$ of depth $D$ there is an equivalent Boolean circuit $\cC'$ of depth at most $3D$ using only $\gor$ and $\gnot$-gates such that $\eval_{\vec{a}}(\mathcal{C})=\eval_{\vec{a}}(\mathcal{C'})$ for all $\vec{a} \in \{0,1\}$.  Moreover, the circuit $\cC'$ clearly can be computed in \DLOGTIME-uniform \Ac0.

\subparagraph{Layered circuits.} A circuit is called layered if we can assign a level number to each gate such that input gates are on level 0 and gates on level $k$ only receive inputs from level $k-1$.

Given an arbitrary $\Ac{}$-circuit $\mathcal{C}$ of depth $D$, we can construct a layered $\Ac{}$-circuit $\cC'$ of depth $D$ as follows: We first make $D+1$ copies of all gates of $\mathcal{C}$ numbered from $0$ to $D$. The input gates of $\cC'$ are the input gates in copy $0$. Then we introduce wires between the gates as in the original circuit, but only between copy $i$ and copy $i+1$. Moreover, for $k \geq 1$ we replace an input gate in copy $k$ by a fan-in one \gor-gate which receives its input from the corresponding input gate in copy $k-1$. The output gate of $\cC'$ is the output gate in copy $D$.
So we obtain a layered $\Ac{}$-circuit of depth $D$ and size $(D+1)\cdot\abs{\mathcal{C}}$. Because the paths that connect input gates with output gates are the same in both circuits,  the new circuit evaluates to $1$ if and only if this is the case for $\mathcal{C}$.  

Notice that we also can perform this construction if $D$ is not the exact depth but only an upper bound. Moreover, if $D$ is given in the input, the construction can be computed in \DLOGTIME-uniform \Ac0. 
Also note that if $\cC$ uses only  $\gor$ and $\gnot$-gates, then also the layered circuit will only use those gates.

\newcommand{\compNodeA}{A}
\newcommand{\compNodeB}{B}

\begin{theorem}\label{thm:PCsimulation}
	Let $\mathcal{C}$ be a layered $\Ac{}$-circuit made of unbounded fan-in \gor-gates and \gnot-gates of size $L$ and depth $D$ and input gates $x_1, \ldots, x_n$.
	There exists a power circuit $(\Gamma, \delta)$ with special vertices $V_1, \dots, V_n$ and $\top$, $\compNodeA$, and $\compNodeB$ satisfying the following properties:
	
	 For $\vec{a} \in \{0,1\}^n$ we define a graph $(\Gamma, \delta_{\vec{a}})$ where $\delta_{\vec{a}}(V_i, \top) = a_i$ and on all other nodes  $\delta_{\vec{a}}$ agrees with $\delta$.
 Then for all $\vec{a} \in \{0,1\}^n$  we have 
	\begin{itemize}
		\item $(\Gamma, \delta_{\vec{a}})$ is a power circuit,
		\item $\depth(\Gamma,\delta_{\vec{a}}) = 2D+\ell+1$ and $\abs{\Gamma} \leq 3(L+D)+\ell+3$,
		\item $P \neq Q \Rightarrow \epsilon(P)\neq \epsilon(Q)$ for all nodes $P,Q \in \Gamma$,  
		\item  $\epsilon(\compNodeA) \leq \epsilon(\compNodeB)$ if and only if $\eval_{\vec{a}}(\mathcal{C}) = 1$. 
	\end{itemize}
	Moreover, given $\cC$, the power circuit $(\Gamma, \delta)$ can be computed in \LOGSPACE.
\end{theorem}
Note that if each gate actually ``knows'' its layer (\ie if the layer number is part of the gate number), $(\Gamma, \delta)$ can be computed even in \DLOGTIME-uniform \Ac0.

Intuitively, \cref{thm:PCsimulation} states that the family of power circuits $((\Gamma, \delta_{\vec{a}}))_{\vec{a} \in \{0,1\}^n}$ simulates the circuit $C$.

\begin{corollary} \label{cor:compRedPCAC}
	Let $k\geq 1$. The following problem is in $\Tc{k}$ and it is hard for $\Ac{k}$ under $\Ac{0}$-Turing reductions:
	\problem{ A power circuit $(\Pi, \delta_{\Pi})$ with $\depth(\Pi) \leq \log^k \abs{\Pi}$ and nodes $A, B \in \Pi$ such that for all $P \in \Pi$ the marking $\Lambda_{P}$ is compact and for all $P \neq Q$, $\epsilon(P) \neq \epsilon(Q)$.}{Is $\eps(A) \leq \eps(B)$?}
\end{corollary}\vspace{-1mm}
\cref{cor:compRedPCAC} shows that \cref{thm:pcred} is almost optimal~-- only the space between $\Ac{k}$ and $\Tc{k}$ is possibly for the ``true'' complexity of comparison in $\log^k$-depth power circuits.
\begin{proof}
	Membership in $\Tc{k}$ is by \cref{cor:compareMarkings}.
	As in \cite[Theorem 4.22]{Vollmer99} we see that $\Ac{k} = \Ac{0}(\Ac{k})$ where for a fixed circuit family the oracle gates for $\Ac{k}$ can be assumed to come from a fixed language in $\Ac{k}$ with layered circuits of depth $\epsilon \log^k n$ for some small enough $\epsilon>0$. By \cref{thm:PCsimulation} for each of these circuits there is a power circuit of depth $\log^k n$ simulating it as stated in the theorem. Thus, we can replace any of these depth-$\epsilon \log^k n$ circuits by an oracle gate for the comparison problem in power circuits where the actual power circuit is hardwired and only depends on the input as described in \cref{thm:PCsimulation}.
\end{proof}

\begin{proof}[Proof of \cref{thm:PCsimulation}]
 We start with a layered $\Ac{}$-circuit $\cC$ of size $L$ and depth $D$ consisting of \emph{input gates} $x_i$ for $i=1,\dots, n$, \gnot gates, $\gor$-gates (of fan-in at most $L$); one of these gates is marked as \emph{output gate}. Each gate is on a unique level: input gates on level 0, and gates on level $k$ only receive inputs from level $k-1$; the output gate is on level $D$. We assume that the gates are numbered from $1$ to $L$ with gates $1$ to $n$ being the input gates (\ie $g_i = x_i$ for $i\in \interval{1}{n}$).

 We write $\ell=\log^{*}(L)+3$. We assume that $L \geq 3$. Then the following inequalities hold for all $k \geq 0$: 
   	\begin{align}
		\begin{aligned}\label{eq:PCassumption}
2L & \leq 2^{L}\\
2^{L+2}&\leq \tow(\ell+k)\\
			\tow(k - 1 + \ell )&\leq \tow(k+\ell) / 2-2^{L + 1}  \\
		\end{aligned} 
	\end{align}

\proofsubparagraph{Invariants of $(\Gamma, \delta_{\vec{a}})$:}
Starting with the circuit $\mathcal{C}$, we design a power circuit $(\GG,\delta)$ with designated nodes $V_{1}, \ldots, V_{n}$ such that for every gate $g$ on level $k$, there is some node $P_g$ such that for all $\vec{a} \in \{0,1\}^n$ the following inequalities hold:
	\begin{align}
		\begin{aligned}\label{eq:Pcomp}
		\tow(\ell-1) < \e(\Lambda_{P_g}) &\leq \tow(2k+1+\ell)-L,\\ 
			\e(\Lambda_{P_g}) & \leq  \tow(2k+\ell)-L &&\text{ if }\eval_{\vec{a}}(g) = 0,\\                                                                                                          
			\e(\Lambda_{P_g})   &\geq \tow(2k+\ell)  &&\text{ if } \eval_{\vec{a}}(g) = 1.                                                                                                      
		\end{aligned} 
	\end{align}
Be aware that $\epsilon$ here denotes the evaluation in $(\Gamma, \delta_{\vec{a}})$.

Note that, in particular, if $g$ is the output gate (which is on level $D$), then $\e(P_g)   \geq \tow(2D + 1 + \ell )$ if and only if $ \eval_{\vec{a}}(g)=1$.

\proofsubparagraph{The construction of $(\Gamma, \del)$:}	
\begin{itemize}
	\item We first create a reduced power circuit consisting of nodes $X_{0}, \ldots, X_{L}$ such that $\epsilon(\Lambda_{X_{i}}) = i$ (\ie  $\epsilon(X_{i}) = 2^i$) and the successor markings $\Lambda_{X_{i}}$ are compact. 
	
\item  We create nodes $T_{0}, \ldots, T_{2D+1+\ell}$  such that $\epsilon(T_{0}) = 1$ and $\delta(T_{i},T_{i-1}) = 1$ for $i \in \interval{1}{2D+1+\ell}$ and $\delta(T_{i},Q) = 0$ otherwise. Then $\epsilon(T_{i}) = \tow(i)$. If there exists a node $X_i$ such that  $\tow(j) = \epsilon(X_i)$, then we set $T_{j} = X_i$.
\item As an abbreviation, we will denote a node $T_{2k+\ell}$ by $R_{k}$ for $k \in \interval{1}{D}$. Note that then $\epsilon(\Lambda_{R_k}) = \tow(2k-1+\ell)$.   Moreover, we write $\top$ for $T_{\ell}$.
\item For every $k \in \interval{1}{D}$ we create a node $S_k$ with $\delta(S_k, R_{k-1}) = 1$ and $\delta(S_k,X_0)=-1$. So $\epsilon(\Lambda_{S_k}) = \tow(2k-2+\ell)-1$. Note that $\Lambda_{S_k}$ is a compact marking and $\epsilon(\Lambda_{S_k})> \tow(\ell-1)$ by the choice of $L$.

\item For every input gate $x_{i}$ we create a node $V_{i}$ with $\delta(V_{i},X_i)=\delta(V_{i} ,T_{\ell-1})=1$.
Notice that $2^L \leq 2^{\tow(\log^*L)} = \tow(\ell -2)$, so $X_i$ and $T_{\ell-1}$ are guaranteed to be distinct nodes and, thus, $\delta$ is well-defined.
 We write $P_{g_i}$ as an alias for $V_i$.
Note that by the definition of $\delta_{\vec{a}}$ in the theorem we have $\delta_{\vec{a}}(V_{i}, \top)=a_i$.  

		\item For every $\gor$-gate $g_{i}$ with incoming edges from gates $u_1, \ldots, u_h$ we create nodes $P_{g_i}, Q_{g_i} \in \GG$ and set $\del(P_{g_i},Q_{g_i}) =1$ and $ \del(Q_{g_i},P_{u_j}) =1$ for $1\leq j \leq h$. In addition we set $\delta(Q_{g_i},X_i)=\delta(P_{g_i},X_i)
=1$.		
		\item  For every $\gnot$-gate $g_i$ on level $k$ with incoming edge from  gate $u$, we create a node $P_{g_i} \in \GG $ and set $\del(P_{g_i},R_k) = \del(P_{g_i},S_k) = 1$, and  $\del(P_{g_i},P_u) = -1$.  In addition, we set $\delta(P_{g_i},X_i)
=1$.		
	\end{itemize}
Finally, we set $\compNodeA = T_{2D+ 1+\ell }$ and $\compNodeB = P_g$ where $g$ is the output gate. Thus, once we have proved \eqref{eq:Pcomp}, we know that   $\e(\compNodeA)   \leq \e(\compNodeB)$ if and only if $ \eval_{\vec{a}}(\cC) = 1$.

\begin{figure}
		 \begin{minipage}[t]{.45\linewidth}	
			\begin{center}
			\tikzstyle{pcnode} = [minimum size = 13pt,circle,draw ]
			\tikzstyle{mknode} = [minimum size = 13pt,rectangle,draw ]	\begin{tikzpicture}[outer sep = 0pt, inner sep = 0.7pt, node distance = 40]
					{\footnotesize
						\node[pcnode] (q)    {\footnotesize $Q_{g_i}$};
						\node[pcnode, above left of = q, xshift=-12, yshift=0] (1)   {\footnotesize $P_{u_1}$};
						\node[pcnode,  right of = 1, xshift=-12] (2)    {\footnotesize $P_{u_2}$};

						\node[ right of = 2,xshift=-12] (4)    {\footnotesize $\cdots$};
							\node[pcnode, right of = 4, xshift=-12] (5)    {\footnotesize $P_{u_h}$};
						\node[pcnode, below of = q,yshift=5] (p)   {\footnotesize $P_{g_i}$};
						
\node[pcnode,  above right of=5, yshift=-8] (x){$X_i$}	;					

						\draw[->] (q)edge node[below left,yshift=3,xshift=-2] {+} (1)
						(q) edge node[below left, xshift=10, yshift=3] {+} (2)

												(q) edge node[below left,yshift=3, xshift=12] {+} (5)
												
(q) edge[bend right]  node[below left,yshift=6, xshift=13] {+} (x)
			
(p) edge[bend right]  node[below left,yshift=2, xshift=14] {+} (x)
						(p) -- node[below left] {+} (q)
;

;					}
				\end{tikzpicture}
				\caption{Power circuit for $\gor$-gates.}\label{fig:orgate}
			\end{center}
			  \end{minipage}%
  \hfill%
  \begin{minipage}[t]{.45\linewidth}

				\begin{center}
				\tikzstyle{pcnode} = [minimum size = 13pt,circle,draw ]
				\tikzstyle{mknode} = [minimum size = 13pt,rectangle,draw ]
				\begin{tikzpicture}[outer sep = 0pt, inner sep = 0.7pt, node distance = 38]
					{\footnotesize
						\node[pcnode] (u)    {\footnotesize $P_u$};
						\node[pcnode, below of = u, node distance = 66] (p)   {\footnotesize $P_{g_i}$};
						\node[pcnode, above left of = p, node distance = 52] (t)   {\footnotesize $R_k$};	
						\node[pcnode, above right of = p] (s)    {\footnotesize $S_k$};
		
\node[pcnode,  above right of=s, yshift=23] (x){$X_i$}	;				
						\draw[->] (p) edge node[above left] {+} (s)
						(p) edge node[below left] {+} (t)
(p) edge node[left, xshift = -3] {$-$} (u)
	(p) edge[bend right=40,] node[left, xshift = 12] {$+$} (x)				;
						
					}
				\end{tikzpicture}
				\caption{Power circuit for $\gnot$-gates.}\label{fig:notgate}
			\end{center}
			 \end{minipage}%
		\end{figure}

\proofsubparagraph{Size and depth bounds:} 
Observe that for every gate $g_i$ of $\mathcal{C}$ we introduce at most two nodes $P_{g_i}$ and $Q_{g_i}$ plus the node $X_i$ in $\Gamma$. So, by adding the number of nodes $T_i$ and $S_i$ plus $X_0$, we obtain that $\abs{\Gamma}\leq 3(L+D)+\ell+3$.

	In the following we define the depth of a node as the length of a longest path starting from it (\ie the depth of $\Gamma$ is the maximal depth of its nodes).
	Each node $X_i$ for $i \in \interval{0}{L}$ has depth at most $\ell$ (see also \cref{ex:binarybasis}).  Each of the nodes $T_k$ has depth $k$, so here we obtain nodes of depth at most $2D+\ell +1$. 
	The node $S_k$ only points to the node $R_{k-1}$ and to $X_0$, so it also has depth $2k - 1 + \ell$ with $k \leq D$. 
	
	By induction we see that if $g$ is a gate on level $k$, then the depth of $P_g$ is at most $2k + \ell + 1$ Therefore, for each $\vec{a} \in \{0,1\}^n$ the depth of $(\Gamma, \delta_{\vec{a}})$ is exactly $2D+\ell +1$.

\proofsubparagraph{Complexity of the construction:} 

The whole construction can be done in \LOGSPACE: We can compute the level of each gate and the depth of the circuit in \LOGSPACE just by following any path from the gate to the input gates. Since the circuit is layered this always gives the same result. 

The actual construction of $(\Gamma, \delta)$ is straightforward: We start by creating only the nodes, without the edges. First create the nodes $X_i$ for $i \in \interval{0}{L}$. Then we add the nodes $T_0, \ldots, T_{2D+1+\ell}$, and $S_1, \ldots, S_{D}$. Here we need to be careful and identify each node $T_i$ with $X_{\tow(i-1)}$ if it exists. Since $\log^*$ can be computed in \LOGSPACE, both $\ell$ and this identification can be computed in \LOGSPACE.
Now it remains to create the nodes $P_g$ and (only for $\gor$-gates) $Q_g$ corresponding to the gates, which also clearly can be done in \LOGSPACE. For every node the outgoing edges can be determined in a straightforward way from the definition.

Let us briefly outline that the construction is also in \DLOGTIME-uniform \Ac0 given that each gate contains its level as part of the number. 
	In this case, the depth of the circuit can be computed as the maximum over all the levels. 
	Now, note that for a $\log n$-bit number $i$, the number $\log^*i$ can be computed in time $\Oh(\log n)$. Hence, $\ell$ can be computed in \DLOGTIME
and the whole construction is straightforward in \Ac0. For the identification of the nodes $T_i$ with $X_{\tow(i-1)}$ we use again that  $\log^*i$ can be computed in \DLOGTIME.  In particular, we can decide in \DLOGTIME whether $T_i = X_j$ for given $i$ and $j$. Therefore, even in  \DLOGTIME-uniform \Ac0 we can hard-wire this identification (notice that these nodes only depend on the size $L$ but not on the actual circuit).

\proofsubparagraph{Proof of \eqref{eq:Pcomp} (correctness):} Let $\vec{a}\in\{0,1\}^n$ be some input to $\mathcal{C}$. Let us show \cref{eq:Pcomp} for every power circuit $(\Gamma, \delta_{\vec{a}})$ by induction starting from the input gates (notice that \cref{eq:Pcomp} shows that $(\Gamma, \delta_{\vec{a}})$ is a power circuit, indeed). Let $g_i = x_i$ be an input gate and first assume that $a_i=0$. Then using \eqref{eq:PCassumption} we obtain
\begin{align*}
\epsilon(\Lambda_{V_{i}}) = \epsilon(T_{\ell-1})+\epsilon(X_{i}) = \tow(\ell-1)+2^{i} < \tow(\ell)-L 
\end{align*}
and $ \epsilon(\Lambda_{V_{i}}) > \tow(\ell-1)$. 	
Now we assume that $a_i = 1$. Then again by \eqref{eq:PCassumption} we have
\begin{align*}
 &\epsilon(\Lambda_{V_{i}}) =\epsilon(T_{\ell-1})+ \epsilon(T_{\ell})+\epsilon(X_{i}) = \tow(\ell-1)+\tow(\ell)+2^{i} >\tow(\ell)\qquad \text{ and }\qquad   \\&\epsilon(\Lambda_{V_{i}}) \leq 2\cdot \tow(\ell)+2^{L}\leq \tow(\ell+1)-L.
\end{align*}
This shows \cref{eq:Pcomp} for nodes $V_i = P_{g_i}$ corresponding to input gates.
		
	Now let $g_i$ be some $\gor$-gate on level $k \geq 1$ with inputs from gates $u_1, \ldots, u_h$. First assume that $\eval_{\vec{a}}(u_j)=0$ for all $j \in \interval{1}{h}$. By induction, $\epsilon(\Lambda_{P_{u_j}}) > \tow(\ell - 1)$ for all $j \in \interval{1}{h}$; hence, also $\epsilon(\Lambda_{g_i}) > \tow(\ell - 1)$. Moreover, by induction, we have $\epsilon(\Lambda_{P_{u_j}})\leq \tow(2k-2+\ell)-L $ for all $j \in \interval{1}{h}$ meaning that
\[\epsilon(P_{u_j})=2^{\epsilon(\Lambda_{P_{u_j}})} \leq 2^{\tow(2k-2+\ell)-L}  = \tow(2k-1+\ell)/2^{L} \leq \tow(2k-1+\ell)/(2L).\]
The last inequality is by \eqref{eq:PCassumption}. We know that $h \leq L$, so
\allowdisplaybreaks\begin{align*}\allowdisplaybreaks
\epsilon(\Lambda_{Q_{g_i}})&=\left(\sum_{j=1}^{h}\epsilon(P_{u_j})\right) + \epsilon(X_i)\\
	&\leq\left( \sum_{j=1}^{h} \frac{1}{2L}\cdot\tow(2k-1+\ell)\right)+2^{i}\\
	&\leq \frac{1}{2}\cdot\tow(2k-1+\ell)+2^{L}\\
	&\leq  \tow(2k-1+\ell)-L\tag{by \eqref{eq:PCassumption}}
\end{align*}
and
\begin{align*}\epsilon(\Lambda_{P_{g_i}})&=2^{\epsilon(\Lambda_{Q_{g_i}})} +\epsilon(X_i)\leq 2^{\tow(2k-1+\ell)-L}+2^{i}\\
	&=\tow(2k+\ell)/2^{L}+2^{i}\\ 
	&\leq\tow(2k+\ell)-L. \tag{by \eqref{eq:PCassumption}}  
\end{align*}
Now assume that $\eval_{\vec{a}}(u_j)=1$ for some $j \in \interval{1}{h}$. The same argument as above also shows the bound $\epsilon(\Lambda_{P_{g_i}}) \leq \tow(2k+1+\ell)-L$ in this case.
 Furthermore, by induction, we have $\epsilon(\Lambda_{P_{u_j}})\geq \tow(2k-2+\ell)$. Hence, 
\begin{align*}
\epsilon(\Lambda_{Q_{g_i}})&=\sum_{j=1}^{h}\epsilon(P_{u_j})+\epsilon(X_i)  > 2^{\tow(2k-2+\ell)} = \tow(2k-1+\ell) \qquad\text{ and}\\
\epsilon(\Lambda_{P_{g_i}})&=\epsilon(Q_{g_i})
+\epsilon(X_i)\geq 2^{\tow(2k-1+\ell)}=\tow(2k+\ell).
\end{align*}

\smallskip

Next, let $g_i$ be a $\gnot$-gate  on level $k \geq 1$ with an incoming edge from gate $u$. Assume that $ \eval_{\vec{a}}(u)=0$. Then, $\epsilon(\Lambda_{P_u})\leq \tow(2k-2+\ell)-L$. So $\epsilon(P_u)\leq \tow(2k-1+\ell)/2^{L}$ and
\begin{align*}
\epsilon(\Lambda_{P_{g_i}})&=\epsilon(R_k)+\epsilon(S_k)-\epsilon(P_u)+\epsilon(X_i)\\
	&\geq \tow(2k+\ell)+\tow(2k-1+\ell)/2-\tow(2k-1+\ell)/2^{L}\\
	&\geq\tow(2k+\ell).
\end{align*}
In addition, we have
\begin{align*}
\epsilon(\Lambda_{P_{g_i}})&=\epsilon(R_k)+\epsilon(S_k)-\epsilon(P_u)+\epsilon(X_i)\\
	&\leq \epsilon(R_k)+\epsilon(S_k)+\epsilon(X_i)\\
	&\leq\tow(2k+\ell)+\tow(2k-1+\ell)/2+2^{L}\\
&\leq \tow(2k+1+\ell)-L. \tag{by \eqref{eq:PCassumption}}
\end{align*}
Now, assume that $ \eval_{\vec{a}}(u)=1$. Then by induction $\epsilon(\Lambda_{P_u})\geq \tow(2k-2+\ell)$. Hence,
\begin{align*}
\epsilon(\Lambda_{P_{g_i}})&=\epsilon(R_k)+\epsilon(S_k)-\epsilon(P_u)+\epsilon(X_i)\\
	&\leq \tow(2k+\ell)+\tow(2k-1+\ell)/2-\tow(2k-1+\ell)+2^{L}\\
	&\leq \tow(2k+\ell)-\frac{1}{2}\cdot \tow(2k-1+\ell) +2^{L}\\ 
	&\leq\tow(2k+\ell)-L. \tag{by \eqref{eq:PCassumption}}
\end{align*}
Finally, we know that $\epsilon(\Lambda_{P_u})\leq \tow(2k-1+\ell)-L$. So, 	
\begin{align*}
\epsilon(\Lambda_{P_{g_i}})&=\epsilon(R_k)+\epsilon(S_k)-\epsilon(P_u)+\epsilon(X_i)\\
	&\geq \tow(2k+\ell)+\tow(2k-1+\ell)/2-\tow(2k+\ell)/2^{L} \\
	&\geq \frac{2^{L}-1}{2^{L}}\tow(2k+\ell) > \tow(\ell-1).
\end{align*}

\proofsubparagraph{Uniqueness of evaluations.} It remains to show that no two nodes have the same evaluation and that each successor marking is compact.
Let $P,Q \in \Gamma$. Now we have to show that $\epsilon(P)\neq \epsilon(Q)$ if $P\neq Q$. First, we set $\mathcal{X} = \set{X_i, T_{j}}{ i \in \interval{0}{L}, j \in \interval{0}{\ell}}$. By construction, $\epsilon(P)\neq \epsilon(Q)$ if $P\neq Q$ is clear for $P,Q \in \mathcal{X}$ (independently of $\vec{a}\in\{0,1\}^n$). 
 We further know that $\epsilon(P) > \tow(\ell)$ for all nodes $P \in \Gamma \setminus \mathcal{X}$ and $\epsilon(P) \leq \tow(\ell)$ for $P \in  \mathcal{X}$; so in particular, $\epsilon(P) \neq \epsilon(Q)$ for all nodes $P \in \Gamma \setminus \mathcal{X}$ and $Q \in \mathcal{X}$.
 
  For  $P \in \Gamma\setminus \mathcal{X}$ we can conclude that  $\tow(\ell)$ divides $ \epsilon(P)$ and so $\epsilon(P) \equiv 0 \mod \tow(\ell)$ (note that $\epsilon(P)$ is a power of two). 
Now let $u_i$ be an arbitrary gate and $v_j$ be an arbitrary $\gor$-gate. Considering the successor markings of the nodes $P_{u_i}$ and $Q_{v_j}$ we obtain the following:  
	\begin{align}
		\begin{aligned}\label{eq:PcompModulo}
			\e(\Lambda_{P_{u_i}}) & \; \equiv &2^i &\mod \tow(\ell),\\                                                                                                          
			\e(\Lambda_{Q_{v_j}})   &\;\equiv &2^j &\mod \tow(\ell),\\
			\epsilon(\Lambda_{T_r}) &\;\equiv& 0 &\mod \tow(\ell)\qquad \text{ for } r \geq \ell + 1,\\
			\epsilon(\Lambda_{S_k})&\;\equiv  &-1 &\mod \tow(\ell).                                                                                                      
		\end{aligned} 
	\end{align}
By the choice of $\ell$ we know that $2^i \not \equiv \alpha \mod \tow(\ell)$ for $\alpha \in \oneset{-1,0} $ and $i \in \interval{1}{L}$, and $2^i \not \equiv 2^j \mod \tow(\ell)$ for $i,j \in \interval{1}{L}$ and $i\neq j$. Since all nodes of $\Gamma$ are of the above form  (remember that for an input gate $g_i$ we wrote $V_i=P_{g_i}$ and also $R_i$, $\top$, $\compNodeA$ and $\compNodeB$ were just aliases for other nodes), $\epsilon(P)\neq \epsilon(Q)$ whenever $P\neq Q$.
Note that these observations hold independently of  $\vec{a}\in\{0,1\}^n$.

The successor markings of the nodes $X_i$, $T_i$, $S_i$ are compact by construction.
By \eqref{eq:PcompModulo} we have $\abs{\epsilon(\Lambda_{P})-\epsilon(\Lambda_{Q})} \geq 2$ for all nodes $P, Q \in \Gamma \setminus (\mathcal{X} \cup \{S_1, \dots, S_{D}\})$ with $P \neq Q$. Moreover, $\abs{\epsilon(\Lambda_{P})-\epsilon(\Lambda_{X_i})} \geq 2$ for all nodes $P \in \Gamma \setminus \mathcal{X}$ and all $i\in \interval{1}{L} $. This shows that $\Lambda_{ P_{i}}$ and $\Lambda_{ Q_{i}}$ are compact for \gor-gates $g_i$. In order to see that the successor markings of nodes corresponding to \gnot-gates are compact, observe that also  $\abs{\epsilon(\Lambda_{R_k})-\epsilon(\Lambda_{S_k})} \geq 2$ and  $\abs{\epsilon(\Lambda_{P_{g_i}})-\epsilon(\Lambda_{S_k})} \geq 2$ for all $k \in \interval{1}{D}$ and $i \in \interval{1}{L}$. Finally, successor markings of nodes corresponding to input gates are compact because $\e(\Lambda_{X_i}) + 2 \leq L +2 < 2^{L} \leq \tow(\ell - 2) = \e(\Lambda_{T_{\ell-1}}) < \tau(\ell-1) -2 = \e(\Lambda_{T_{\ell}}) - 2$.
Thus, all successor markings are compact.
\end{proof}

\subsection{P-hardness of power circuit comparison}\label{sec:hardness}
Finally, let us apply \cref{thm:PCsimulation} in order to prove some \P-hardness results on comparison in power circuits. Here, we use \LOGSPACE-reductions.

\begin{corollary} \label{thm:compRedPC}
	The following problem is \P-complete:
	\problem{ A power circuit $(\Pi, \delta_{\Pi})$ and nodes $\compNodeA, \compNodeB \in \Pi$ such that for all $P \in \Pi$ the marking $\Lambda_{P}$ is compact and for all $P \neq Q$, $\epsilon(P) \neq \epsilon(Q)$.}{Is $\eps(\compNodeA) \leq \eps(\compNodeB)$?}
\end{corollary}\vspace{-1mm}
Note that a weaker form of this result already has been stated in the second author's dissertation \cite{Weiss15diss}.
\begin{proof}
Membership in \P is by \cite[Proposition 49]{MyasnikovUW12} (or \cref{cor:compareMarkings} together with the observation that the circuit family we construct is uniform). \P-hardness is by \cref{thm:PCsimulation} since the circuit value problem (with circuits of unrestricted depth) is \P-complete.
\end{proof}

Notice that the only feature the power circuit in \cref{thm:compRedPC} lacks for being reduced is the sorting of the nodes.
In particular, under the assumption $\NC \neq \P$, it is not possible to sort the nodes of a given power circuit in \NC. 
\begin{remark}
	\begin{enumerate}[(a)]
		\item It is an immediate consequence of \cref{thm:compRedPC} that the comparison problem of two markings in a power circuit is \P-complete. This is because for two nodes $\compNodeA$ and $\compNodeB$ in a power circuit $(\Pi, \delta_{\Pi})$ we have $\epsilon(\compNodeA)\leq \epsilon(\compNodeB)$ if and only if $\epsilon(\Lambda_\compNodeA) \leq \epsilon(\Lambda_\compNodeB)$. 
		\item If the input is given as in \cref{thm:compRedPC}, we can check in \Ac0 whether $\epsilon(\compNodeA) = \epsilon(\compNodeB)$
		because this is the case if and only if  $\Lambda_{\compNodeA}(P)=\Lambda_{\compNodeB}(P)$ for all $P \in \Gamma$ (see \cref{lem:compareM}). This can be viewed as a hint that also in an arbitrary power circuit testing for equality might be easier than comparing for less than.
	\end{enumerate}
\end{remark}

\begin{corollary} 
	The following problem is \P-complete:
	\problem{A dag $\Gamma$ such that each node has a successor marking.}{Is $\Gamma$ already a power circuit? }
\end{corollary}

\begin{proof}
	The comparison problem for power circuits (\cref{thm:compRedPC}) can be reduced to this problem in a straightforward way:
	As input we have a power circuit $(\Pi, \delta_{\Pi})$ with two nodes $\compNodeA,\compNodeB\in \Pi$. Then we construct a dag as follows: We take $(\Pi, \delta_{\Pi})$ and add a node $R$ with $\delta(R,\compNodeA) = 1$ and $\delta(R,\compNodeB) = \mOne$. Then $\epsilon(\Lambda_{R}) \geq 0$ if and only if $\epsilon(\compNodeA) \geq \epsilon(\compNodeB)$, and so the newly defined dag is a power circuit if and only if $\epsilon(\compNodeA) \geq \epsilon(\compNodeB)$. 
\end{proof}

\begin{corollary}\label{cor:membershipBS}
	The following problem is \P-complete:
	\problem{A \PC representation of $w \in \BG$.}{Is $w \in \BS12$?}
\end{corollary}

\begin{proof}
	By \cite{muw11bg,DiekertLU13ijac}, the problem is in \P since the algorithms for the word problem also work with power circuit representations as input and they can be used to decide membership in $\BS12$ (this is due to \hyperref[lem:Britton]{Britton's Lemma}). Thus, it remains to show the hardness part.
	
	We reduce from the comparison problem for power circuits (\cref{thm:compRedPC}). So let $(\Pi, \delta_{\Pi})$ be a power circuit and let $M$ be a marking on $\Pi$. Now we consider the word $w = b(2^{\epsilon(M)},0)b^{-1} \in \Delta^*$. Then $(\Pi, \delta_{\Pi})$ together with the marking $M$ is a power circuit representation of $w$. For $w$ to be in $\BS12$ we need the $b$'s to cancel. This happens if and only if $\epsilon(M)\geq 0$. So, $w \in \BS12$ if and only if  $\epsilon(M) \geq 0$. 
\end{proof}

As a last result, let us state a straightforward lower bound for the problem of checking two markings for equality. Note that here we even need an arbitrary power circuit as input and cannot restrict it as in \cref{thm:compRedPC}.

\begin{proposition}\label{prop:eqNLhard}
	The following problem is \NL-hard:
	\problem{Given a power circuit and markings $M,K$.}{Is $\eps(M) = \eps(K)$?}
\end{proposition}
\begin{proof}
	Reduce the $s$-$t$-connectivity problem for dags to this problem. Given a dag $G = (V,E)$ and vertices $s,t \in V$ make two copies of this graph and add a label $+1$ to every edge. In one copy of $G$ add an additional node $P$ and an edge $(t,P)$. Let $s_1, s_2$ denote the two copies of $s$. Then $\eps(s_1) = \eps(s_2)$ if and only if there is no path from $s$ to $t$ in $G$.
\end{proof}

\section{Conclusion}\label{sec:conclusion}

We showed that the word problem of the Baumslag group can be solved in \Tc{2}. The proof relies on the fact that all power circuits used during the execution of
the algorithm have logarithmic depth. The comparison problem for such power circuits is in \Tc1, although for arbitrary power circuits it is \P-complete. We conclude with some open problems:
\begin{itemize}
	\item Is it possible to reduce the complexity of the word problem of the Baumslag group any further~-- \eg to find a \LOGSPACE algorithm? 
	\item 
Can we prove some non-trivial lower bounds (the word problem is \Nc1-hard as \BG contains a non-abelian free group \cite{Robinson93phd})?
	\item The problem of comparing two markings on a power circuit for equality is \NL-hard -- is it also \P-complete like comparison with less than?		
	\item Is the word problem of the Baumslag group with power circuit representations as input \P-complete? (By \cref{cor:membershipBS} this holds for the subgroup membership problem for \BS12 in \BG. Moreover, as a consequence of \cref{prop:eqNLhard}, this variant of the word problem is \NL-hard.)
	\item By \cref{cor:compRedPCAC} for every $k$ the comparison problem for power circuits of depth $\log^kn$ is in \Tc{k} and hard for \Ac{k} under \Ac0-Turing reductions. Thus, the question remains whether, indeed, this problem is complete for $\Tc{k}$ under $\Ac0$-Turing-reductions.
\end{itemize}

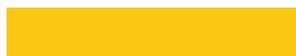
\end{document}